\documentclass[11pt]{article}
\pdfoutput=1

\usepackage[margin=0.75in]{geometry}

\usepackage[ascii]{inputenc}
\usepackage[bookmarksnumbered,hypertexnames=false,colorlinks=true,linkcolor=blue,urlcolor=blue,citecolor=blue,anchorcolor=green]{hyperref}
\usepackage{bbm,braket,microtype,mathrsfs,amsmath,amssymb,color,amsthm,mathtools,fullpage,graphicx,enumitem,ae,aecompl,float,mathabx,caption,bm,mathdots}
\usepackage[capitalize]{cleveref}
\usepackage[numbers,sort&compress]{natbib}
\usepackage[T1]{fontenc}
\usepackage[english]{babel}
\usepackage{newpxtext}
\usepackage{thmtools,thm-restate}
\usepackage{xpatch}
\DeclareMathOperator{\tr}{tr}
\DeclareMathOperator{\diag}{diag}
\DeclareMathOperator{\SL}{SL}
\DeclareMathOperator{\GL}{GL}
\DeclareMathOperator{\Ten}{Ten}

\DeclareMathOperator{\Herm}{Herm}
\DeclareMathOperator{\HWV}{HWV}
\DeclareMathOperator{\spec}{spec}
\DeclareMathOperator{\Det}{Det}
\DeclareMathOperator{\Sym}{Sym}
\DeclareMathOperator{\poly}{poly}
\DeclareMathOperator{\rk}{rank}

\DeclareMathOperator{\im}{im}
\DeclareMathOperator{\conv}{conv}
\DeclareMathOperator{\capacity}{capacity}
\newtheorem{thm}{Theorem}[section]\crefname{thm}{Theorem}{Theorems}
\newtheorem{lem}[thm]{Lemma}\crefname{lem}{Lemma}{Lemmas}
\newtheorem{prb}[thm]{Problem}\crefname{prb}{Problem}{Problems}
\newtheorem{rem}[thm]{Remark}\crefname{rem}{Remark}{Remarks}
\newtheorem{cor}[thm]{Corollary}\crefname{cor}{Corollary}{Corollaries}
\newtheorem{dfn}[thm]{Definition}\crefname{dfn}{Definition}{Definitions}
\newtheorem{prp}[thm]{Proposition}\crefname{prp}{Proposition}{Propositions}
\newtheorem{exa}[thm]{Example}\crefname{exa}{Example}{Examples}
\floatstyle{boxed}\newfloat{Algorithm}{ht}{alg}\crefname{Algorithm}{Algorithm}{Algorithms}

\floatstyle{boxed}\newfloat{ Algorithm}{ht}{algn}\crefname{algn}{Algorithm}{Algorithms}

\renewcommand{\vec}{\bm}
\newcommand{\CC}{\mathbb C}
\newcommand{\RR}{\mathbb R}
\newcommand{\PP}{\mathbb P}
\newcommand{\ZZ}{\mathbb Z}
\newcommand{\ot}{\otimes}

\newcommand{\eps}{\varepsilon}
\newcommand{\expon}{\exp}
\newcommand{\norm}[1]{\lVert#1\rVert}
\newcommand{\Norm}[1]{\left\lVert#1\right\rVert}

\newcommand{\cP}{\mathsf{P}}

\begin{document}

\title{Efficient algorithms for tensor scaling, quantum marginals, and moment polytopes}

\author{
Peter B\"{u}rgisser\thanks{Institut f\"{u}r Mathematik, Technische Universit\"{a}t Berlin, email: pbuerg@math.tu-berlin.de.}
\and
Cole Franks\thanks{Department of Mathematics, Rutgers University, email: wcf17@math.rutgers.edu.}
\and
Ankit Garg \thanks{Microsoft Research New England, email: garga@microsoft.com.}
\and
Rafael Oliveira\thanks{Department of Computer Science, University of Toronto, email: rafael@cs.toronto.edu.}
\and
Michael Walter\thanks{QuSoft, Korteweg-de Vries Institute for Mathematics, Institute of Physics, and Institute for Logic, Language and Computation, University of Amsterdam, email: m.walter@uva.nl.}
\and
Avi Wigderson \thanks{Institute for Advanced Study, Princeton, email: avi@ias.edu.}
}
\pagenumbering{gobble}
\maketitle
\begin{abstract}

We present a polynomial time algorithm to approximately {\em scale} tensors of any format to {\em arbitrary} prescribed marginals (whenever possible).
This unifies and generalizes a sequence of past works on matrix, operator and tensor scaling.
Our algorithm provides an efficient weak membership oracle for the associated {\em moment polytopes}, an important family of implicitly-defined convex polytopes with exponentially many facets and a wide range of applications.
These include the entanglement polytopes from quantum information theory (in particular, we obtain an efficient solution to the notorious one-body quantum marginal problem) and the Kronecker polytopes from representation theory (which capture the asymptotic support of Kronecker coefficients).
Our algorithm can be applied to succinct descriptions of the input tensor whenever the marginals can be efficiently computed, as in the important case of matrix product states or tensor-train decompositions, widely used in computational physics and numerical mathematics.

Beyond these applications, the algorithm enriches the arsenal of ``numerical'' methods for classical problems in invariant theory
that are significantly faster than ``symbolic'' methods which explicitly compute invariants or covariants of the relevant action.
We stress that (like almost all past algorithms) our convergence rate is polynomial in the approximation parameter; it is an intriguing question to achieve exponential convergence rate, beating symbolic algorithms exponentially, and providing {\em strong} membership and separation oracles for the problems above.

We strengthen and generalize the alternating minimization approach of previous papers by introducing the theory of highest weight vectors from representation theory into the numerical optimization framework.
We show that highest weight vectors are natural potential functions for scaling algorithms and prove new bounds on their evaluations to obtain polynomial-time convergence.
Our techniques are general and we believe that they will be instrumental to obtain efficient algorithms for moment polytopes beyond the ones consider here, and more broadly, for other optimization problems
possessing natural symmetries.
\end{abstract}

\newpage
\tableofcontents
\newpage
\pagenumbering{arabic}
\section{Introduction and summary of results}\label{sec:intro}

\subsection{Moment polytopes}\label{subsec:moment polytopes}

As this paper is quite technical, with some non-standard material for computer scientists, we begin
with motivating the main object we study, as it is extremely natural from an optimization perspective,
the {\em moment polytope}. Consider first the following diverse set of problems, trying to pick up
common features among them (besides the obvious guess that they all are special cases of the framework we
consider in this paper).

\begin{enumerate}
\item\label{it:schur} \textbf{The Schur-Horn Theorem:}  Can a given Hermitian matrix be conjugated by
unitary matrices to achieve a given diagonal?

\item\label{it:horn} \textbf{Eigenvalues of sums of Hermitian matrices:} Do there exist
Hermitian $n\times n$ matrices $A$, $B$, $C$ with prescribed eigenvalues such that $A+B=C$?

\item\label{it:opt}  \textbf{Optimization:} Can a given non-negative matrix be converted to another with prescribed
row and column sums, by only reweighing its rows and columns?

\item\label{it:qit} \textbf{Quantum information:} Can multiple parties, each holding a particle of a
pure quantum state, locally transform their particles so that each particle is maximally entangled with
the others?

\item\label{it:brascamp} \textbf{Analytic inequalities:} Given $m$ linear maps $A_i:\RR^n \rightarrow \RR^{n_i}$ and $p_1,\ldots, p_m \ge 0$, does there
exist a finite constant $C$ such that for all integrable functions $f_i: \RR^{n_i} \rightarrow \RR_+$ we have
$$\int_{x\in \RR^n} \prod_{i=1}^m  f_i(A_i x) dx \,\, \leq C\, \prod_{i=1}^m  \norm{f_i}_{1/p_i}?$$
An important special case of such framework\footnote{These inequalities
are the celebrated Brascamp-Lieb inequalities, which capture many more important inequalities
such as H\"{o}lder's, Loomis-Whitney, and many others. See for instance~\cite{garg2017algorithmic}
for a more detailed discussion.} is Cauchy-Schwarz, with $p_1 = p_2 = 1/2, m=2, n=n_1=n_2=1, C=1, A_i= 1$.

\item\label{it:pit} \textbf{Algebraic complexity:} Given an arithmetic formula (with inversion gates) in
non-commuting variables, is it non-zero?

\item\label{it:newton} \textbf{Polynomial support:} Given oracle access to a homogeneous polynomial $p$ with non-negative
integer coefficients on $n$ variables, is a specified monomial (given as integer vector of exponents)
in the Newton polytope\footnote{Given a polynomial $p(x_1, \ldots, x_n)$, define its support as the
set of monomials whose coefficient in $p$ is nonzero. The Newton polytope of $p(x_1, \ldots, x_n)$ is
given by the convex hull of the exponent vectors of these monomials.} of $p$?

\end{enumerate}

Some of the problems above are in $\cP$ and for others, there are sufficient hints that they are in $\cP$ (see \cite{Horn54, LSW, burgisser2017alternating, garg2017algorithmic, garg2016deterministic, Gurvits05}). While they may seem non-linear in their inputs, convexity plays an important role in each of them, as they all reduce to solving linear programs (implicitly defined with large number of facets).
More specifically, each input to each problem defines a
point and a polytope, and the answer is yes iff the point is in the polytope.
These polytopes turn out to be special cases of {\em moment polytopes}.

This appearance of linearity and convexity is quite surprising, in some settings more so than others.
Indeed, moment polytopes arise (and are used to understand problems) in many diverse settings such as
symplectic geometry, algebraic geometry, lattice theory and others \cite{GS84, Silva08}. The snag is that these polytopes are
often defined by  a huge number of inequalities (e.g. see \cite{garg2017algorithmic}); typically the number is exponential or larger in the dimension of
the input.\footnote{However, in many of these areas even finiteness provides progress, as even
decidability may not be obvious.}
This motivates our efforts to develop efficient algorithms for them.

In order to explain the appearance of convex polytopes in these settings, we need to notice
another common aspect of all problems above: their answers remain invariant under some group action!
This is easy to see in some of the examples, which explicitly specify the groups. In the first,
for matrices of size $n$, it is $U(n)$, the group of transformations conjugating the input.
In the second, each of the three matrices may be conjugated by a unitary. In the 3rd, it is the product $\text{T}(n) \times \text{T}(n)$ of two (positive) diagonal invertible matrices which scale
(resp.) the rows and columns. In the 4th problem, as each party is allowed to perform quantum
measurements with post selection, the group representing each party's operations is
$\GL(n)$ if its particle has $n$ states, and so the full group is a direct product of these $\GL(n)$'s.
The 5th problem is invariant to basis changes in the host space $\RR^n$ and the other $m$ spaces.\footnote{This reveals the Brascamp-Lieb polytopes \cite{BCCT, garg2017algorithmic} as special cases of (slices of) moment polytopes, which have an efficient weak separation oracle.}
The 6th is much harder to guess without Cohn's characterization of the free skew field, but turns out
to be $\GL(n) \times \GL(n)$ acting on a different representation of the formulas. In the 7th, though it may not seem
useful at first sight, $\text{T}(n)$ acts by simply scaling every variable of the polynomial
by a nonzero constant factor.

Having mentioned the two common features of the problems above (convexity and the invariance under
a group action) we will now illustrate how one can use the structure of the group action in order to
obtain moment polytopes. Let $G$ be a ``nice" \footnote{The technical definition requires the group to be algebraic, reductive and connected and so on. But for the purpose of this paper, one can think of groups like $\GL(n)$, $\text{T}(n)$, their direct products etc.} group acting linearly and continuously on a vector
space $V$ and $v$ be a point in $V$. The {\em orbit} of a point $v \in V$ is the set of all vectors obtained by
the action of $G$ on $v$. The {\em orbit closure} of $v$ is simply the closure of its orbit in the Euclidean topology. As the previous paragraph observed, all of the problems above are questions
about the orbit closures, which suggests understanding orbit closures is a
fundamental task with many applications. A natural approach to study such orbit closures is by looking
at the infinitesimal action of the group on every point $v$.

This brings us to the {\em moment map}, denoted by $\mu_G(v)$, which is essentially a gradient of the log of the norm of $v$ along the group
action.\footnote{Indeed, the original name
was {\em momentum map}, and is inspired from Hamiltonian physics, in which momentum is proportional to the derivative
of position. Apparently moment maps are common in physics, where they are used to obtain conserved quantities (i.e. invariants)
from symmetries of the phase space of symplectic manifolds describing some Hamiltonian system. In the general setting, we have the action of continuous group
on a manifold, and the moment map provides a reverse map, from the manifold to the group (or more
precisely, to the dual of the Lie algebra of the group). }
More explicitly, for each point $v$ we can define the function $f_v(g) = \log \norm{g \cdot v}_2^2$, and
$\mu_G(v)$ will be the gradient of $f_v(g)$ evaluated at the identity element of $G$.
The moment map carries a lot of information about the group action, and one of its striking features
is that the set of possible spectra of the image of any orbit closure under the moment map is a rational convex polytope  \cite{Kostant73, Atiyah82, GS82_convexity, kirwan1984convexity}! That is a mouthful. So consider an example to see what we mean. Consider the action of $G = \GL(n)$ on some vector space $V$. Then the moment map maps $V$ to $\text{M}(n)$ (set of all $n \times n$ matrices). Then the collection $\spec(\mu_G(v))$, as $v$ varies over an orbit-closure forms a rational convex polytope. Here $\spec(M)$ denotes the vector of eigenvalues of $M$ arranged in decreasing order. Note that $\mu_G(v)$ is a quadratic function of $v$, so the appearance of convexity is extremely surprising and non trivial.
This polytope, which we will more explicitly see in the next section, is the so called
{\em moment polytope} of the group action $G$ on the orbit of $v$.

In the matrix scaling case (Problem \ref{it:opt}), it turns out that the moment map applied to a certain matrix $A$
gives us precisely the marginals of $A$ (that is, the vector of row sums and column sums normalized to sum
$1$).\footnote{There is a slight technicality here and the moment map is actually the absolute values squared of
the entries of $A$.}
Thus, testing whether $A$ can be scaled to another matrix with prescribed row and column sums is
equivalent to testing whether the prescribed vector of row and column sums belongs to the moment polytope of the
action of $\text{T}(n)\times \text{T}(n)$ on $A$. Similarly, all of the seven problems listed above fit into this
framework (membership in moment polytope) for a suitable choice of group and representation.\footnote{For some
of the problems mentioned above, it is non-trivial to phrase them as moment polytopes.}

The reader might notice the dual nature of the problems above. They are both of algebraic as well as analytic nature. This phenomenon is extremely general and crucial for our paper. The analytic nature helps in designing algorithms, making the problem amenable to general optimization techniques, while the algebraic helps with analysis of these analytic algorithms and provides potential functions to track progress made by these algorithms. We will see that this will be the case for us as well.

\subsection{Our setting}

In this paper, we will be concerned with the moment polytopes of a natural ``basis" change group action on tensors, which are of interest for several reasons. The moment polytopes in this setting capture fundamental problems in quantum many-body physics - the so called {\em one-body quantum marginal problem}. They also capture fundamental problems in representation theory related to Kronecker coefficients, which are central objects of study in algebraic combinatorics and play an important role in geometric complexity theory. Moreover, as we will see, these moment polytopes generalize many of the settings described above and we believe that their complexity is representative of the complexity of general moment polytopes.

These moment polytopes (and their related problems) are most natural to state from the point of view of {\em quantum systems} and their {\em quantum marginals} \footnote{These generalize the classical notion of marginals of a probability distribution on several variables.}, so we start by defining them first.
But before we define quantum systems some brief notation must be established.

Let $\Ten(n_0;n_1,\dots,n_d)=\CC^{n_0}\ot\CC^{n_1}\ot\dots\ot\CC^{n_d}$ denote the space of $d+1$ dimensional
tensors of format $n_0\times n_1\times\dots\times n_d$, and let $X$ be a tensor in
$\Ten(n_0; n_1, \ldots, n_d)$. If we regard $X$ as a vector, with $X^\dagger$ being it's conjugate transpose,
then $\rho_X = XX^\dagger$ is a Hermitian positive semidefinite (PSD) operator on $\Ten(n_0;n_1,\dots,n_d)$. We will denote by $\norm{X} = \tr[\rho_X]^{1/2}$ the $\ell_2$ norm of $X$ (when viewed as a vector).
With this notation in mind, we then define a {\em quantum system} with $d+1$ subsystems as a PSD operator
on $\Ten(n_0;n_1,\dots,n_d)$ with unit trace  \footnote{A reader not familiar with the basics of quantum systems may want to skip a couple of paragraphs ahead.}.

Given a quantum system $\rho$ on~$\Ten(n_0;n_1,\dots,n_d)$ and a subset $I\subseteq\{0,1,\dots,d\}$, we
define its \emph{(quantum) marginals} or \emph{reduced density matrices} by~$\rho^{(I)} = \tr_{I^c}[\rho]$,
where $\tr_{I^c}$ denotes the partial trace over tensor factors~$I^c = \{0,\dots,d\} \setminus I$.
In the same way that $\rho$ describes the state of the entire quantum system, $\rho^{(I)}$ characterizes
the state of the subsystems labeled by~$I$ (in an analogous way to the classical marginal of a probability
distribution). For~$I=\{i\}$, we write~$\rho^{(i)}$; these operators are known as the
\emph{one-body marginals} or \emph{one-body reduced density matrices} of $\rho$. Each $\rho^{(i)}$ is
uniquely characterized by the property that
\begin{align}\label{eq:marginal}
  \tr[\rho^{(i)} A^{(i)}] =
  \tr[\rho ( I_{n_0} \ot I_{n_1}\ot\dots\ot I_{n_{i-1}}\ot A^{(i)}\ot I_{n_{i+1}}\ot\dots\ot I_{n_d})]
\end{align}
for all operators $A^{(i)}$ on $\CC^{n_i}$.

For a tensor $X \in$ $\Ten(n_0;n_1,\dots,n_d)$ and a given subset $I\subseteq\{0,1,\dots,d\}$ of the subsystems,
the marginals of~$\rho_X$ with respect to
$I$ have a particularly simple description: using the standard basis, identify~$X$ with
a matrix $M_X^{(I)} \in L(\CC^{n_I}, \CC^{n_{I^c}})$, where we denote $n_I := \prod_{i\in I} n_i$.
The matrix $M_X^{(I)}$ is known as a flattening, unfolding, or
matricization~\cite{landsberg2012tensors,hackbusch2012tensor} of the tensor $X$.
Then, $\rho_X^{(I)} = M_X^{(I)} (M_X^{(I)})^\dagger$ is its Gram matrix.

Given a Hermitian operator $\sigma$ on~$\CC^n$ (i.e., an $n \times n$ Hermitian matrix), we
write~$\spec(\sigma)=(s_1,\dots,s_n)$ for the vector of eigenvalues of $\sigma$, ordered non-increasingly.
If $\sigma$ is PSD with unit trace then its eigenvalues form a probability distribution, so $\spec(\sigma)$
is an element of $P_+(n) := \{ (s_1,\dots,s_n) : s_1\geq\dots\geq s_n\geq0 : \sum_j s_j=1 \}$.
We also abbreviate $P_+(n_1,\dots,n_d) := P_+(n_1) \times \dots \times P_+(n_d)$.
We will be particularly interested in characterizing the eigenvalues of the one-body marginals, motivated by
the following fundamental problem in quantum mechanics~\cite{klyachko2006quantum}:

\begin{prb}[One-body quantum marginal problem]\label{prb:qmp}
  Given $\vec p\in P_+(n_1,\dots,n_d)$, decide if there exists a tensor $Y\in\Ten(1;n_1,\dots,n_d)$ such that
  $\spec(\rho_Y^{(i)}) = \vec p^{(i)}$~for all $i=1,\dots,d$.
\end{prb}

\begin{rem} Note that the above problem is equivalent to the following, given density matrices $($PSD matrices
with unit trace $)$ $\rho^{(1)},\ldots, \rho^{(d)}$, determine if there exists a tensor (pure state)
$Y\in\Ten(1;n_1,\dots,n_d)$ such that $\rho_Y^{(i)} = \rho_i$~for all $i=1,\dots,d$. Since a unitary change
of basis comes for free on each subsystem, only the eigenvalues of $\rho^{(1)},\ldots, \rho^{(d)}$ are relevant.
\end{rem}

The above problem is extremely fundamental from the point of view of quantum many-body physics. It is a special case of the more general quantum marginal problem, which puts constraints on the marginals of multiple systems and is known to be QMA-complete (for growing $d$) \cite{Liu06}.

We note that the normalization to trace one is natural; since $\tr[\rho_Y^{(i)}]=\lVert Y\rVert^2$ for all
$i$, we can simultaneously rescale all marginals simply by rescaling the tensor.

Now we discuss how \cref{prb:qmp} can be phrased as a question about moment polytopes~\cite{ness1984stratification,brion1987image,klyachko2006quantum,ressayre2010geometric,vergne2014inequalities,walter2014multipartite}.
Let $G=\GL(n_1)\times\dots\times\GL(n_d)$, where $\GL(n)$ denotes the group of invertible $n\times n$-matrices.
Then $G$ acts on~$V=\Ten(n_0;n_1,\dots,n_d)$ by
\begin{align*}
  (g^{(1)},\dots,g^{(d)}) \cdot X := (I_{n_0}\ot g^{(1)}\ot\dots\ot g^{(d)}) X.
\end{align*}
As the group acts by rescaling slices of the tensor, we will call any $Y\in G\cdot X$ a \emph{tensor scaling} of
$X$.\footnote{The extra coordinate with dimension $n_0$ can be equivalently thought of as enumerating an
$n_0$-tuple of tensors in $\Ten(n_1,\dots,n_d)$ and the group $G$ acts simultaneously on all the tensors in
the tuple. Much of the theory remains similar if one sets $n_0 = 1$ and that can be done mentally on a
first reading. In the quantum language, it is the difference between acting on pure states $(n_0 = 1)$ vs acting
on mixed states $(n_0 > 1)$.}

What is the moment map in this setting? It turns out that the moment captures exactly the notion of one-body
quantum marginals. It is more convenient to define the moment map on the projective space (namely restrict
ourselves to tensors of unit norm), since we don't care about the scalar multiples. We will denote the projective
space corresponding to $V$ by $\PP(V)$ and identify it with the set of rank-one trace-one
PSD operators on $V$, $\PP(V) = \{ \rho = [X] = XX^\dagger/X^\dagger X : 0\neq X\in V\}$. Then the moment map can
be written as\footnote{After identifying the Lie algebra of~$K$ with its dual.}
\begin{align}\label{eq:moment map}
  \mu\colon\PP(V)\to\Herm(n_1)\times\dots\times\Herm(n_d), \quad
  \rho \mapsto (\rho^{(1)},\dots,\rho^{(d)}),
\end{align}
where~$\Herm(n)$ denotes the space of Hermitian $n\times n$-matrices.
Now consider a projective subvariety $\mathcal X$ of $\PP(V)$ such as $\mathcal X = \PP(V)$ or an orbit-closure\footnote{Here, the closure can be taken either in the Euclidean or in the Zariski topology.} i.e. $\mathcal X = \overline{G \cdot [X]}$ for some given tensor $X\in V$.\footnote{In general $\mathcal X$ can be any $G$-stable irreducible projective subvariety of $\PP(V)$.} Let us look at the collection of marginal eigenvalues when restricted to tensors in~$\mathcal X$:
\begin{align}\label{eq:moment polytope}
  \Delta(\mathcal X) := \{(\spec(\rho^{(1)}),\dots,\spec(\rho^{(d)})) : \rho\in\mathcal X\} \subseteq P_+(n_1,\dots,n_d).
\end{align}
We emphasize again the amazing, surprising and non-trivial fact that $\Delta(\mathcal X)$ is a rational convex polytope~\cite{ness1984stratification,kirwan1984cohomology,kirwan1984convexity,brion1987image} -- known as the \emph{moment polytope} or \emph{Kirwan polytope} of $\mathcal X$.\footnote{Note that we have identified $\PP(V)$ with the set of rank $1$ density matrices and hence it is far from being a convex set - yet the spectrum of its image under the moment map is convex.}
This means that $\Delta(\mathcal X)$ can in principle be given in terms of finitely many affine inequalities in eigenvalues of the one-body marginals~\cite{klyachko2006quantum,ressayre2010geometric,vergne2014inequalities}.
In particular, the preceding applies to $\mathcal X=\PP(V)$, so we can rephrase \cref{prb:qmp} as follows:
\emph{Given~$\vec p \in P_+(n_1,\dots,n_d)$, is it a point in~$\Delta(1;n_1,\dots,n_d):=\Delta(\PP(V))$?}
More generally, we can consider the following decision problem:

\begin{prb}[General moment polytope]\label{prb:general}
  Given $\vec p\in P_+(n_1,\dots,n_d)$, decide if there exists a tensor $[Y]\in\mathcal X$ such that $\spec(\rho_Y^{(i)}) = \vec p^{(i)}$~for all $i=1,\dots,d$.
\end{prb}

When $\mathcal X = \overline{G \cdot [X]}$ is the \emph{orbit closure} of some given tensor $X\in V$, we will abbreviate the moment polytope by $\Delta(X) := \Delta(\overline{G \cdot[X]})$.
In quantum information theory, moment polytopes of orbit closures have been called \emph{entanglement polytopes} as they characterize the multipartite entanglement from the perspective of the one-body marginals~\cite{walter2013entanglement,sawicki2014convexity}.
But, along with the corresponding invariant-theoretic multiplicities, they are also of interest in algebraic and geometric complexity theory~\cite{burgisser2011overview,burgisser2011geometric,christandl2012computing,christandl2017universal}.
The corresponding decision problem is the following:

\begin{prb}[Moment polytope of orbit closure]\label{prb:orbit closure}
  Given $X\in\Ten(n_0;n_1,\dots,n_d)$ and $\vec p\in P_+(n_1,\dots,n_d)$, decide if there exists $Y\in\overline{G\cdot X}$ such that $\spec(\rho_Y^{(i)}) = \vec p^{(i)}$~for all $i=1,\dots,d$.
\end{prb}

That is, \cref{prb:orbit closure} asks whether $\vec p=(\vec p^{(1)},\dots,\vec p^{(d)})$ is a point in~$\Delta(X)$.\footnote{When $n_0=1$, there is a  physical interpretation of the orbit-closure. $Y\in\overline{G\cdot X}$ means that $Y$ can be obtained to arbitrary precision from $X$ (which is naturally understood as a $d$-partite quantum state) by a class of quantum operations known as \emph{stochastic local operations and classical communication (SLOCC)}~\cite{bennett2000exact}.
SLOCC can be intuitively understood as follows: we imagine that different parties hold the different systems of a quantum state; SLOCC then corresponds to a sequence of local quantum operations and measurements, where we allow for \emph{post-selection} on specific measurement outcomes. \cref{prb:orbit closure} then asks if given a tensor $X \in \Ten(n_1,\dots,n_d)$, does there exist a $Y$ obtainable by a sequence of SLOCC operations from $X$ s.t. $\spec(\rho_Y^{(i)}) = \vec p^{(i)}$ for all $i$. This is a generalization of the SLOCC entanglement distillation question where $\vec p^{(i)}$ is the uniform distribution for all $i$.}

One can show that \cref{prb:orbit closure} is intimately related to \cref{prb:qmp,prb:general}:
$\vec p\in\Delta(\mathcal X)$ iff $\vec p\in\Delta(X)$ for a generic $X\in\mathcal X$ (\cref{cor:generic is enough}).
We will explain this in \cref{sec:git}.
We will therefore focus our attention on \cref{prb:orbit closure}.

It is natural to go beyond the decision problem and look for an algorithm that finds a tensor~$Y$ with the desired marginals, as well as the group element that transforms~$X$ into~$Y$.
Since such an $Y$ will be in the orbit through~$X$, we demand only that the marginals are correct up to some target accuracy.

\begin{dfn}[$\eps$-close]\label{dfn:eps-close}
The marginals of $Y\in\Ten(n_0;n_1,\dots,n_d)$ are said to be \emph{$\eps$-close} to $\vec p\in P_+(n_1,\dots,n_d)$ if
$\lVert \spec(\rho_Y^{(i)}) - \vec p^{(i)} \rVert_1 \leq \eps$ for $i=1,\dots,d$.
Here, $\lVert \vec x\rVert_1 = \sum_j \lvert x_j\rvert$ is the $\ell^1$-norm.
\end{dfn}

\begin{prb}[Tensor scaling]\label{prb:scaling}
  Given $X\in\Ten(n_0;n_1,\dots,n_d)$, $\vec p \in \Delta(X)$, and $\eps>0$, find $g_{\eps}\in G$ such that $Y=g_{\eps}\cdot X$ has marginals that are $\eps$-close to~$\vec p$.
\end{prb}

While it may not be immediately clear, there exist scalings as in \cref{prb:scaling} for any $\eps>0$ if and only if the answer to \cref{prb:orbit closure} is yes, i.e., if and only if $\vec p\in\Delta(X)$.

The polytopes $\Delta(\mathcal X)$ admit alternative characterization in terms of invariant theory~\cite{ness1984stratification}.
We explain this connection in \cref{sec:git}, as it is central to the analysis of our algorithms.
For now, we only mention an important special case.
Let $g(\vec\lambda,\vec\mu,\vec\nu)$ denote the \emph{Kronecker coefficients}, which are fundamental objects in the classical representation theory of the symmetric and general linear groups~\cite{fulton2013representation,stanley2000positivity}.
They also feature in geometric complexity theory as a potential way of creating representation theoretic obstructions~\cite{mulmuley2007geometric,burgisser2011overview}.
For example, $g(\vec\lambda,\vec\mu,\vec\nu)$ can be defined as the multiplicity of the irreducible $S_k$-representation $[\vec\lambda]$ in the tensor product $[\vec\mu]\ot[\vec\nu]$.
Here, $\vec\lambda$, $\vec\mu$, and $\vec\nu$ are partitions, which we may think of nonincreasing vectors in $\ZZ^n_{\geq0}$ with $\sum_j \lambda_j = \sum_j \mu_j = \sum_j \nu_j =: k$.
Then,
\begin{align}\label{eq:kronecker equivalence}
  \exists \: \text{integer} \: s\ge 1: g(s\vec\lambda,s\vec\mu,s\vec\nu) > 0
\quad\Leftrightarrow\quad
  \frac1k(\vec\lambda,\vec\mu,\vec\nu) \in \Delta(1;n,n,n),
\end{align}
so the solution to the one-body quantum marginal problem captures precisely the \emph{asymptotic support} of the Kronecker coefficients~\cite{christandl2006spectra,klyachko2006quantum,christandl2007nonzero}. We note that the problem of deciding whether $g(\vec\lambda,\vec\mu,\vec\nu) > 0$ is known to be NP-hard \cite{ikenmeyer2017vanishing}. However since the asymptotic vanishing of Kronecker coefficients is captured by the quantum marginal problem, it has been conjectured that it should have a polynomial time algorithm and we make progress towards this question.\footnote{We note that the closely related Littlewood-Richardson coefficients (which capture the same problem for the representations of the general linear group) satisfy the so called saturation property: $c(\vec\lambda,\vec\mu,\vec\nu) > 0$ iff $c(s\vec\lambda,s\vec\mu,s\vec\nu) > 0$ \cite{KT99}. Hence the asymptotic support is the same as support for this case and this is also a key ingredient in the polynomial time algorithms for testing if $c(\vec\lambda,\vec\mu,\vec\nu) > 0$ \cite{BI12}.} Since Kronecker coefficients are so poorly understood, understanding their asymptotic support would also go a long way in understanding them.

\subsection{Prior work}\label{subsec:prior}

As mentioned above, \cref{prb:qmp} can be approached by first computing (the defining inequalities of) the moment polytope~$\Delta(n_0;n_1,\dots,n_d)$.
The problem of computing moment polytopes has a long history in mathematics (e.g., \cite{atiyah1982convexity,guillemin1982convexity,guillemin1984convexity,kirwan1984cohomology,kirwan1984convexity,brion1999general,berenstein2000coadjoint,ressayre2010geometric,belkale2010tangent}).
That the one-body quantum marginal problem falls into this framework was first noticed by Klyachko~\cite{klyachko2004quantum}, who gave a complete description of the polytopes in terms of finite lists of inequalities (cf.\ \cite{daftuar2005quantum,klyachko2006quantum,altunbulak2008pauli}).
Before that, only low-dimensional special cases were known~\cite{higuchi2003one,bravyi2003requirements,franz2002moment}.
Further developments include the minimal complete description from~\cite{ressayre2010geometric} and the cohomology-free variant~\cite{vergne2014inequalities}.
Yet, all these descriptions in terms of inequalities are largely computationally infeasible;
explicit descriptions are known up to formats $3\times3\times9$~\cite{klyachko2006quantum} and $4\times4\times4$~\cite{vergne2014inequalities}, and when all dimensions are two~\cite{higuchi2003one}.

\Cref{prb:orbit closure,prb:general} can in principle be approached using classical computational invariant theory (e.g., \cite{sturmfels2008algorithms,derksen2015computational,walter2013entanglement}), based on the invariant-theoretic description of~$\Delta(\mathcal X)$ and degree bounds (we recall both in~\cref{sec:git}).
In practice, however, this is completely infeasible except for very small dimensions.
The problem of describing~$\Delta(\mathcal X)$ also falls into the framework of~\cite{ressayre2010geometric}, but it is not clear how to turn this into an algorithm. In summary, all the methods described above are computationally expensive and take time at least exponential in the input size.

None of the preceding algebraic methods can be used to solve \cref{prb:scaling}, since they only decide membership but do \emph{not} produce the transformation that produces a tensor with the desired target spectra.
This calls for the development of numerical algorithms for \cref{prb:scaling}. Curiously this development came stemmed from motivations in algebraic complexity and the PIT problem. The first such algorithm was proposed in \cite{gurvits2004classical}. Its complexity analysis, that brought on the connection to invariant theory (and other fields, some mentioned above) was achieved in \cite{garg2016deterministic}. In the language we use here, it deals with $d=2$ (operator scaling) and uniform marginals, and results in polynomial time algorithms for problems in diverse areas discussed there.\footnote{The underlying algebraic problem associated with operator scaling, namely
non-commutative singularity and rank of symbolic matrices found a different, algebraic algorithm in the works
of~\cite{ivanyos2017constructive, derksen2015}} The operator scaling problem was then extended in two directions, which we mention next:
one direction being general values of $d$ (tensor scaling) and the other being $d=2$ and arbitrary marginals.

For general~$d$, a deterministic algorithm was given in~\cite{burgisser2017alternating} (based on a proposal in~\cite{verstraete2003normal} for $n_0=1$). Very recently, a randomized polynomial time algorithm for operator scaling to general marginals was given in~\cite{franks2018operator}.
 The two papers \cite{burgisser2017alternating, franks2018operator} study two different generalizations of the operator scaling problem in \cite{garg2016deterministic}. The present paper completes a natural square by studying a common generalization of the problems studied in \cite{burgisser2017alternating, franks2018operator}. All these algorithms can be seen as noncommutative generalizations of the Sinkhorn-Knopp algorithm for `matrix scaling`~\cite{Sink, sinkhorn1967concerning}.

It was shown recently known that \cref{prb:qmp} is in NP$\cap$coNP~\cite{burgisser2017membership}.
In view of \cref{eq:kronecker equivalence}, this should be contrasted with the NP-hardness of deciding whether a single Kronecker coefficient is zero or not~\cite{ikenmeyer2017vanishing}.

\subsection{Summary of results}\label{subsec:results}
Our main result in this paper is a randomized algorithm for tensor scaling to general marginals (\cref{prb:scaling}). As a consequence, we obtain algorithms for all other problems.

\begin{thm}\label{thm:main} There is a randomized algorithm running in time $\poly(N, 1/\eps)$, that takes as input $X\in\Ten(n_0;n_1,\dots,n_d)$ with Gaussian integer entries (specified as a list of real and complex parts, each encoded in binary, with bit size $\le b$) and
$\vec p\in P_+(n_1,\dots,n_d)$ with rational entries (specified as a list of numerators and denominators, each encoded in binary, with bit size $\le b$). The algorithm either correctly identifies that $\vec p \notin \Delta(X)$, or it outputs a scaling $g \in G$ such that the marginals of $g \cdot X$ are $\eps$-close to the target spectra~$\vec p$. Here $N$ is the total bit-size of the input, $N = 2 n_0 n_1 \cdots n_d b + 2 (n_1 + \cdots n_d) b$.
\end{thm}

As a consequence of \cref{thm:main}, we obtain a randomized algorithm for a promise version of the membership~\cref{prb:orbit closure} (and hence for \cref{prb:qmp}, see \cref{cor:generic is enough}).

\begin{cor}\label{cor:promise_mem_orbit_closure} There is a randomized algorithm running in time $\poly(N, 1/\eps)$, that takes as input $X\in\Ten(n_0;n_1,\dots,n_d)$ with Gaussian integer entries (specified as a list of real and complex parts, each encoded in binary, with bit size $\le b$) and
$\vec p\in P_+(n_1,\dots,n_d)$ with rational entries (specified as a list of numerators and denominators, each encoded in binary, with bit size $\le b$). The algorithm distinguishes between the following two cases:
\begin{enumerate}
\item $\vec p \in \Delta(X)$.
\item $\vec p$ is $\eps$-far (in $\ell_1$ norm) from any point $\vec q \in \Delta(X)$.
\end{enumerate}
 Here $N$ is the total bit-size of the input, $N = 2 n_0 n_1 \cdots n_d b + 2 (n_1 + \cdots n_d) b$.
\end{cor}

This yields the following corollary.

\begin{cor}\label{cor:promise_qmp}  There is a randomized algorithm running in time $\poly(n_0 n_1 \cdots n_d, b, 1/\eps)$, that takes as input $\vec p\in P_+(n_1,\dots,n_d)$ with rational entries (specified as a list of numerators and denominators, each encoded in binary, with bit size $\le b$). The algorithm distinguishes between the following two cases:
\begin{enumerate}
\item $\vec p \in \Delta(1; n_1,\ldots, n_d)$ i.e. there exists $Y\in\Ten(n_1,\ldots, n_d)$ such that $\spec(\rho_Y^{(i)}) = \vec p^{(i)}$ for all $i$.
\item $\vec p$ is $\eps$-far (in $\ell_1$ norm) from any point $\vec q \in \Delta(1; n_1,\ldots, n_d)$.
\end{enumerate}
\end{cor}

As described before, \cref{prb:qmp} captures the asymptotic vanishing of Kronecker coefficients. Hence we get the following corollary which describes a randomized polynomial time algorithm for a promise version of the asymptotic Kronecker problem.

\begin{cor}\label{cor:promise_Kronecker} There is a randomized algorithm running in time $\poly(n, b, 1/\eps)$, that takes as input three partitions $\vec \lambda, \vec \mu, \vec \nu \in \ZZ^n_{\geq0}$ with entries described in binary with bit-size at most $b$. The algorithm distinguishes between the following two cases:
\begin{enumerate}
\item  There exists and integer $s\ge 1$ s.t. $g(s\vec\lambda,s\vec\mu,s\vec\nu) > 0$.
\item For all $\vec \lambda', \vec \mu', \vec \nu'$ s.t. $g\left(\vec \lambda', \vec \mu', \vec \nu'\right) > 0$, it holds that $\left( \vec \lambda', \vec \mu', \vec \nu'\right)/\left| \left(\vec \lambda', \vec \mu', \vec \nu'\right) \right|$ is $\eps$-far (in $\ell_1$-norm) from $(\vec \lambda, \vec \mu, \vec \nu)/|(\vec \lambda, \vec \mu, \vec \nu)|$.
\end{enumerate}
Here $g$ denotes the Kronecker coefficient and $ |(\vec \lambda, \vec \mu, \vec \nu)| = \sum_j \lambda_j = \sum_j \mu_j = \sum_j \nu_j$.
\end{cor}

In many applications, the tensor~$X$ can be more succinctly represented than by its $n_0 n_1\cdots n_d$ many coordinates.
If the representation is preserved by scalings and allows for efficient computation of the marginals, then this yields a useful optimization of \cref{alg:scaling}.
A prime example of which are the so called \emph{matrix-product states} or \emph{tensor-train decompositions} with polynomial bond dimension~\cite{verstraete2003normal,orus2014practical}.
We won't define these states here (see \cref{sec:extensions} for a formal definition) but we will just say that these have much smaller (exponentially smaller in $d$) descriptions than specifying all the $n_0 n_1 \cdots n_d$ coordinates of the tensors.
This class includes the \emph{unit tensors} and the \emph{matrix multiplication tensors}, which are central objects in algebraic complexity theory~\cite{burgisser2011geometric,burgisser2013algebraic} and whose moment polytopes are \emph{not} known!

\begin{thm}[Informal]\label{thm:succinct} There is a randomized algorithm running in time $\poly(N, b, 1/\eps)$, that takes as input a matrix-product state $X\in\Ten(n_0;n_1,\dots,n_d)$ with input size $N$ and
$\vec p\in P_+(n_1,\dots,n_d)$ with rational entries (specified as a list of numerators and denominators, each encoded in binary, with bit size $\le b$). The algorithm either correctly identifies that $\vec p \notin \Delta(X)$, or it outputs a scaling $g \in G$ such that the marginals of $g \cdot X$ are $\eps$-close to the target spectra~$\vec p$.
\end{thm}

It is a very exciting open problem to improve the running time dependence on $\eps$ in~\cref{cor:promise_mem_orbit_closure}, \cref{cor:promise_qmp} and \cref{cor:promise_Kronecker} to $\poly(\log 1/\eps)$. This would yield randomized polynomial time algorithms for~\cref{prb:qmp}, \cref{prb:orbit closure} and the asymptotic Kronecker problem due to the following theorem that we prove in \cref{sec:distance}.

\begin{restatable}[Minimal gap]{thm}{minimalgap}\label{thm:minimal gap}
  Let $X\in\Ten(n_0;n_1,\dots,n_d)$ be a nonzero tensor.
  If $[Y]\in\overline{G \cdot [X]}$ is a scaling with marginals that are $\gamma(n_1,\dots,n_d,\ell)$-close to $\vec p$, then $\vec p\in\Delta(\mathcal X)$. Here $\gamma(n_1,\dots,n_d,\ell) = \expon\left( -O\left((n_1 + \cdots + n_d) \log(\ell \max_j n_j)\right) \right)$ and $\ell$ is the minimal integer s.t. $\ell \vec p$ has integral entries.
\end{restatable}

An analogous result for the full moment polytope $\Delta(\PP(V))$ was proven in \cite{burgisser2017membership}. We believe that the inverse exponential bound in the above theorem cannot be improved to an inverse polynomial bound. Therefore developing scaling algorithms with runtime dependence $\poly(\log 1/\eps)$ is of paramount importance.

Before describing our algorithm and high level intuition for its analysis, let us describe the algorithm and analysis for a rather special case of matrix scaling, which turns out to very enlightening.

\subsection{Simple example: matrix scaling}\label{sec:matrix_scaling}

The matrix scaling problem (Problem \ref{it:opt} in \cref{subsec:moment polytopes}) provides us with a template for what is to come, and understanding
the evolution of a particular algorithm for this problem will give us intuition on how to solve the
more general tensor scaling problem, and how invariant theory naturally appears.

If one wants to scale a given $n \times n$ matrix $A$ to a doubly stochastic matrix (that is, one whose rows and
columns each sum to $1$), a natural algorithm (first proposed in \cite{Sink}) arises from the fact that the group is a Cartesian product. We can alternately use scalings of the form $(R, I) \in \text{T}(n) \times \text{T}(n)$ to normalize the row sums
of $A$ and scalings of the form $(I, C) \in \text{T}(n) \times \text{T}(n)$ to normalize the column sums of $A$.

To this end, set $R(A)$ to be a diagonal matrix having $R(A)_{i,i}$
to be the inverse of the sum of the elements of the $i^{\text{th}}$ row of $A$, and define $C(A)$ in a similar
way for the columns of $A$. The algorithm can be described as follows:
repeatedly (for a polynomial number of iterations) apply the following steps:
\begin{itemize}
	\item Normalize the rows of $A$. That is, $A \leftarrow R(A) \cdot A$
	\item Normalize the columns of $A$. That is, $A \leftarrow A \cdot C(A)$.
\end{itemize}
If, throughout this process, matrix $A$ never gets sufficiently close to a doubly stochastic matrix
(in $\ell_2$ distance), then we will conclude that $A$ cannot be scaled to doubly stochastic; otherwise we
can conclude that $A$ can be scaled to doubly stochastic. The process also gives us a way to \emph{obtain} the scalings that approach doubly stochastic - while there are multiple algorithms for the decision problem (which turns out to be the bipartite perfect matching problem), not all help find the scalings!

The analysis of this algorithm (from \cite{LSW}; also see \cite{GurYianilos} for a different potential function) is extremely simple, and follows a three step approach based on
a progress measure $P(A) = \text{Perm}(A)$.\\

The following two properties of the potential function will be useful for us.
\begin{enumerate}
\item If $A$ is scalable to doubly stochastic, then $P(A) > 0$.
\item $P(A) \le 1$ if $A$ row or column normalized.
\end{enumerate}
The three step approach then is the following:
\begin{enumerate}
\item{} [\textbf{Lower bound}]: Initially $P(A) > 2^{-\poly(n)}$ (wlog we assume $A$ is row normalized) \footnote{There is some dependency on the bit complexity of the input that we are ignoring.}.
\item{} [\textbf{Progress per step}]: If $A$ is row or column normalized and sufficiently far from being doubly stochastic, then normalizing $A$ increases $P(A)$. One can explicitly bound the increase using a robust version of the AM-GM inequality.
\item{} [\textbf{Upper bound}]: $P(A)$ is bounded by $1$ if $A$ is row or column normalized.
\end{enumerate}

This three-step analysis shows that the scaling algorithm is able to solve the doubly stochastic scaling
problem in polynomial time. The difficult part of the analysis is coming up with
a potential function satisfying the properties above. This is the role played by invariant theory
later. A source of good potential functions will turn out to be {\em highest weight vectors},
which are (informally speaking) ``eigenvectors'' of the action of certain subgroups of the main group action.
Note that the permanent is an eigenvector of the action of
$\text{T}(n) \times \text{T}(n)$ since $\text{Perm}(RXC)$ equals $\left(\prod_i R_{i,i} \cdot \prod_j C_{j,j} \right) \cdot \text{Perm}(X)$ for $(R,C) \in \text{T}(n) \times \text{T}(n)$.

If we want to solve the more general scaling problem, where we are given
a prescribed value for the row and column sums, say as an non-negative integer vector
$(r,c) = (r_1, \ldots r_n, c_1, \ldots, c_n)$, the same natural algorithm can be applied. The only change
one needs to make in the algorithm above is that we will now normalize the rows of $A$ to have sums
$(r_1, \ldots, r_n)$ and the columns to have sum $(c_1, \ldots, c_n)$. The analysis is also quite similar: one can choose the
potential function, for example, to be the permanent of matrix $B$ obtained from $A$ by repeating $i^{\text{th}}$ row $r_i$ times and
$j^{\text{th}}$ column $c_j$ times. However, the distinction between the uniform and the non-uniform versions of the problems is much starker in our higher dimensional non-commutative setting, as we will see next.

\subsection{Techniques and proof overview}

Our algorithm and its analysis generalize two recent works \cite{burgisser2017alternating, franks2018operator}, which in turn generalize the analysis of matrix scaling in \cref{sec:matrix_scaling}. The paper \cite{burgisser2017alternating} studies the special case when $\vec p^{(i)}$ is the uniform distribution (over a set of size $n_i$) for all $i$ while the paper \cite{franks2018operator} studies the special case $d=2$. Our algorithm is a natural common generalization of the algorithms in \cite{franks2018operator, burgisser2017alternating} while our analysis generalizes the analysis in \cite{burgisser2017alternating} replacing the use of invariants with highest weight vectors (we will explain what these are later).

Let us develop some intuition for the algorithm. It is usually the case with scaling problems, as we saw with matrix scaling, and more generally in the framework of alternating minimization, that one of the constraints is easy to satisfy by scaling. The same is true for the problem we have at hand. We are given a tensor $X\in\Ten(n_0;n_1,\dots,n_d)$. Suppose we want $\spec\left(\rho_X^{(i)}\right) = \vec p^{(i)}$. With the shorthand $\vec p_\uparrow^{(i)}:=(p_{n_i}^{(i)}, \dots, p_1^{(i)})$, we act on $X$ by $g = \left(I,I,\ldots, \diag\left(\vec p_\uparrow^{(i)}\right)^{1/2} \left(\rho_X^{(i)}\right)^{-1/2} ,\ldots, I\right)$, where the non-trivial element is in the $i^{\text{th}}$ location. This is will satisfy the $i^{\text{th}}$ constraint. Or indeed, one can choose any matrix $R$ s.t. $R R^{\dagger} = \rho_X^{(i)}$ and act on $X$ by $g = \left(I,I,\ldots, \diag\left(\vec p_\uparrow^{(i)}\right)^{1/2} R^{-1} ,\ldots, I\right)$. This will also satisfy the $i^{\text{th}}$ constraint. By choosing each time to fix the index which is ``farthest'' from its target spectrum, we have defined an iterative algorithm (up to the choice of $R$ at each step) that keeps on alternately fixing the constraints. It turns out that this algorithm works (for any choice of $R$ at each step!) when $\vec p^{(i)}$'s are all uniform and converges in a polynomial number of iterations \cite{burgisser2017alternating}.

Interestingly, the choice of $R$ that works for general $\vec p^{(i)}$'s is that of upper triangular matrices!\footnote{This choice works for all $\vec p^{(i)}$'s. We don't know if this choice of upper-triangularity is necessary. There is also a nice interpolation between the case of uniform $\vec p^{(i)}$'s and $\vec p^{(i)}$'s with distinct entries. See \cref{subsec:parabolic scaling}.} This was the choice made in \cite{franks2018operator} as well. This restriction on scaling factors will make the analysis more complicated as we shall soon see. One intuitive reason for the difference between the uniform and the general case is the following: in the general case, we made an arbitrary decision to try to scale $X$ to have marginals $\diag\left(\vec p^{(i)}\right)$ while we could have chosen to scale it to any $\rho^{(i)}$ s.t. $\spec (\rho^{(i)}) = \diag\left(\vec p^{(i)}\right)$. This choice of basis is not present in the uniform case since all bases are the same!

This restriction on scaling factors creates another problem: it disconnects the orbit space (see example below). Thus, we need to initialize the algorithm with a random basis change of the given input, and only then resume the restricted scaling. This idea is used as well in \cite{franks2018operator}. We explain, via an example, why this random basis change (or at least a ``clever'' basis change) is needed at the start of the algorithm. Consider the diagonal unit tensor $X \in \Ten(1;2,2,2)$, where $X_{j,k,\ell} = 1$ iff $j = k = \ell$. It is easy to see that without the initial randomization, the algorithm (which chooses an upper triangular $R$ at each step) would only produce diagonal tensors $Y$ ($Y_{j,k,\ell} \neq 0$ iff $j = k = \ell$). And the marginals of any such tensor are isospectral. On the other hand, the $G$-orbit of~$X$ is dense in $\Ten(1;2,2,2)$ and so $\Delta(X)=\Delta(1;2,2,2)$. In particular, $X$ can be scaled to tensors with non-isospectral marginals.

The algorithm is described as \cref{alg:scaling}. The following is the main theorem regarding the analysis of \cref{alg:scaling} from which \cref{thm:main} follows up to an analysis of the bit complexity of \cref{alg:scaling}.
\begin{Algorithm}[th!]
\textbf{Input}:
$X\in\Ten(n_0;n_1,\dots,n_d)$ with Gaussian integer entries (specified as a list of real and complex parts, each encoded in binary, with bit size $\le b$) and
$\vec p\in P_+(n_1,\dots,n_d)$ with rational entries (specified as a list of numerators and denominators, each encoded in binary, with bit size $\le b$) such that $p^{(i)}_{n_i}>0$~for all $i=1,\dots,d$.
\\[.1ex]

\textbf{Output:}
Either the algorithm correctly identifies that $\vec p \notin \Delta(X)$, or it outputs $g\in G$ such that the marginals of $Y := g \cdot X$ satisfy \cref{eq:final marginals}; in particular the marginals are $\eps$-close to the target spectra~$\vec p$.
\\[.1ex]

\textbf{Algorithm:}\vspace{-.2cm}
\begin{enumerate}
\item \label{it:randomize}
Let $\ell>0$ be the least integer such that $\ell \vec p^{(i)}$ has integer entries~for all $i=1,\dots,d$ (i.e. $\ell$ is the common denominator of all $p^{(i)}_j$).
Let $g = (g^{(1)},\dots,g^{(d)})$ denote the tuple of matrices ($g^{(i)}$ is $n_i \times n_i$) whose entries are chosen independently uniformly at random from $\{1,\dots,M\}$, where
$M := 2 d K$ and $K := (\ell d \textstyle\max_{i=1}^d n_i )^{d \max_{i=1}^d n_i^2}$.
\item\label{it:full rank}
For $i=1,\dots,d$, if the marginal~$\rho_{g \cdot X}^{(i)}$ is singular then output $\vec p\not\in\Delta(X)$ and halt. \\
Otherwise, update $g^{(1)} \leftarrow g^{(1)}/\Norm{g \cdot X}$.
\item\label{it:borel scale}
For $t=1,\dots,T := \bigg\lceil\textstyle\frac{32\ln2} {\eps^2} \left( 3\sum_{i=0}^d \log_2(n_i) + b + d\log_2(M) \right)\bigg\rceil$, repeat the following:
\begin{itemize}
\item Compute $Y := g \cdot X$ and, for $i=1,\dots,d$, the one-body marginals~$\rho_Y^{(i)}$ and the distances~$\eps^{(i)} := \lVert \rho_Y^{(i)} - \diag(\vec p_\uparrow^{(i)}) \rVert_{\tr}$.
\item Select an index $i\in\{1,\dots,d\}$ for which~$\eps^{(i)}$ is largest. If~$\eps^{(i)}\leq\eps$, output $g$ and halt.
\item Compute the Cholesky decomposition\footnotemark $\rho_Y^{(i)} = R^{(i)} (R^{(i)})^\dagger$, where $R^{(i)}$ is an upper-triangular matrix.
Update $g^{(i)} \leftarrow \diag(\vec p_\uparrow^{(i)})^{1/2} (R^{(i)})^{-1} g^{(i)}$.
\end{itemize}
\item\label{it:give up} Output $\vec p\not\in\Delta(X)$.
\end{enumerate}
\caption{Scaling algorithm for \cref{thm:scaling}}\label{alg:scaling}
\end{Algorithm}
\footnotetext{Usually the Cholesky decomposition refers to $\rho = L L^{\dagger}$ where $L$ is lower triangular. However using such a decomposition for a different matrix, one can easily obtain $\rho = R R^{\dagger}$, where $R$ is upper triangular. Simply set $R = P L P$ where $P$ is a permutation matrix which swaps $i$ and $n-i$ and $P \rho P = L L^{\dagger}$, where $L$ is lower triangular.}

\begin{restatable}[Tensor scaling]{thm}{mainthm}\label{thm:scaling}
Let $X \in \Ten(n_0;n_1,\dots,n_d)$ be a (nonzero) tensor whose entries are Gaussian integers of bitsize no more than~$b$.
Also, let $\vec p \in P_+(n_1,\dots,n_d)$ with rational entries of bitsize no more than~$b$ such that $p^{(i)}_{n_i}>0$~for all $i=1,\dots,d$.
Finally, let $\eps>0$.

Then, with probability at least 1/2, \cref{alg:scaling} either correctly identifies that $\vec p\not\in\Delta(X)$, or it outputs $g\in G$ such that the marginals of $Y=g\cdot X$ are $\eps$-close to~$\vec p$. In fact, we have
\begin{align}\label{eq:final marginals}
  \lVert \rho_Y^{(i)} - \diag(\vec p_\uparrow^{(i)}) \rVert_{\tr} \leq \eps \quad \text{ for } i=1,\dots,d
\end{align}
in the latter case, where $\lVert A\rVert_{\tr} = \tr[\sqrt{A^\dagger A}]$ is the \emph{trace norm}.
\end{restatable}

\begin{rem} Note that the condition
$$
\lVert \rho_Y^{(i)} - \diag\left(\vec p_\uparrow^{(i)}\right) \rVert_{\tr} \leq \eps
$$
implies that
$$
\lVert \spec\left(\rho_Y^{(i)} \right) - \diag\left(\vec p^{(i)}\right) \rVert_{1} \leq \eps
$$
See \cref{lem:Bhatia}.
\end{rem}

To analyze our algorithm and prove \cref{thm:scaling}, we follow a three-step argument similar to the analysis in \cref{sec:matrix_scaling}. This has been used to great effect for operator scaling and tensor scaling in~\cite{gurvits2004classical,garg2016deterministic, burgisser2017alternating, franks2018operator} after identifying the appropriate potential function.

As we described in \cref{sec:matrix_scaling}, the appropriate potential functions to choose are the ones which are eigenvectors of an appropriate group action. In the matrix scaling case, we were acting by $\text{T}(n) \times \text{T}(n)$ and hence we chose the potential function to be permanent which is an eigenvector for this group action. In our algorithm, we are acting by the group corresponding to (direct products of) upper triangular matrices (this is known as the Borel subgroup). So for us, the right potential functions to consider are functions which are eigenvectors for the action of (tuples of) upper triangular matrices. One such class of functions are the so called highest weight vectors from representation theory\footnote{Here we restrict our attention to the action on polynomials because that is what we need to describe the intuition for the analysis of the algorithm. But the discussion of weight vectors applies to arbitrary (rational) representations of the group $G$, see \cref{subsec:hwtheory}.}, which we come to next.

What are highest weight vectors? We have the action of $G$ on $V = \Ten(n_0;n_1,\dots,n_d)$. Let us consider the space of degree $k$ polynomial functions on $V$, denoted by $\CC[V]_{k}$. The action of $G$ on $V$ induces an action of $G$ on $\CC[V]_{k}$ given by $(g \cdot P)(v) = P\left( g^{-1} \cdot v\right)$. Consider a tuple of vectors $\vec \lambda = \left( \vec \lambda^{(1)},\ldots, \vec \lambda^{(d)}\right)$, $\vec \lambda^{(i)} \in \ZZ^{n_i}$. Then we say that $P$ is a highest weight vector with weight $\vec \lambda$ if
$$
g \cdot P = \prod_{i=1}^d \prod_{j=1}^{n_i} \left(g^{(i)}_{j,j}\right)^{\lambda^{(i)}_j} P
$$
for all $g = \left( g^{(1)},\ldots, g^{(d)}\right)$ such that $g^{(i)}$ is an upper triangular matrix for each $i$. Note that this necessitates $\sum_{j=1}^{n_i} \lambda^{(i)}_j = -k$ for each $i$. This also necessitates (not trivial to see why) that for all~$i$, $\lambda^{(i)}_1 \ge \cdots \ge \lambda^{(i)}_{n_i}$.

The following two properties of highest weight vectors will be crucial for our analysis:

\begin{enumerate}
\item{} [\cite{ness1984stratification}, see \cref{thm:mumford}]: Let $\vec p \in P_+(n_1,\dots,n_d)$ be a rational vector. Then $\vec p \in \Delta(X)$ iff there exists an integer $k \ge 1$ s.t. $\vec \lambda = k \vec p$ has integer entries and there exists a highest weight vector $P$ with weight $\vec \lambda^*$ s.t. $(g \cdot P) \left( X\right) \neq 0$ for some $g \in G$. Here $\vec \lambda^* = \left( \left( -\lambda^{(1)}_{n_1},\ldots, -\lambda^{(1)}_{1}\right), \ldots, \left( -\lambda^{(d)}_{n_d},\ldots, -\lambda^{(d)}_{1}\right)\right)$. This extends a fact used in previous papers: the uniform vector is is $\Delta(X)$ iff some invariant polynomial does not vanish on $X$.
\item{} [\cref{prp:hwv eval}] The space of highest weight vectors with weight $\vec \lambda^*$ is spanned by polynomials $P$ with integer coefficients that satisfy the following bound
\begin{align}
|P(X)| \le (n_1 \cdots n_d)^k \norm{X}^k \label{eqnintro:1}
\end{align}
This extends an identical bound in past papers from invariant polynomials to highest weight vectors.
\end{enumerate}

 We use classical constructions of highest weight vectors \cite{procesi2007lie,burgisser2013explicit,burgisser2017alternating} to derive the second fact. These constructions are only semi-explicit (e.g. it is not clear if they can be evaluated efficiently), however they suffice for us because we only need a bound on their evaluations for their use as a potential function. We note that such bounds on their evaluations haven't been observed before in the invariant theory literature (except in \cite{burgisser2017alternating} for the special case of invariants) whereas for us they are extremely crucial! We also emphasize that it is crucial for us that the bound is singly exponential in $k$. Some naive strategies of using solution sizes for linear systems only yield bounds that are doubly exponential in $k$.

The potential function we use is $\Phi(g) = |P(g \cdot X)|^{1/k}$. Here $P$ is some highest weight vector of degree $k$ (for some $k$), integer coefficients and weight $\lambda^*$ that satisfies $(g \cdot P) \left( X\right) \neq 0$ as well as \cref{eqnintro:1}. Such a $P$ exists by the discussion above. Using these properties, a three-step analysis, similar to the one in \cref{sec:matrix_scaling}, follows the following outline.

\begin{enumerate}
\item{} [\textbf{Lower bound}]: Since $(g \cdot P) \left( X\right) \neq 0$ for some $g$, therefore for a random choice of $g$, $|P(g \cdot X)| \neq 0$. Furthermore, since we choose $g$ to have integer coefficients, $|P(g \cdot X)| \ge 1$. After the normalization in Step (2), we get $\Phi(g) \ge 1/f(n_0,\ldots, n_d, d, b, M)$. It is not hard to see that $f(n_0,\ldots, n_d, b, M) \le 2^b M^d \left(n_0 n_1 \cdots n_d \right)^2$.
\item{} [\textbf{Progress per step}]: $\Phi(g)$ increases at each step. Furthermore, if the current spectrum are ``far" from the target spectrum, then one can explicitly bound the increase. Here the highest weight vector property of $P$ as well as Pinsker's inequality from information theory play an important role.
\item{} [\textbf{Upper bound}]: $\Phi(g) \le n_1 \cdots n_d$ always. This follows from \cref{eqnintro:1} and the fact that we maintain the unit norm property of $g  \cdot X$ after the normalization in Step (2) of the algorithm.
\end{enumerate}

These three steps imply that in a polynomial number of iterations, one should get close to the target spectrum. A complete analysis is presented in \cref{thm:scaling:proof}. Note that to ensure that we only use a polynomial amount of random bits for the initial randomization, we need the highest weight vectors to have degree at most exponential in the input parameters. This is achieved by relying on Derksen's degree bounds \cite{derksen2001polynomial} (see \cref{prp:effective mumford}).

\subsection{Additional discussion}

We would like to point out two important distinctions between the analysis for matrix scaling in \cref{sec:matrix_scaling} and our analysis here. First is that, as we have seen, there is a major difference between the uniform and the non-uniform versions of our problem - while this was not the case for matrix scaling. This phenomenon is general and is a distinction between commutative and non-commutative group actions. It has to do with the fact that all irreducible representations of commutative groups are one-dimensional, whereas for non-commutative groups they are not. Secondly, in the matrix scaling analysis, the upper bound was easy to obtain as well. Whereas for us, the upper bound step is the hardest and requires the use of deep results in representation theory. The upper bound steps were the cause of main difficulty in the papers \cite{garg2016deterministic, burgisser2017alternating, franks2018operator} as well \footnote{In some of the papers, lower bound is the hard step, due to the use of a dual kind of potential function.} and this is one key point of distinction between commutative and non-commutative group actions.

We believe that our framework of using the highest weight vectors as potential functions for the analysis of analytic algorithms is the right way to approach moment polytope problems - even beyond the cases that we consider in this paper.

 The approach taken in \cite{franks2018operator} (for the case of $d=2$) is one of reducing the non-uniform version of the problem to the uniform version, which was solved in \cite{garg2016deterministic} for the case of $d=2$ (the reduction in \cite{garg2017algorithmic} is a simple special case of the reduction in \cite{franks2018operator}). The reduction is complicated and a bit ad hoc. We generalize this reduction to our setting ($d > 2$) in \cref{sec:reductions}, and providing a somewhat more principled view of the reduction along the way. However, it still seems rather specialized and mysterious compared to the general reduction in geometric invariant theory from the ``non-uniform'' to the ``uniform'' case (also known as the \emph{shifting trick}, see \cref{subsec:shifting trick}).

We also note that applying the results of \cite{burgisser2017alternating} to the reduction in \cref{sec:reductions} in a black-box manner does not yield our main theorem (\cref{thm:main}) - the number of iterations would be exponential in the \emph{bit-complexity} of $\vec p$, and we would even require an exponential number of bits for the randomization step! To remedy these issues with the reduction in \cref{sec:reductions} one must delve in to the relationship between the reduction and the invariant polynomials. We will see, by fairly involved calculation, that invariant polynomials evaluated on the reduction will result in the same construction of highest weight vectors anyway. This teaches us two lessons:
\begin{enumerate}
\item Highest weight vectors are the only suitable potential functions in sight. Though it may have other conceptual benefits, the reduction in \cref{sec:reductions} is no better than the shifting trick for the purpose of obtaining potential functions!
\item We had to look at the construction of highest weight vectors in \cref{subsec:hwv} \emph{before} calculating them from the reduction - the calculation might not have been so easy a priori! Again, the classical construction of highest weight vectors saves the day.
\end{enumerate}

It is interesting to discuss some of the salient features and possible variations of \cref{alg:scaling} (we expand on these points in the main text):

\begin{itemize}
\item \textbf{Iterations and randomness.} The algorithm terminates after at most \linebreak $T=\poly\left(\max_{i=0}^d n_i,d,b, 1/\eps\right)$ iterations and uses $\log_2(M)=\poly(\max_{i=0}^d n_i,d, b)$ bits of randomness.
For fixed or even inverse polynomial $\eps>0$, this is polynomial in the input size. In fact, this is better than the number of iterations in \cite{franks2018operator}: there, the number of iterations also depended on $\left(p^{(i)}_{n_i} \right)^{-1}$. 
\item \textbf{Bit complexity}: To get an algorithm with truly polynomial run time, one needs to truncate the group elements $g^{(i)}$'s up to polynomial number of bits after the decimal point. We provide an explanation on why this doesn't affect the performance of the algorithm in \cref{subsec:bit_complexity}.
\item \textbf{Degenerate spectra.} If $\lambda^{(i)}$ is degenerate, i.e. $\lambda^{(i)}_j =\lambda^{(i)}_k$ for some $j \neq k$, then we may replace the Cholesky decomposition in step~\ref{it:borel scale} by into two block upper triangular matrices, where the block sizes are the degeneracies - the set of such matrices is a so-called \emph{parabolic subgroup} of the general linear group (\cref{subsec:parabolic scaling}).
Moreover, the random matrix~$g$ need only be generic up to action of the parabolic subgroup.
In particular, when scaling to uniform spectra then no randomization is required and we can use Hermitian square roots, so \cref{alg:scaling} reduces to the uniform tensor scaling algorithm of~\cite{burgisser2017alternating}.
\item \textbf{Singular spectra.} As written, \cref{it:borel scale} of \cref{alg:scaling} fails if the spectra are singular, that is if for some $i$ we have $r_i:=\rk\diag (\vec p^{(i)}) < n_i$. However, in this case, one may first pass to a smaller tensor tensor $X_+$ obtained by restricting the $i^{th}$ index to the last $r_i$ coordinates. We'll show in \cref{sec:singular_spectra} that $X_+$ is scalable by upper triangulars to marginals $\diag(p^{(i)}_{r_i}, \dots, p^{(i)}_{1})$, $i \in [d]$ if and only if $X$ is scalable by upper triangulars to $\diag(0,\ldots,0, p^{(i)}_{r_i}, \dots, p^{(i)}_{1})$, $i \in [d]$.

\end{itemize}

\noindent We discuss extensions of \cref{alg:scaling} for more general varieties with ``good'' parametrizations in \cref{sec:extensions}.

\subsection{Conclusions and open problems}

We provide an efficient weak membership oracle for moment polytopes corresponding to a natural class of group actions on tensors. This generalizes recent works on operator and tensor scaling and also yields efficient algorithms for promise versions of the one-body quantum marginal problem and the asymptotic support of Kronecker coefficients. Our work leaves open several interesting questions some of which we state below.

\begin{itemize}
\item Improve the dependency on error $\eps$ in the running time of \cref{alg:scaling} to $\poly(\log(1/\eps))$. As discussed, this will immediately yield polynomial time algorithms for the one-body quantum marginal problem. This is open even for the uniform version of the problem. Here the notion of geodesic convexity of certain ``capacity" optimization problems should play a key role (e.g. see \cite{AGLOW18}).
\item Extend the weak membership oracle we develop to moment polytopes of other group actions, using Kirwan's gradient flow~\cite{kirwan1984cohomology} as proposed in~\cite{walter2014multipartite}.
The quantitative tools developed in this paper naturally extend to this setup and will elaborate on this in forthcoming work.
\item Develop separation oracles for moment polytopes. A related question is: can we optimize over moment
polytopes? This will have algorithmic applications on the problem of computing quantum functionals, as described
in~\cite{christandl2017universal}. In this paper, Strassen's support functionals are generalized
to quantum functionals, which are defined by convex optimization over
the entanglement polytope. Thus, separation oracles for moment polytopes could lead to efficient algorithms for
computing quantum functionals, which are important for comparing tensor powers (see~\cite{S86, S88}).
\item Find natural instances of combinatorial optimization problems which can be encoded as moment polytopes. Some examples can be found in \cite{garg2017algorithmic}.
\end{itemize}

\subsection{Roadmap of the paper}

In Section~\ref{sec:git}, we present results from geometric invariant theory and explain how they can be made quantitative.
We use this in Section~\ref{sec:analysis}, where we analyze the proposed tensor scaling algorithm. In

Section~\ref{sec:reductions}, we explain how the reduction in \cite{franks2018operator} can be naturally understood in the framework of this paper. In Section~\ref{sec:distance},
we show a lower bound on the distance to the moment polytope of any rational point not contained in it. This lower
bound depends only on the description of the rational point and the dimension of our tensor space $V$, and it allows us to solve membership problems by using the tensor scaling algorithm.
In Section~\ref{sec:extensions}, we extend our algorithm to general varieties and degenerate spectra. In Appendix~\ref{sec:polytope} we
discuss the Borel polytope, providing an alternate proof that it is in fact a rational polytope.

\section*{Acknowledgements}
We would like to thank Shalev Ben David, Robin Kothari, Anand Natarajan, Frank Verstraete, John Watrous, John Wright, and Jeroen Zuiddam for interesting discussions.

PB is partially supported by DFG grant BU 1371 2-2.
CF is supported in part by Simons Foundation award 332622.
MW acknowledges financial support by the NWO through Veni grant no.~680-47-459.
AW is partially supported by NSF grant CCF-1412958.

\section{Geometric invariant theory}\label{sec:git}

In this section, we present some results from geometric invariant theory that will feature centrally in the analysis of our algorithm in \cref{sec:analysis}.
While stated for tensors, all results in this section can easily be extended to arbitrary rational representations of connected complex reductive algebraic groups. Most of the results are well known and only some are new. All previously known results will be cited with references and we will make sure to highlight the new components. \cref{subsec:hwtheory} discusses basics of the highest weight theory. \cref{subsec:shifting trick} gives a formal definition of the moment map and also discusses the so called ``shifting trick" that reduces the problem of membership in moment polytopes to a null cone problem. \cref{subsec:effective mumford} considers degree bounds for highest weight vectors which are used to bound the initial randomness used in \cref{alg:scaling}. \cref{subsec:hwv} recalls a classical construction of highest weight vectors and uses this construction to prove bounds on their evaluations (crucial in the analysis of \cref{alg:scaling}). \cref{subsec:borel} develops a necessary and sufficient condition for Borel scalability (i.e., scaling using tuples of upper-triangular matrices).

As before, let $G=\GL(n_1)\times\dots\times\GL(n_d)$, $K=U(n_1)\times\dots\times U(n_d)$, $V=\Ten(n_0;n_1,\dots,n_d)=\CC^{n_0}\ot\CC^{n_1}\ot\dots\ot\CC^{n_d}$, and $\mathcal X\subseteq\PP(V)$ a $G$-stable irreducible projective subvariety (e.g., an orbit closure).

\subsection{Highest weight theory}\label{subsec:hwtheory}
We first recall the representation theory of~$\GL(n)$ (see, e.g., \cite{fulton2013representation} for an introduction).
Let~$W$ be a finite-dimensional $\GL(n)$-representation, equipped with a $U(n)$-invariant inner product.
Let $T(n) \subseteq \GL(n)$ denote the subgroup consisting of invertible diagonal matrices, called the \emph{maximal torus} of~$\GL(n)$.
Since $T(n)$ is commutative, its action can be jointly diagonalized.
Thus, any finite-dimensional $\GL(n)$-representation~$W$ can be written as a direct sum of so-called \emph{weight spaces}, $W = \bigoplus_\omega W_{(\vec\omega)}$, where $T(n)$ acts on any vector $w\in W_{(\vec\omega)}$ as $T \cdot w = \chi_{\vec\omega}(T) w$ for all $T\in T(n)$.
Here, $\omega$ is an integer vector and $\chi_{\vec\omega}(T) = \prod_{j=1}^n T_{j,j}^{\omega_j}$.
We write $\Omega(W)$ for the set of all weights that occur in~$W$.
Now let $B(n) \subseteq \GL(n)$ denote the \emph{Borel subgroup} of invertible upper-triangular matrices, which contains~$T(n)$.
A \emph{highest weight vector} is a vector~$w\in W$ that is an eigenvector of the $B(n)$-action.
Let~$\vec\lambda$ denote its weight, which is now called \emph{highest weight}.
Necessarily, $\lambda_1\geq\dots\geq\lambda_n$, i.e., $\vec\lambda$ is ordered non-increasingly, and we have that $R \cdot w = \chi_{\vec\lambda}(R) w$ for all $R\in B(n)$, where $\chi_{\vec\lambda}(R) = \prod_{j=1}^n R_{j,j}^{\lambda_j}$.
We denote by $\HWV_{\vec\lambda}(W)$ the space of highest weight vectors in~$W$ with highest weight~$\vec\lambda$.
The \emph{irreducible} representations of $\GL(n)$ contain a unique (up to scalar multiple) highest weight vector and are characterized by its highest weight.
We write $V_{\vec\lambda}$ for the irreducible representation (which we always equip with a $K$-invariant inner product, denoted~$\braket{-,-}$) and $v_{\vec\lambda}$ for a highest weight vector (which we choose to be of unit norm).
Thus, $\HWV_{\vec\mu}(V_{\vec\lambda}) = \CC v_{\vec\lambda}$ if $\vec\lambda=\vec\mu$, and zero otherwise.
It is known that $\partial_{t=0} \braket{v_{\vec\lambda}, \exp(At) \cdot v_{\vec\lambda}} = \tr[A \diag(\vec\lambda)]$ for all $n\times n$-matrices~$A$.
It can also be verified that $\GL(n) \cdot [v_{\vec\lambda}] = U(n) \cdot [v_{\vec\lambda}]$ (in particular, this $G$-orbit is closed). The dual of an irreducible representation is also irreducible with highest weight $\vec\lambda^*=- \vec \lambda_\uparrow$, so that $V_{\vec\lambda}^* \cong V_{\vec\lambda^*}$.

\bigskip

We now consider the group $G=\GL(n_1)\times\dots\times\GL(n_d)$.
All the preceding notions generalize immediately by considering tuples or tensor products of the relevant objects, and we shall use similar notation.
Thus, the \emph{maximal torus} is $T=T(n_1)\times\dots\times T(n_d)$, the \emph{Borel subgroup} is $B=B(n_1)\times\dots\times B(n_d)$.
\emph{Highest weight vectors} satisfy
\begin{align}\label{eq:hwv}
  R \cdot w = \chi_{\vec\lambda}(R) w,
\quad\text{where}\quad
\chi_{\vec\lambda}(R) = \prod_{i=1}^d \prod_{j=1}^{n_i} (R^{(i)}_{j,j})^{\lambda^{(i)}_j}
\end{align}
for all tuples $R=(R^{(1)},\dots,R^{(d)})\in B$, and \emph{weight vectors} satisfy the same relation restricted to $T\subseteq B$.
\emph{Weights} and \emph{highest weight} are now tuples $\vec\lambda=(\vec\lambda^{(1)},\dots,\vec\lambda^{(d)})$ of integer vectors as before.
The sums~$\sum_{j=1}^{n_i} \lambda^{(i)}_j$ are necessarily equal for $i=1,\dots,d$, and we will denote them by~$\lvert\vec\lambda\rvert$.
Thus, $\vec\lambda/\lvert\vec\lambda\rvert \in P_+(n_1,\dots,n_d)$.
We denote by $\HWV_{\vec\lambda}(W)$ the space of highest weight vectors in a $G$-representation~$W$.
The irreducible representations of~$G$ are again labeled by their highest weight and denoted by $V_{\vec\lambda}$.
Indeed, they are simply given by tensor products of the corresponding $\GL(n_i)$-representations, i.e., $V_{\vec\lambda} = V_{\vec\lambda^{(1)}}\ot\dots\ot V_{\vec\lambda^{(d)}}$; the same holds for their highest weight vectors.
For every tuple of matrices $A=(A^{(1)},\dots,A^{(d)})$ ($A^{(i)}$ is $n_i\times n_i$), we have that
\begin{align}\label{eq:moment map coadjoint orbit}
  \partial_{t=0} \braket{v_{\vec\lambda}, \exp(At) \cdot v_{\vec\lambda}} = \sum_{i=1}^n \tr[A^{(i)} \diag(\vec\lambda^{(i)})],
\end{align}
where $\exp(At) := \exp(A^{(1)}t)\ot\dots\ot\exp(A^{(d)}t)$.
As before, we write $\vec\lambda^*=((\vec\lambda^{(1)})^*,\dots,(\vec\lambda^{(1)})^*)$, so that $V_{\vec\lambda}^*\cong V_{\vec\lambda^*}$.

\subsection{Moment map and shifting trick}\label{subsec:shifting trick}

Let $W$ be a~$G$-representation. The associated \emph{moment map} is defined as
\begin{align}\label{eq:general moment map}
  \mu_W\colon \PP(W) \to  \Herm(n_1)\times\dots\times\Herm(n_d), \quad
  [w] \mapsto \mu_W([w]) = (\mu_W^{(1)},\dots,\mu_W^{(d)})
\end{align}
by the property that
\begin{align*}
  \sum_{i=1}^d \tr[\mu^{(i)}_W([w]) A^{(i)}] = \partial_{t=0} \frac{\braket{w, \exp(At) \cdot w}}{\braket{w, w}}
\end{align*}
for all tuples of matrices $A=(A^{(1)},\dots,A^{(d)})$ ($A^{(i)}$ is $n_i\times n_i$).
Note that $\mu_W$ is $K$-equivariant, i.e., $\mu_W^{(i)}([(U^{(1)},\dots,U^{(d)}) \cdot w]) = U^{(i)} \mu_W^{(i)}([w]) (U^{(i)})^\dagger$ for all unitary $n_i\times n_i$-marices $U^{(i)}$ and $i=1,\dots,d$.
Given a $G$-stable irreducible projective subvariety~$\mathcal Z\subseteq\PP(W)$, we define the corresponding \emph{moment} or \emph{Kirwan polytope} by
\begin{align}\label{eq:general moment polytope}
  \Delta_W(\mathcal Z) := \{(\spec(\mu^{(1)}_W([w])),\dots,\spec(\mu^{(d)}_W([w]))) : [w]\in \mathcal Z\} \subseteq P_+(n_1,\dots,n_d).
\end{align}
It is known that $\Delta_W(\mathcal Z)$ is always a rational convex polytope~\cite{ness1984stratification,kirwan1984cohomology,kirwan1984convexity,brion1987image} (and we will see below why this is the case).

In \cref{sec:intro}, we had already seen an example of a moment map and a moment polytope.
Indeed, \cref{eq:moment map,eq:moment polytope} are precisely the special cases of \cref{eq:general moment map,eq:general moment polytope} when $W=V=\Ten(n_0;n_1,\dots,n_d)$, as follows readily from \cref{eq:marginal}.
Thus, it is natural to think of the moment map as a generalization of the notion of a `marginal'!
For a second example, note that \cref{eq:moment map coadjoint orbit} and the fact that
$\GL(n) \cdot [v_{\vec\lambda}] = U(n) \cdot [v_{\vec\lambda}]$ imply that
$\mu^{(i)}_{V_{\vec\lambda}}([v_{\vec\lambda}]) = \diag(\vec\lambda^{(i)})$, so
$\Delta_{V_{\vec\lambda}}(G \cdot [v_{\vec\lambda}]) = \{\vec\lambda\}$ is a single point.

These two examples can be combined in a simple but useful way, known as the
`shifting trick'.

\begin{lem}[Shifting trick, geometric part~\cite{brion1987image}]\label{lem:geometric shift}
  Let $\vec p\in P_+(n_1,\dots,n_d)$ and $\ell>0$ an integer such that $\vec\lambda:=\ell\vec p$ is integral.
  Let $\mathcal X\subseteq\PP(V)$ be a $G$-stable irreducible subvariety.
  Consider the representation $W:=\Sym^\ell(V) \ot V_{\vec\lambda^*}$.
  Then,
  \begin{align}\label{eq:geometric shift}
    \mu^{(i)}_W([Y^{\ot \ell} \ot v_{\vec\lambda^*}])
    = \ell \mu^{(i)}_V([Y]) + \diag((\vec\lambda^*)^{(i)})
    = \ell \rho^{(i)}_Y + \diag((\vec\lambda^*)^{(i)})
  \end{align}
  for all $[Y]\in\mathcal X$, $\lVert Y\rVert=1$, and $i=1,\dots,d$.
  In particular, $\vec p\in\Delta(\mathcal X)$ if and only if there exists $[Y]\in\mathcal X$ such that $\mu_W([Y^{\ot \ell} \ot v_{\vec\lambda^*}]) = 0$.
\end{lem}
\begin{proof}
  By definition, $\vec p\in\Delta(\mathcal X)$ means that there exists $Y\in\mathcal X$ such that $\spec(\rho^{(i)}_Y) = \vec p^{(i)} = \vec\lambda^{(i)}/\ell$ for all $i=1,\dots,d$.
  By applying a suitable element in $U(n_i)$, we may in fact assume that
  \begin{align*}
    \rho^{(i)}_Y = \diag(\vec p^{(i)}_\uparrow) = -\frac1\ell\diag((\vec\lambda^*)^{(i)}).
  \end{align*}
  for all $i=1,\dots,d$.
  But note that
  \begin{align*}
    \mu^{(i)}_W([Y^{\ot \ell} \ot v_{\vec\lambda^*}])
  = \ell \mu^{(i)}_V([Y]) + \mu^{(i)}_{V_{\vec\lambda^*}}([v_{\vec\lambda^*}])
  = \ell \rho^{(i)}_Y + \diag((\vec\lambda^*)^{(i)})
  \end{align*}
  (the first equation follows easily from the product rule and the second from the two examples that we just discussed), so the two conditions are indeed equivalent.
\end{proof}

\begin{rem}
\Cref{lem:geometric shift} can also be stated in the following way:
Consider the $G$-stable irreducible subvariety $\mathcal Z := \{ [X^{\ot\ell} \ot g \cdot v_{\vec\lambda^*}] : [X] \in\mathcal X, g\in G \}$ of $\PP(\Sym^\ell(V)\ot V_{\vec\lambda^*})$.
Then $\vec p\in\Delta(\mathcal X)$ if and only if $0\in\Delta(\mathcal Z)$.
\end{rem}

\Cref{lem:geometric shift} shows that membership in the moment polytope~$\Delta(\mathcal X)$ can always be reduced to zeros of the moment map -- special `uniform marginals' -- provided we are willing to work in a larger space.

The `shifting trick' has an invariant-theoretic counterpart.
To state it, consider~$\CC[W]=\bigoplus_k \CC[W]_{(k)}$, the algebra of polynomials on~$W$, graded by degree.
Then $G$ acts on polynomials~$P\in\CC[W]_{(k)}$ by~$(g \cdot P)(X) := P(g^{-1} X)$, so each $\CC[W]_{(k)}$
is also a rational representation of~$G$. Thus, this allows us to speak of polynomials that are highest
weight vectors and, in particular, of $G$-invariant polynomials.
Then we have the following result (see, e.g. \cite{brion1987image}):

\begin{lem}[Shifting trick, invariant-theoretic part]\label{lem:invariant-theoretic shift}
  Let $\vec\lambda$ be a highest weight and $\ell=\lvert\vec\lambda\rvert$.
  Let $Q$ be a $G$-invariant polynomial in $\CC[\Sym^\ell(V) \ot V_{\vec\lambda^*}]^G_{(m)}$.
  Then $P(X) := Q(X^{\ot\ell} \ot v_{\vec\lambda^*})$ is a highest weight vector in $\HWV_{m\vec\lambda^*}(\CC[V]_{(\ell m)})$.
  Conversely, every highest weight vector arises in this way.
\end{lem}

The significance of \cref{lem:invariant-theoretic shift} is that it allows us to reduce questions about highest weight vectors to polynomials with the zero highest weight, i.e., invariant polynomials.

\subsection{An effective version of Mumford's theorem}\label{subsec:effective mumford}

In this section, we prove degree bounds for the nonvanishing of highest weight vectors (\cref{prp:effective mumford}) using Derksen's degree bounds in the invariant setting \cite{derksen2001polynomial} and the shifting trick introduced in \cref{subsec:shifting trick}. These degree bounds will prove useful to upper bound the initial amount of randomness needed in \cref{alg:scaling}. Bounding the amount of randomness is an easy consequence of the Schwartz-Zippel lemma and the degree bounds and this is done in \cref{cor:random is generic}.

The following theorem shows that points in the moment polytope are characterized by the nonvanishing of highest weight vectors in the algebra of polynomials -- as perhaps already suggested by the analogy beween \cref{lem:geometric shift,lem:invariant-theoretic shift}.
Since we will be interested in the moment polytope of the subvariety $\mathcal X\subseteq \PP(V)$, we state the theorem in this situation (however, it generalizes verbatim to general $G$-representations).

\begin{thm}[\cite{ness1984stratification}]\label{thm:mumford}
Let $\vec p\in P_+(n_1,\dots,n_d)$ with rational entries.
Then, $\vec p\in\Delta(\mathcal X)$ if and only if there exists a positive integer $k>0$ such that $\vec\lambda=k\vec p$ is integral and there exists a highest-weight vector $P\in\HWV_{\vec\lambda^*}(\CC[V]_{(k)})$ such that $P(X)\neq0$ for some $X\in\mathcal X$.
\end{thm}

Explicitly, $P\in\HWV_{\vec\lambda^*}(\CC[V]_{(k)})$ means that
\begin{align}\label{eq:hwv poly}
  P(R \cdot X)
= \chi_{\lambda^*}(R^{-1}) P(X)
= \left( \prod_{i=1}^d \prod_{j=1}^{n_i} (R^{(i)}_{j,j})^{\lambda^{(i)}_{n_i+1-j}} \right) P(X)
\end{align}
for all $R=(R^{(1)},\dots,R^{(d)})\in B$.
We give a proof of a refinement of \cref{thm:mumford} in \cref{subsec:borel}.

\Cref{thm:mumford} alone does not appear to give an efficient way of characterizing the moment polytope since it provides no bound on the degree~$k$ nor a recipe for finding a point~$X$ s.t. $P(X) \neq 0$.
In fact, it is known that even deciding the existence of highest-weight vectors is NP-hard~\cite{ikenmeyer2017vanishing}.
Our algorithm does not solve the membership problem via the dual description provided by Mumford's theorem.
Instead, suitable highest weight vectors will feature as potential functions in the analysis of our algorithm!

We will nevertheless require a more effective understanding of \cref{thm:mumford}.
This will be the concern of the remainder of this section.
We start by observing that the algebra of highest weight vectors with highest weight a multiple of~$\vec p$ is finitely generated.
Thus, \cref{thm:mumford} can be made more effective by bounding the degree of the highest weight vectors that need to be considered.
This is achieved by our next result, which relies on recent degree bounds by Derksen~\cite{derksen2001polynomial}:

\begin{prp}[Effective Mumford's Theorem]\label{prp:effective mumford}
Let $\vec p\in P_+(n_1,\dots,n_d)$ and $\ell>0$ an integer such that $\vec\lambda=\ell\vec p$ is integral.
If~$\vec p\in\Delta(\mathcal X)$, then there exists an integer~$m>0$ and a highest weight vector $P\in\HWV_{m\vec\lambda^*}(\CC[V]_{(\ell m)})$ of degree $\ell m\leq K$, where
\begin{align*}
  K := \left(\ell d \textstyle\max_{i=1}^d n_i\right)^{d \max_{i=1}^d n_i^2},
\end{align*}
such that $P(X)\neq0$ for some $X\in\mathcal X$.
\end{prp}
\begin{proof}
By \cref{lem:invariant-theoretic shift}, any highest weight vector can be written as $P(X)=Q\left(X^{\ot \ell}\ot v_{\vec\lambda^*}\right)$, where $Q$ is a $G$-invariant polynomial on $W=\Sym^\ell(V)\ot V_{\vec\lambda^*}$.
Thus, by \cref{thm:mumford}, $\vec p\not\in\Delta(\mathcal X)$ if and only if $Q\left(X^{\ot \ell} \ot v_{\vec\lambda^*}\right) = 0$ for all $X\in\mathcal X$ and all nonconstant homogeneous polynomials $Q\in\CC[W]^G$.
By definition, the latter means that $X^{\ot \ell} \ot v_{\vec\lambda^*}$ is in the null cone of the $G$-action on~$W$ or, equivalently, of the action of the subgroup~$\tilde G=\SL(n_1)\times\dots\times\SL(n_d)$.\footnote{Since the action is scale invariant for each $\GL(n_i)$.}
The latter has finite kernel, so Derksen's degree bound from~\cite[Proposition 2.1]{derksen2001polynomial} is applicable.
It shows that the null cone is already defined by $G$-invariant polynomials of degree
\begin{align*}
  m
\leq H^{t-{\tilde m}} A^{\tilde m}
\leq H^d \left( \textstyle \sum_{i=1}^d n_i \lambda^{(i)}_1 \right)^{\sum_{i=1}^d (n_i^2 - 1)}
\leq (H\ell d)^{d H^2}/\ell,
\end{align*}
where
$H:=\max_{i=1}^d n_i$,
$A:=\sum_{i=1}^d n_i \lambda^{(i)}_1$, 
$\tilde m:=\sum_{i=1}^d (n_i^2 - 1)$, and 
$t:=\sum_{i=1}^d n_i^2$.
Since $\deg(P)=\ell\deg(Q)$, we obtain the desired degree bound.
\end{proof}

In \cref{subsec:hwv}, we will prove bounds on the evaluation of highest weight vectors (\cref{prp:hwv eval}).

\bigskip

\Cref{prp:effective mumford} shows that we only need to consider finitely many highest weight vectors to test
whether $\vec p\in\Delta(\mathcal X)$ (e.g., a basis of the space of all highest weight vectors of degree
$\leq K$).
However, we still need to test if $P(X)\neq0$ for some~$X\in\mathcal X$.
How can we find such an~$X$?
Clearly, $P(X)\neq0$ for some~$X\in\mathcal X$ iff $P(X)\neq0$ for generic $X\in\mathcal X$, so:

\begin{cor}[\cite{brion1987image}]\label{cor:generic is enough}
$\vec p\in\Delta(\mathcal X)$ if and only if $\vec p\in\Delta(X)$ for generic $X\in\mathcal X$.
\end{cor}

In fact, $\Delta(\mathcal X)=\Delta(X)$ for generic $X\in\mathcal X$, since both are rational convex polytopes. This can be seen from the nontrivial fact that the algebra of \emph{all} highest weight vectors is also finitely generated \cite{Gross73}.

To make \cref{cor:generic is enough} effective, we need a way to select generic elements in~$\mathcal X$.
This may be done using the Schwartz-Zippel lemma -- provided we have a suitable parametrization of $\mathcal X$. We will show how to do this in full generality in \cref{subsec:parametrizations}. We will also show there that if $\lambda^{(i)}$ is degenerate i.e. $\lambda^{(i)}_j = \lambda^{(i)}_k$ for $j \neq k$, we can use significantly fewer (sometimes even zero!) random bits to generate $X$. For now, we only carry this out for $\mathcal X = \overline{G \cdot X }$. This is critical for our analysis of \cref{it:randomize} in \cref{alg:scaling}.

\begin{cor}[Good starting points]\label{cor:random is generic}
Let $X \in \Ten(n_0; n_1, \dots, n_d)$. Suppose $\vec p\in\Delta(X)$ and $\ell>0$ such that $\vec\lambda:=\ell\vec p$ is integral. Let $M(n)$ denote the space of $n\times n$ matrices.
Choose all $\max_{i=1}^d n_i^2$ entries of the $d$-tuple of matrices $A=(A^{(1)},\dots,A^{(d)}) \in M(n_1)\times\dots\times M(n_d)$  independently and uniformly at random from~$\{1,\dots,M\}$, with
\begin{align*}
   M = 2 d K,
\quad
  K := \left(\ell d \textstyle\max_{i=1}^d n_i \right)^{d \max_{i=1}^d n_i^2}.
\end{align*}
Set
$$Z = A \cdot X:= \left(A^{(1)}\otimes \dots \otimes A^{(d)}\right) X.$$
Then, there exists a highest weight vector $P\in\HWV_{m\vec\lambda^*}(\CC[V]_{(\ell m)})$ of degree
$0<\ell m\leq K$ such that, with probability at least $1/2$, $P(A \cdot X)\neq0$.

\end{cor}
\begin{proof}
Set $\vec\lambda:=\ell\vec p$.
According to \cref{prp:effective mumford}, there exists a highest weight vector $P\in\HWV_{m\vec\lambda^*}(\CC[V]_{(\ell m)})$ of degree~$0<\ell m\leq K$ such that $P(Y)\neq0$ for some $Y\in\mathcal X = \overline{G \cdot X}$. Then
$
  Q(A) := P(A \cdot X)
$
is not equal to the zero polynomial on $M:= M(n_1)\times\dots\times M(n_d)$. If this were not the case, then $P$ would vanish in the Zariski dense (in $\mathcal X$) set $G\cdot X$, contradicting $P(Y) \neq 0$.
Its degree is no larger than $d K$, so the Schwartz-Zippel lemma implies that for our random choice of~$A$, $Q(A) = P(A \cdot X)\neq0$ with probability at least~$1/2$.
\end{proof}

\subsection{Construction of highest weight vectors}\label{subsec:hwv}
In this section, we will recall a classical construction of the space of highest weight vectors $\HWV_{\vec\lambda^*}(\CC[V]_{(k)})$ in the polynomial ring (cf.\ \cite{procesi2007lie,burgisser2013explicit,burgisser2017alternating}) and prove a bound on their evaluation. This bound will be crucial in the analysis of \cref{alg:scaling}.
Here, $\vec\lambda$ is a highest weight with $\lvert\vec\lambda\rvert=k$.

Any polynomial $P\in\CC[V]_{(k)}$ can be written as $P(X) = p(X^{\ot k})$, where $p$ is a linear form in $(V^{\ot k})^*$.
If $P$ is a highest weight vector then we can assume that $p$ itself is a highest weight vector of the same highest weight.
Next, note that
\begin{align}\label{eq:hwv factorize}
  \HWV_{\vec\lambda^*}((V^{\ot k})^*)
= ((\CC^{n_0})^{\ot k})^* \ot \HWV_{\left(\vec\lambda^{(1)}\right)^*}(((\CC^{n_1})^{\ot k})^*) \ot\dots\ot \HWV_{\left(\vec\lambda^{(d)}\right)^*}(((\CC^{n_d})^{\ot k})^*).
\end{align}
The right-hand side spaces are the spaces of highest-weight vectors for a single $\GL(n_i)$, and so are labeled by a single partition~$\left(\vec\lambda^{(i)}\right)^*$.

We thus start by constructing $\HWV_{\vec\mu^*}(((\CC^n)^{\ot k})^*)$ for a single $\GL(n)$.
Here, $\mu_1\geq\dots\geq\mu_n$ are integers with $\sum_{j=1}^d \mu_j=k$.
It is well-known that this space is nonzero only if $\mu_n\geq0$, i.e., $\vec\mu$ is a partition of~$k$ into at most $n$~parts.

First, consider the linear form $\Det_\ell \in ((\CC^n)^{\ot \ell})^*$ given by
\begin{align*}
  \Det_\ell(v_1\ot\dots\ot v_\ell) := \det \left[(v_i)_{n+1-j}\right]_{i,j=1,\dots,\ell}
\end{align*}
where we assume that $\ell\leq n$ (i.e., interpret the vectors as the columns of an $n\times \ell$-matrix and compute the determinant of the bottom-most $\ell\times \ell$-block).
Clearly, $\Det_\ell\neq0$ since it is nonzero on, e.g., the last $\ell$ standard basis vectors of $\CC^n$.
We claim that $\Det_\ell$ is a highest weight vector of weight $(0,\dots,0,-1,\dots,-1)$, with $\ell$ minus ones and $n-\ell$ zeros.
Indeed, if $R$ is an upper-triangular $n\times n$-matrix then
\begin{align*}
&\quad (R \cdot \Det_\ell)(v_1\ot\dots\ot v_\ell)
= \Det_\ell(R^{-1} v_1 \ot \dots \ot R^{-1} v_\ell)
= \det \left[(R^{-1} v_i)_{n+1-j}\right]_{i,j=1,\dots,\ell} \\
&= \det \left[\textstyle\sum_{j'=1}^\ell \left(R^{-1}\right)_{n+1-j,n+1-j'} (v_i)_{n+1-j'} \right]_{i,j=1,\dots,\ell}
= \left( \prod_{j=1}^\ell \left(R_{n+1-j,n+1-j}\right)^{-1} \right) \Det_l(v_1\ot\dots\ot v_\ell)
\end{align*}
where the last step follows from the multiplicativity of the ordinary determinant.

Now recall that the highest weight~$\vec\mu$ is a partition of~$k$ into at most $n$~parts.
Let $\vec\mu'$ denote its transpose, i.e., $\vec\mu'_1$ is the height of the first column of $\vec\mu$, etc., up to the last column, whose height is~$\mu'_{\mu_1}$
Note that each $\mu'_j\leq n$ and $\sum_j \mu'_j = k$.
Thus we can consider the vector
\begin{align*}
  \Det_{\vec\mu^*} := \Det_{\mu'_1} \ot \dots \ot \Det_{\mu'_{\mu_1}} \in ((\CC^n)^{\ot k})^*
\end{align*}
which is thus a nonzero highest weight vector of highest weight~$\vec\mu^*=(-\mu_n,\dots,-\mu_1)$.
We can produce many further highest weight vectors by permuting the~$k$ tensor factors by some~$\pi\in S_k$:
\begin{align*}
  \Det_{\vec\mu^*,\pi}(v_1\ot\dots\ot v_k) := \Det_{\vec\mu^*}(v_{\pi(1)}\ot\dots\ot v_{\pi(k)})
\end{align*}

\begin{lem}\label{lem:schur weyl}
The linear forms $\Det_{\vec\mu^*,\pi}$ for $\pi\in S_k$ span~$\HWV_{\vec\mu^*}(((\CC^n)^{\ot k})^*)$.
\end{lem}
\begin{proof}
Schur-Weyl duality asserts that the space of highest weight vectors is an irreducible $S_k$-representation.
It is therefore spanned by the $S_k$-orbit of any nonzero vector.
\end{proof}

As a direct consequence of \cref{lem:schur weyl} and the discussion surrounding \cref{eq:hwv factorize}, we obtain that the polynomials
\begin{align}\label{eq:hwv concrete}
  P(X) := (\eps_{i_1,\dots,i_k}\ot\Det_{(\vec\lambda^{(1)})^*,\pi^{(1)}}\ot\dots\ot\Det_{(\vec\lambda^{(d)})^*,\pi^{(d)}})(X^{\ot k}),
\end{align}
span the space of highest weight vectors~$\HWV_{\vec\lambda^*}(\CC[V]_{(k)})$.
Here, we have $i_1,\dots,i_k\in\{1,\dots,n_0\}$, $\pi^{(1)},\dots,\pi^{(d)}\in S_k$, and $\eps_{i_1,\dots,i_m}$ denotes the dual basis of the standard product basis of $(\CC^{n_0})^{\ot m}$.
We summarize in the following proposition, where we also establish a bound on their evaluation. We note that while the bound on the evaluations is an elementary consequence of the above construction, this has not appeared before in the literature. At the same time, this is a crucial part of our analysis of \cref{alg:scaling}.

\begin{prp}\label{prp:hwv eval}
The space of highest weight vectors~$\HWV_{\vec\lambda^*}(\CC[V]_{(k)})$ is nonzero only if
$\lambda^{(i)}_{n_i} \geq 0$ for all $i=1,\dots,d$.
In this case, it is spanned by the polynomials~$P(X)$ defined in \cref{eq:hwv concrete}, where $i_1,\dots,i_k\in\{1,\dots,n_0\}$ and $\pi^{(1)},\dots,\pi^{(d)}\in S_k$.
These are polynomials with integer coefficients, and they satisfy the bound
\begin{align*}
  \lvert P(X)\rvert \leq (n_1\dots n_d)^k \lVert X\rVert^k
\end{align*}
for all tensors $X\in V$.
\end{prp}
\begin{proof}
It only remains to verify the bound.
For this, let $X\in V=\Ten(n_0;n_1,\dots,n_d)$ be an arbitrary tensor.
Expand
\begin{align*}
  X
= \sum_{j^{(1)}=1}^{n_1}\dots\sum_{j^{(d)}=1}^{n_d} v_{j^{(1)},\dots,j^{(d)}} \ot e_{j^{(1)}} \ot \dots \ot e_{j^{(d)}},
\end{align*}
where the $v_{j^{(1)},\dots,j^{(d)}}$ are vectors in~$\CC^{n_0}$ and the $e_{j^{(i)}}$ the standard basis vectors of $\CC^{n_i}$, $i=1,\dots,d$.
Thus,
\begin{align*}
  X^{\ot k}
= \sum_{J^{(1)}\colon[k]\to[n_1]}\dots\sum_{J^{(d)}\colon[k]\to[n_d]}
&\left( \otimes_{\alpha=1}^k v_{J^{(1)}(\alpha),\dots,J^{(d)}(\alpha)} \right) \ot \left( \otimes_{\alpha=1}^k e_{J^{(1)}(\alpha)} \right) \ot \\
& \dots \ot \left( \otimes_{\alpha=1}^k e_{J^{(d)}(\alpha)} \right)
\end{align*}
and so
\begin{align*}
  P(X)
= \sum_{J^{(1)}\colon[k]\to[n_1]}\dots\sum_{J^{(d)}\colon[k]\to[n_d]}
&\eps_{i_1,\dots,i_k}\left( \otimes_{\alpha=1}^k v_{J^{(1)}(\alpha),\dots,J^{(d)}(\alpha)} \right) \cdot\Det_{(\vec\lambda^{(1)})^*,\pi^{(1)}}\left( \otimes_{\alpha=1}^k e_{J^{(1)}(\alpha)} \right)\\
& \cdots\Det_{(\vec\lambda^{(d)})^*,\pi^{(d)}}\left( \otimes_{\alpha=1}^k e_{J^{(d)}(\alpha)} \right).
\end{align*}
Let's consider a single summand.
The first factor is a product of $k$ many components of the vectors $v_{j^{(1)},\dots,j^{(d)}}$,
and so is bounded in absolute value by $\lVert X\rVert^k$.
The remaining factors are products of determinants of submatrices of matrices whose columns are standard basis vectors, hence equal to zero or $\pm1$.
Together,
\begin{align*}
  \lvert P(X)\rvert \leq \sum_{J^{(1)}\colon[k]\to[n_1]}\dots\sum_{J^{(d)}\colon[k]\to[n_d]} \lVert X\rVert^k = (n_1\cdots{}n_d)^k \lVert X\rVert^k.
\end{align*}
so we obtain the desired bound.
\end{proof}

\subsection{Borel scaling}\label{subsec:borel}
In this section, we prove a necessary and sufficient condition for scalability using upper-triangular matrices (i.e., the Borel subgroup) in terms of the non vanishing behavior of highest weight vectors. We claim no originality for this connection - this is probably well known to experts in geometric invariant theory. In fact, \cref{prp:borel,cor:borel} can also be proved as a consequence of (the analysis of) our algorithm!
But we believe that it is useful to give an independent argument which explains the initial randomization in \cref{alg:scaling} and puts it into a general context.

\begin{dfn}[Borel ``polytope'']
Define the \emph{Borel polytope} $\Delta^B(Y)\subset \Delta (Y)$ by
\begin{align}\label{eq:borel polytope}
  \Delta^{B}(X) :=
  \left\{ \vec{p}: \diag\left(\vec p_\uparrow^{(i)}\right)=\rho_Y^{(i)},\dots, \diag\left(\vec p_\uparrow^{(d)}\right) =\rho_Y^{(d)} \textrm{ for some } Y \in \overline{B \cdot [X]}\right\}.\end{align}
  Equivalently, $\vec{p} \in  \Delta^{B}(X)$ if and only if there exists $[Y]\in\overline{B\cdot [X]}$ such that $\spec\left(\rho^{(i)}_Y\right) = \vec p^{(i)}$ for all $i=1,\dots,d$.

\end{dfn}
It is a well known fact that $\Delta^B(X)$ is a polytope, but we only review the argument much later in \cref{sec:polytope}, hence the quotations.
\begin{prp}[Borel scaling]\label{prp:borel}
Let $X\in\Ten(n_0;n_1,\dots,n_d)$ be a (nonzero) tensor, $\vec p\in P_+(n_1,\dots,n_d)$, $k > 0$, and $\vec\lambda:=k\vec p$. If there exists a highest weight vector $P\in\HWV_{\lambda^*}(\CC[V]_{(k)})$ such that $P(X)\neq0$, then $\vec{p} \in \Delta^B(X)$.
\end{prp}
\begin{proof}
Assume $P\in\HWV_{\lambda^*}(\CC[V]_{(k)})$ and $P(X)\neq0$. We first show $X^{\ot k} \ot v_{\vec\lambda^*} \in W := \Sym^k(V) \ot V_{\vec\lambda^*}$ is not in the null cone, i.e., that
  \begin{align}\label{eq:borel semistable}
    \inf_{g\in G} \lVert (g \cdot X)^{\ot k} \ot (g \cdot v_{\vec\lambda^*}) \rVert
  = \inf_{R\in B} \lVert (R \cdot X)^{\ot k} \ot (R \cdot v_{\vec\lambda^*}) \rVert
  > 0.
  \end{align}
  The equality follows from the QR-decomposition and unitary invariance of the norm.
  Thus, using \cref{eq:hwv},
  \begin{align*}
    \lVert (R \cdot X)^{\ot k} \ot (R \cdot v_{\vec\lambda^*}) \rVert
  = \lVert R \cdot X \rVert^k \, \lvert\chi_{\vec\lambda^*}(R)\rvert.
  \end{align*}
  On the other hand, the assumption and \cref{eq:hwv poly} show that
  \begin{align*}
    P(R \cdot X) \chi_{\vec\lambda^*}(R) = P(X) \neq 0,
  \end{align*}
  and hence
  \begin{align*}
    \lVert R \cdot X\rVert^k \,\lvert \chi_{\vec\lambda^*}(R) \rvert
  = \frac {\lvert P(X) \rvert} {\lvert P\left(\frac {R\cdot X}{\lVert R\cdot X\rVert}\right) \rvert}
  \geq \frac {\lvert P(X) \rvert} {\sup_{\lVert Z\rVert=1} \lvert P(Z) \rvert} > 0,
  \end{align*}
  where we used that $P$~is a continuous function, so its supremum on the space of tensors of unit norm is finite.
  This uniform lower bound establishes \cref{eq:borel semistable}.

  In view of \cref{eq:borel semistable}, the infimum can be attained by $Y^{\ot k} \ot v_{\lambda^*}$ for some $[Y] \in \overline{B \cdot [X]}$.
  We may assume that $Y$ is a unit vector (otherwise rescale~$X$ appropriately).
  Since $Y^{\ot k}\ot v_{\lambda^*}$ has minimal norm in its $G$-orbit, its squared norm does not change to first order under the infinitesimal action of any one-parameter subgroup, such as $\exp(At)$, where $A=(A^{(1)},\dots,A^{(d)})$ is a tuple of Hermitian $n_i\times n_i$-matrices.
  But then
  \begin{align*}
    0
    = \partial_{t=0} \frac12\norm{ \exp(At) \cdot \left( Y^{\ot k}\ot v_{\lambda^*}\right) }^2
    = \mu_W([Y^{\ot k} \ot v_{\vec\lambda^*}]),
  \end{align*}
  by definition of the moment map (\cref{eq:general moment map}), so it follows from \cref{eq:geometric shift} that $\rho^{(i)}_Y = \diag\left(\vec p_\uparrow^{(i)}\right)$ for all $i=1,\dots,d$. Thus, $\vec p \in \Delta^B(X)$.
\end{proof}

Of course, the conclusion of the proposition is equivalent to the statement that for all $\eps>0$ there exists $R\in B$ such that the marginals of $R \cdot X$ are $\eps$-close to the prescribed ones.

The proof of the above proposition uses the following notion of ``capacity" which generalizes the optimization problems considered in \cite{gurvits2004classical, garg2016deterministic, burgisser2017alternating}.

\begin{dfn}[$\vec p$-capacity]\label{dfn:capacity} Given a tensor $X\in\Ten(n_0;n_1,\dots,n_d)$ and $\vec p\in P_+(n_1,\dots,n_d)$, define
$$
\capacity_{\vec{p}}(X) =  \inf_{R\in B}  \norm{(R \cdot X)} |\chi_{\vec p^*}(R)|
$$
Here $B$ is the Borel subgroup (tuples of upper triangular matrices). And
$$
\chi_{\vec p^*}(R) = \prod_{i=1}^d \prod_{j=1}^{n_i} \left( R^{(i)}_{j,j}\right)^{-p^{(i)}_{n+1-j}}
$$
\end{dfn}

The proof of \Cref{prp:borel} yields the following connection: $X$ is Borel-scalable to marginals $\vec p$ iff $\capacity_{\vec{p}}(X) > 0$. The proof of \Cref{prp:borel} when combined with \cref{prp:hwv eval} also yields the following lower bound on capacity. This greatly generalizes the lower bounds in \cite{gurvits2004classical, garg2016deterministic, garg2017algorithmic, burgisser2017alternating}.

\begin{thm}\label{thm:capacity_lb} Let $X\in\Ten(n_0;n_1,\dots,n_d)$ be a tensor with integer entries s.t. $\capacity_{\vec{p}}(X) > 0$. Then
$$
\capacity_{\vec{p}}(X)  \ge \frac{1}{n_1 \cdots n_d}
$$
\end{thm}

\Cref{prp:borel} implies that it suffices to scale by elements from the Borel subgroup -- provided we randomize the starting point.
Indeed, if $\vec p\in\Delta(\mathcal X)$ then we may first select a generic $X\in\mathcal X$ such that $P(X)\neq0$ for some highest weight vector (e.g., using \cref{cor:random is generic}), which is precisely what we do in our algorithm.
We summarize:

\begin{cor}\label{cor:borel}
Let $\vec p\in\Delta(X)$.
Then, for generic $Y\in \overline{G \cdot X}$, we have $\vec p \in \Delta ^B(X)$.
\end{cor}

We also have the following converse to \cref{prp:borel}.

\begin{prp}\label{prp:borel converse}
Let $X\in\Ten(n_0;n_1,\dots,n_d)$ be a (nonzero) tensor and $\vec p\in P_+(n_1,\dots,n_d)$.
Let $\ell>0$ such that $\vec\lambda:=\ell\vec p$ is integral.
If $\vec p \in \Delta^B(X)$, then there exists $m>0$ and a highest weight vector $P\in\HWV_{m\vec\lambda^*}(\CC[V]_{(\ell m)})$ such that $P(X)\neq0$.
\end{prp}
\begin{proof}[Sketch of proof]
  Consider $Z := Y^{\ot \ell} \ot v_{\vec\lambda^*} \in W := \Sym^\ell(V)\ot V_{\vec\lambda^*}$.
  By \cref{lem:geometric shift}, the assumption means that $\mu_W([Z])=0$.
  As in \cite[Proof of Theorem 3.2]{burgisser2017alternating}, one can show that $Z$ is a vector of minimal norm in its $G$-orbit.
  Therefore, $0\not\in \overline{G \cdot (Y^{\ot \ell} \ot v_{\vec\lambda^*})}$ and so there exists a homogeneous $G$-invariant polynomial $Q \in \CC[\Sym^\ell(V)\ot V_{\vec\lambda^*}]_{(m)}$ for some $m>0$ such that $Q(Z)\neq0$.
  But then $P(Z) = Q(Z^{\ot k} \ot v_{\vec\lambda^*})$ is a highest weight vector of highest weight $k=\ell m$ such that $P(Y)\neq 0$.
  Since $P$ is an eigenvector of the $B$-action and $[Y]\in\overline{B \cdot [X]}$, it follows that also $P(X)\neq0$.
\end{proof}

\section{Analysis of Algorithm~\ref{alg:scaling}}\label{sec:analysis}
In this section we analyze our tensor scaling algorithm, Algorithm \ref{alg:scaling}. \cref{sec:progress} contains an analysis of the progress made per step (\cref{prp:progress}). \cref{thm:scaling:proof} contains the proof of \cref{thm:scaling}. \cref{subsec:bit_complexity} contains a sketch of bit-complexity analysis of \cref{alg:scaling}. \cref{sec:singular_spectra} contains a reduction from the singular spectra setting to the setting of non-singular spectra.

\subsection{Scaling step}\label{sec:progress}

Consider a single scaling step (\ref{it:borel scale} in \cref{alg:scaling}).
Given a tensor~$Y$ (of unit $\ell_2$ norm) with nonsingular marginals, this amounts to the update
\begin{align*}
 Y \leftarrow Y' := \diag\left(\vec p^{(i)}_{\uparrow}\right)^{1/2} \left(R^{(i)}\right)^{-1} \cdot Y
\end{align*}
where $\rho_Y^{(i)} = R^{(i)} (R^{(i)})^\dagger$ is the Cholesky decomposition of the $i$-th marginal. Here $\vec p^{(i)}_{\uparrow}$ denotes $\left( p^{(i)}_{n_i},\ldots, p^{(i)}_{1}\right)$.

The following proposition shows that the highest weight vectors grow by a constant factor in each scaling step.

\begin{prp}[Progress per step]\label{prp:progress}
  Let $P\in\HWV_{\vec\lambda^*}(\CC[V]_{(k)})$ and $\vec p:=\vec\lambda/k$. Then,
  \begin{align*}
    \lvert P(Y') \rvert \geq 2^{\frac k{32\ln2} \left\lVert \diag\left(\vec p^{(i)}_{\uparrow}\right) - \rho^{(i)} \right\rVert_{\tr}^2} \lvert P(Y) \rvert.
  \end{align*}
\end{prp}
The following crucial lemma is needed for the proof. We delay its proof until after the proof of \cref{prp:progress}.

\begin{lem}\label{lem:generalized lsw}
Let $\rho$ be a PSD $n\times n$-matrix with unit trace such that $\rho=RR^\dagger$, where $R$ is an arbitrary $n\times n$-matrix.
Then, $\vec q=(\lvert R_{1,1}\rvert^2,\dots,\lvert R_{n,n}\rvert^2)$ is a subnormalized probability distribution, and, for any probability distribution~$\vec p$,
\begin{align*}
  D_{KL}(\vec p\Vert\vec q) \geq \frac1{16\ln2} \lVert \diag(\vec p) - \rho \rVert_{\tr}^2.
\end{align*}
where $D_{KL}(\vec p\Vert\vec q) := \sum_{j=1}^n p_j \log_2 (p_j/q_j)$ is the KL-divergence.
\end{lem}
\cref{prp:progress} follows straightforwardly from \cref{lem:generalized lsw}:
\begin{proof}[Proof of \cref{prp:progress}:]
Since $P$ is a highest-weight vector and $\diag(\vec p^{(i)}_\uparrow)^{1/2} (R^{(i)})^{-1}$ is upper-triangular, \cref{eq:hwv poly} shows that
\begin{align*}
  \lvert P(Y') \rvert^2
&= \left( \prod_{j=1}^{n_i} (p^{(i)}_{n_i+1-j})^{kp^{(i)}_{n_i+1-j}} \lvert R^{(i)}_{j,j} \rvert^{-2kp^{(i)}_{n_i+1-j}} \right) \lvert P(Y)\rvert^2\\
&= \left( \prod_{j=1}^{n_i} (p^{(i)}_{n_i+1-j})^{p^{(i)}_{n_i+1-j}} \lvert R^{(i)}_{j,j} \rvert^{-2p^{(i)}_{n_i+1-j}} \right)^k \lvert P(Y)\rvert^2\\
&= 2^{k D_{KL}(\vec p^{(i)}_{\uparrow} \Vert \vec q^{(i)})} \lvert P(Y)\rvert^2
\;\geq\; 2^{\frac k{16\ln2}  \left\lVert \diag\left(\vec p^{(i)}_{\uparrow}\right) - \rho^{(i)} \right\rVert_{\tr}^2} \lvert P(Y)\rvert^2.
\end{align*}
The inequality is \cref{lem:generalized lsw}, which applies because $Y$ is unit norm and hence $\rho^{(i)}$ is unit trace.
\end{proof}

\begin{proof}[Proof of \cref{lem:generalized lsw}:]
  To see that $\vec q$ is subnormalized, observe that
  \begin{align*}
    \sum_{j=1}^n q_j
  = \sum_{j=1}^n \lvert R_{j,j}\rvert^2
  \leq \tr[R R^\dagger] = \tr[\rho] = 1.
  \end{align*}
  On the one hand, we can now apply Pinsker's inequality in the form~\cite[Thm.~10.8.1]{wilde2013quantum}, where the second distribution is allowed to be subnormalized:
  \begin{align*}
    D_{KL}(\vec p\Vert\vec q)
  \geq \frac1{2\ln2} \lVert \vec p - \vec q \rVert_1^2
  = \frac1{2\ln2} \lVert \diag(\vec p) - \diag(\vec q)\rVert_{\tr}^2.
  \end{align*}
  On the other hand,
  \begin{align*}
    D_{KL}(\vec p\Vert\vec q)
  &= D_{KL}(\vec p\Vert\vec q/\lVert\vec q\rVert_1) - \log_2 \lVert\vec q\rVert_1
  \geq - \log_2 \lVert\vec q\rVert_1
  \geq \frac1{\ln2} \left( 1 - \lVert\vec q\rVert_1 \right),
  \end{align*}
  since $\ln x\leq x-1$ for all $x>0$.
  Further,
  \begin{align*}
  \frac1{\ln2} \left( 1 - \lVert\vec q\rVert_1 \right)
  = \frac1{\ln2} \sum_{j\neq k}^n \lvert R_{j,k}\rvert^2
  = \frac1{\ln2} \lVert R - \diag(\vec r) \rVert_F^2
  \geq \frac1{4\ln2} \lVert \rho - \diag(\vec q) \rVert_{\tr}^2,
  \end{align*}
  where $\vec r=(R_{1,1},\dots,R_{,n,n})$ is the diagonal of~$R$.
  In the last step, we used that for any two matrices~$A$ and $B$, $\lVert AA^\dagger - BB^\dagger \rVert_{\tr} \leq \lVert A + B \rVert_F \lVert A - B\rVert_F \leq (\lVert A\rVert_F + \lVert B\rVert_F) \lVert A - B \rVert_F$ (see \cite[Proof of X.2.4]{bhatia2013matrix}).
  Averaging both inequalities, we find that
  \begin{align*}
    D_{KL}(\vec p\Vert\vec q)
  &\geq \frac1{4\ln2} \lVert \diag(\vec p) - \diag(\vec q)\rVert_{\tr}^2 + \frac1{8\ln2} \lVert \rho - \diag(\vec q) \rVert_{\tr}^2\\
  &\geq \frac1{8\ln2} \left( \lVert \diag(\vec p) - \diag(\vec q)\rVert_{\tr}^2 + \lVert \rho - \diag(\vec q) \rVert_{\tr}^2 \right)\\
  &\geq \frac1{16\ln2} \left( \lVert \diag(\vec p) - \diag(\vec q)\rVert_{\tr} + \lVert \rho - \diag(\vec q) \rVert_{\tr} \right)^2\\
  &\geq \frac1{16\ln2} \lVert \diag(\vec p) - \rho \rVert_{\tr}^2. \qedhere
  \end{align*}
\end{proof}

The following lemma can be found in \cite{bhatia2013matrix}.

\begin{lem}[Lemma IV.3.2 in \cite{bhatia2013matrix}]\label{lem:Bhatia} Let $A$ and $B$ be Hermitian matrices. Then
$$
\norm{A-B}_{\tr} \ge \norm{\spec(A) - \spec(B)}_1
$$
Here $\spec(A)$ denotes the vector of eigenvalues of $A$ arranged in decreasing order.
\end{lem}

As a consequence, we get that at the end of \cref{alg:scaling}, the output tensor $Y = g \cdot X$ satisfies
$$
\norm{\spec\left( \rho^{(i)}_Y\right) - \diag\left( \vec p^{(i)}\right)}_1 \le \eps
$$
for all $i=1$ to $d$.

\subsection{Proof of Theorem~\ref{thm:scaling}}\label{thm:scaling:proof}

In this section, we prove \cref{thm:scaling} which we restate for convenience.

\mainthm*

\begin{proof}[Proof of \cref{thm:scaling}]
Let $X\in\Ten(n_0;n_1,\dots,n_d)$ be a tensor whose entries are Gaussian integers, $\vec p\in P_+(n_1,\dots,n_d)$ a rational spectrum such that $p^{(i)}_{n_i}>0$ for all $i=1,\dots,d$, and $\eps>0$.
Assume that $\vec p\in\Delta(X)$.
We need to show that, with probability at least 1/2, \cref{alg:scaling} terminates in step~\ref{it:borel scale} by outputting an appropriate scaling.

In step~\ref{it:randomize} we select a tuple of random matrices~$g$ according to the parameters explained in \cref{cor:random is generic}.

Thus it follows from \cref{cor:random is generic} that, with probability at least 1/2, $g \cdot X\neq0$ and $[g \cdot X] \in \overline{G\cdot [X]}$ and there exists a highest weight vector~$P\in\HWV_{\vec\lambda^*}(\CC[V]_{(k)})$ of degree $0<k\leq K$, where $\vec\lambda=k\vec p$, such that $P(g \cdot X)\neq0$.

We may condition on this event.
By \cref{prp:hwv eval}, we may further assume that $P$ has integer coefficients and that it satisfies the bound
\begin{align}\label{eq:hwv bound}
  \lvert P(Y)\rvert \leq (n_1\dots n_d)^k \lVert Y\rVert
\end{align}
for all tensors $Y\in V$.

For step~\ref{it:full rank}, note that $P(g \cdot X)\neq0$, so $\vec p\in\Delta(g \cdot X)$.
On the other hand, the ranks of the one-body marginals $\rho^{(i)}$ are invariant under scaling by the group action~$G$~\cite{dur2000three,burgisser2017alternating}.
Therefore, since the target spectra~$\vec p^{(i)}$ have full rank by assumption, this means that the $\rho_{g \cdot X}^{(i)}$ necessarily must have full rank also.
It follows that the algorithm does not halt and instead proceeds to step~\ref{it:borel scale}. This also implies that $g$ is not singular and is hence actually \emph{in} $G$, because $\rho_{g \cdot X}^{(i)} = g^{(i)} A g^{(i), \dagger}$ for some matrix $A$, and this would be singular if $g^{(i)}$ were singular.\\

We now move to the scaling step~\ref{it:borel scale}.
Let us denote by $g[t]\in G$ the value of the group element~$g$ at the beginning of the $t$-th iteration, and by $Y[t] := g[t] \cdot X$ the corresponding tensor.
Suppose for sake of finding a contradiction that the algorithm has not terminated after~$T$ steps but instead proceeds to \ref{it:give up}.
We will prove the following three statements:
\begin{itemize}
\item \textbf{Lower bound:} $\lvert P(Y[1]) \rvert \geq 2^{-k(2\sum_{i=0}^d \log_2(n_i) + b + d\log_2(M))}$,
\item \textbf{Progress per step:} $\lvert P(Y[t+1])\rvert > 2^{\frac k{32\ln 2}\eps^2} \lvert P(Y[t])\rvert$ for $t=1,\dots,T$,
\item \textbf{Upper bound:} $\lvert P(Y[t]) \rvert \leq 2^{k\sum_{i=1}^d \log_2(n_i)}$.
\end{itemize}
For the lower bound, note that $g[1] = (g^{(1)} / \lVert g\cdot X\rVert, g^{(2)},\dots,g^{(d)})$ and so $Y[1] := g \cdot X / \norm{ g \cdot X}$.
Now,
\begin{align*}
  \lvert P(Y[1]) \rvert
= \frac{\lvert P(g \cdot X) \rvert}{\lVert g\cdot X\rVert^k}
\geq \frac1{\lVert g\cdot X\rVert^k},
\end{align*}
since $P$ is homogeneous and both $P$ and $g \cdot X$ have integer coefficients only.
On the other hand,
\begin{align*}
  &\lVert g\cdot X\rVert
\leq \lVert X\rVert \prod_{i=1}^d \lVert g^{(i)} \rVert_{\operatorname{op}}
\leq \lVert X\rVert \prod_{i=1}^d \lVert g^{(i)} \rVert_F
\leq \sqrt{2 n_0 n_1 \dots n_d} 2^b \prod_{i=1}^d (n_i M) \\
= &2^{\frac12 + \frac12 \log_2(n_0 n_1 \dots n_d) + b + \log_2(n_1\dots{} n_d) + d\log_2(M)}
\leq 2^{2\sum_{i=0}^d \log_2(n_i) + b + d\log_2(M)},
\end{align*}
where $\lVert \cdot\rVert_{\operatorname{op}}$ denotes the operator norm (assuming $n_0 n_1\cdots{}n_d>1$).
Thus the lower bound follows.
The progress per step follows directly from the analysis in the preceding section (\cref{prp:progress}) and the fact that $Y[t]$ remain unit vectors throughout, which we prove below.
The upper bound also follows from the fact that $Y[t]$ remain unit vectors throughout and \cref{eq:hwv bound}.
The unit norm condition is clear for $Y[1]$, and for $t=1,\dots,T$ we have that
\begin{align*}
&\quad \lVert Y[t+1]\rVert^2
= \tr[\rho^{(i)}_{Y[t]}]\\
&= \tr[\diag(\vec p_\uparrow^{(i)})^{1/2} (R^{(i)})^{-1} \rho^{(i)}_{Y[t-1]} ((R^{(i)})^{-1})^\dagger \diag(\vec p_\uparrow^{(i)})^{1/2}] \\
&= \tr[\diag(\vec p_\uparrow^{(i)})]
= 1,
\end{align*}
where $i$ is the index of the marginal that we selected in the $t$-th scaling step and $R^{(i)}$ the Cholesky factor of $\rho^{(i)}_{Y[t]}$.
Thus we have proved all three statements.
Together, they imply that
$$
\frac{T k}{32\ln 2}\eps^2
< k \left(\sum_{i=1}^d \log_2(n_i) \right) + k \left(2\sum_{i=0}^d \log_2(n_i) + b + d\log_2(M) \right)
$$
The $k$ appears on both sides (which is indeed crucial since $k$ could be exponential in the input parameters!) and we get that
\begin{align*}
\frac{T}{32\ln 2}\eps^2
< \sum_{i=1}^d \log_2(n_i) + 2\sum_{i=0}^d \log_2(n_i) + b + d\log_2(M)
\leq 3\sum_{i=0}^d \log_2(n_i) + b + d\log_2(M)
\end{align*}
and so
\begin{align*}
  T < \frac{32\ln2} {\eps^2} \left( 3\sum_{i=0}^d \log_2(n_i) + b + d\log_2(M) \right),
\end{align*}
which is the desired contradiction.
\end{proof}

\subsection{Bit complexity analysis of Algorithm~\ref{alg:scaling}}\label{subsec:bit_complexity}

In this section, we give a sketch on how to implement \cref{alg:scaling} so that all the intermediate computations are done on numbers with polynomial number of bits. The analysis we provide seems simpler than the analysis in \cite{garg2016deterministic, franks2018operator}. We use notation from the proof of \cref{thm:scaling} in \cref{thm:scaling:proof}. We will maintain the group elements $g[t]$'s. Suppose the algorithm, given $g[t]$, would have chosen to normalize the $i^{\text{th}}$ coordinate. That is, it would have multiplied $A$ to the $i^{\text{th}}$ element of $g[t]$.  We will truncate $A$ to obtain $B$, and set $g[t+1]$ to be $B$ multiplied to the $i^{th}$ coordinate of $g[t]$ (we will also normalize by an appropriate factor close to $1$ to be specified later). The truncation will be done to a fixed number of bits after the binary point.

Suppose we have a tensor $Y = Y[t] = g[t] \cdot X$ (at some point $t$) and a marginal $\rho^{(i)}_Y$ that is at least $\eps$-far from $\diag(\vec p^{(i)}_\uparrow)$.
Assume moreover that we know that $\lambda_{\min}\left(\rho^{(k)}_Y\right) \geq 2^{-C}$ for all $k \in [d]$ (for $C\geq0$) and that $\lVert Y \rVert \leq 1$.

Consider the Cholesky factorization $\rho_Y^{(i)} = R R^\dagger$ ($R$ upper triangular) and define
\begin{align*}
  A := \diag\left(\vec p^{(i)}_\uparrow\right)^{1/2} R^{-1}.
\end{align*}
We will use the notation $\lesssim$ to suppress polynomial factors in $\max_i n_i,d,b$. Note that
\begin{align*}
  \lVert A \rVert^2 \leq \lVert R^{-1} \rVert^2 = \lVert \rho^{-1} \rVert_{\tr} \lesssim \lambda^{-1}_{\min}(\rho) \lesssim 2^C.
\end{align*}
Let $B$ be a truncation of $A$ to $Q$ bits after the binary point to be determined later.
(In total we need $O(n_i^2 (C+Q))$ bits to store $B$).
This means that 
\begin{align*}
  \lVert A - B \rVert \lesssim 2^{-Q},
\end{align*}
and, as a consequence,
\begin{align}\label{eq:relerr}
  \left\lvert \frac {\lvert B_{jj} \rvert} {\lvert A_{jj} \rvert} - 1 \right\rvert
\leq \frac {\lvert B_{jj} - A_{jj} \rvert} {\lvert A_{jj} \rvert}
\lesssim 2^{-Q} \lVert Y \rVert
\leq 2^{-Q},
\end{align}
where we used that
$\lvert A_{jj} \rvert \gtrsim \lvert R_{jj} \rvert^{-1} \geq \lVert R \rVert^{-1} = \lVert Y\rVert^{-1}$.
Now,
\begin{align*}
&\quad \lvert P(B \cdot Y)\rvert^2
= \prod_j \lvert B_{jj} \rvert^{2 k p^{(i)}_{\uparrow,j}} \lvert P(Y)\rvert^2
= 2^{\sum_j 2 k p^{(i)}_{\uparrow,j} \log \lvert B_{jj} \rvert} \lvert P(Y)\rvert^2 \\
&= 2^{\sum_j 2 k p^{(i)}_{\uparrow,j} \log \frac{\lvert B_{jj} \rvert}{\lvert A_{jj} \rvert}} \lvert P(A \cdot Y)\rvert^2
\geq 2^{\sum_j 2 k p^{(i)}_{\uparrow,j} \log \frac{\lvert B_{jj} \rvert}{\lvert A_{jj} \rvert}} 2^{k c \eps^2} \lvert P(Y)\rvert^2
\end{align*}
where $c=1/32\ln2$.
The last step is only applicable if $\norm{Y} \le 1$ (\cref{prp:progress} applies even if $\norm{Y} \le 1$), so we will make sure that that is satisfied throughout.
Assuming $2^{-Q} \leq c' \eps^2/2$, which is $\leq1/2$, \cref{eq:relerr} implies that $\left\lvert \log \frac {\lvert B_{jj} \rvert} {\lvert A_{jj} \rvert} \right\rvert \lesssim c \eps^2/2$, and so
\begin{align*}
  \lvert P(B \cdot Y)\rvert^2 \geq 2^{k c \eps^2/2} \lvert P(Y)\rvert^2,
\end{align*}
so we make progress.

Let's see how the norm changed under this scaling step:
\begin{align*}
  \lVert B \cdot Y\rVert
\leq \lVert B - A \rVert \lVert Y\rVert + \norm{A \cdot Y}
= \lVert B - A \rVert \lVert Y\rVert + 1
\leq 2^{-Q} \lVert Y\rVert + 1 
\leq 2^{-Q} + 1
\end{align*}
However, we wanted the norm to remain $\le 1$ throughout. So we will normalize by a factor of $\kappa = 2^{-Q} + 1$. Then
\begin{align*}
 \bigg\lvert P\left(\frac{1}{\kappa} B \cdot Y\right)\bigg\rvert^2 \geq 2^{k c \eps^2/2 - 2k \log(\kappa)} \lvert P(Y)\rvert^2,
\end{align*}
So as long as $2^{-Q} \ll \eps^2$, we still make progress.

Finally, let's lower-bound the smallest eigenvalue of the new marginals.
Assume first that we applied $A$, not $B$.
Then the $i$-th marginal would be $\diag(\vec p^{(i)}_\uparrow)$, hence $\lambda_{\min}\left( \rho^{(i)}\right)\geq p^{(i)}_{n_i}$, while $\lambda_{\min}$ decreases by a most a factor of $p^{(i)}_{n_i}$ for all other marginals since we have
\begin{align*}
  \lambda_{\min}(A^\dagger A) \geq p^{(i)}_{n_i} \lambda_{\min}\left(\left(\rho_Y^{(i)}\right)^{-1}\right) = p^{(i)}_{n_i}\lambda^{-1}_{\max}\left(\rho_Y^{(i)}\right) \geq \frac {p^{(i)}_{n_i}} {\lVert Y\rVert} \geq p^{(i)}_{n_i}.
\end{align*}
On the other hand,
\begin{align*}
  \lVert B \cdot Y - A \cdot Y \rVert \leq 2^{-Q},
\end{align*}
so the ideal marginals should be $O(2^{-Q})$ close to the real marginals (in particular their eigenvalues).
If we choose $Q$ such that $2^{-Q} \leq 2^{-(C+1)}$, say, then this should mean that, roughly speaking, $\lambda_{\min}\geq 2^{-C-1}$ after the step. The normalization by $\kappa$ also doesn't affect the $\lambda_{\min}$ much. If we run for $T$ iterations, which will turn out to be polynomial by the above analysis, $C$ will remain bounded by $T \poly(\max_i n_i, d,b)$. Thus $Q$ can be chosen to be $O(T \poly(\max_i n_i, d,b))$ and the total bit size remains bounded throughout the algorithm.

\subsection{Singular spectra}\label{sec:singular_spectra}

\begin{dfn}
Let $r_i:=\rk \diag(\vec p^{(i)})$. Define
 $\vec{p}^{(i),+}:= (p_1, \dots, p_{r_i})$, and $\vec{p}_+ = (\vec{p}^{(1),+}, \dots, \vec{p}^{(d),+})$,
 define $X_+ \in \Ten_{n_0; r_1, \dots, r_d} (\CC)$ to be the restriction of $X$ to the coordinates $n_1 - r_j + 1 \leq i_j \leq n_j$ for all $j \in [d]$, and
 define $B_+ = B(r_1) \times \dots \times B(r_d)$.
\end{dfn}
The following lemma shows that to determine scalability of $X$ by $B$ to specified marginals, it is enough to consider scalability of $X_+$ by $B_+$ to the ``positive parts'' of the same marginals.

\begin{lem}\label{lem:singular_spectra}
 $\vec p_+ \in \Delta^{B_+}(X_+)$ if and only if $\vec p \in \Delta^{B}(X)$. Furthermore, if $b_+ \in B_+$ such that $b_+ \cdot X_+$ has marginals that are $\epsilon/2$-close to $\vec p_+$, then by in linear time one can obtain $b \in B$ such that $b \cdot X$ has marginals that are $\epsilon$-close to $\vec p$.
\end{lem}

\begin{proof}
Suppose $\vec p_+ \in \Delta^{B_+}(X)$ and $b_+ \in B_+$ such that $b_+ \cdot X_+$ has marginals that are $\epsilon/2$-close to $\vec p_+$. It will be trivial to obtain $b$ - simply set $b$ to be the block-diagonal matrix $(\delta I_{n_i - r_i} \oplus b_+)$, where $\delta$ is at most, say, $\epsilon^{1/d}/ (4\|X\|)$.

On the other hand, suppose $\vec p \in \Delta^{B}(X)$. By \cref{prp:borel converse} and \cref{prp:hwv eval}, there exists a highest weight vector $P$ of weight $\vec \lambda^*$ (for some $k$ s.t. $\vec \lambda = k \vec p$ is integral) such that $P(X)\neq0$. Since $P$ has weight $\vec \lambda^*$ and $p_{r_i + 1}^{(i)}=  \dots = p_{n_i}^{(i)} = 0$, we have
$$ P(X) = P( (\delta I_{n_i - r_i})\oplus I_{r_i})\cdot X)$$
for all $i$ and all $\delta > 0$. Applying this in succession for $i = 1,\dots d$ and allowing $\delta \to 0$ shows $P(X) = P((X_+)_-)$, where for $Y \in \Ten_{n_0; r_1, \dots, r_d}(\CC)$ we define $(Y)_-$ to be the ``padded'' element of $\Ten_{n_0; n_1, \dots, n_d}(\CC)$ agreeing with $Y$ on coordinates $i_0,i_1, \dots, i_d$ where $n_1 - r_j + 1 \leq i_j \leq n_j$ for all $j \in [d]$ and zero on all other coordinates. Now it's easy to check that $P_+:Y \to P((Y)_-)$ is a highest weight vector of weight $(\vec \lambda_+)^*$ ($\vec \lambda_+ = k \vec p_+$) for the action of $\GL(r_1)\times \dots \times \GL(r_d)$ on $\Ten_{n_0; r_1, \dots, r_d}(\CC)$ and $P_+(X_+) \neq 0$. By \cref{prp:borel}, $\vec p_+ \in \Delta^{B_+}(X_+)$.
\end{proof}
The lemma implies that we can modify \cref{alg:scaling} to prove \cref{thm:main}, namely that there is an efficient scaling algorithm that works even if the target marginals are singular.
\begin{cor}\label{cor:singular_algorithm} \cref{thm:main} is true.
\end{cor}
\begin{proof}
We modify \cref{alg:scaling} as follows: Before \cref{it:full rank}, simply replace $g \cdot X$ by $(g \cdot X)_+$. Perform \cref{it:borel scale} using $X_0 = (g \cdot X)_+$ as the initial tensor and updating a Borel $b_+ \in B_+$ in each step. If ever $b_+\cdot X_0$ is close enough to satisfying the marginal condition, obtain $b$ from $b_+$ as in \cref{lem:singular_spectra} and output $b g$.

We proceed with the analysis. As in the proof of \cref{thm:scaling}, condition on a successful performance of \cref{it:randomize}. That is, condition on having found $g$ such that there exists a highest weight vector $P$ of weight $\lambda^* = k \vec p^*$ degree $k \leq K$ such that $P(g \cdot X) \neq 0$ satisfying $\lvert P(Y)\rvert \leq (n_1\dots n_d)^k \lVert Y\rVert^k$ for all $Y \in V$.

From the proof of \cref{lem:singular_spectra}, from $P$ such that $P(g \cdot X) \neq 0$ we may obtain $P_+$ such that $P_+((g \cdot X)_+) = P_+(X_0) \neq 0$ which is also of degree $k \leq K$ and satisfies the same bound $\lvert P_+(Y)\rvert \leq (n_1\dots n_d)^k \lVert Y\rVert^k$ for all $Y \in \Ten_{n_0; r_1, \dots, r_d}(\CC)$. The degree and evaluation bound follow because we obtained $P_+$ by simply setting some variables of $P$ to zero. The rest of the analysis is the same as that of \cref{thm:scaling}.
\end{proof}

We also get the following corollary (from the proof of \cref{lem:singular_spectra}) relating singular spectra with non-singular spectra.

\begin{cor} It holds that $\vec p \in \Delta(X)$ iff for a generic $g$, $\vec p_+ \in \Delta((g\cdot X)_+)$.
\end{cor}

This is a generalization of the following well known fact: A matrix $M$ (say complex $n \times n$) has rank $\ge r$ iff for generic $U,V$ of  dimensions $r \times n$ and $n \times r$, respectively, it holds that rank of $UMV$ is full.

\section{A reduction to uniform tensor scaling}\label{sec:reductions}

It is interesting to see how the shifting tricks of \cref{subsec:shifting trick} compare with the reduction in \cite{franks2018operator}, which treats the $d = 2$ case of \cref{prb:orbit closure} (also known as \emph{operator scaling})\footnote{the problems are equivalent by an isomorphism between mixed states and completely positive maps known as \emph{state-channel duality} \cite{jamiolkowski1972linear}.}. There, scaling to \emph{any} tensor marginals by $B(n_1)\times B(n_2)$ was reduced to scaling to \emph{uniform} tensor marginals by $\GL(\ell) \times \GL(\ell)$ where $\ell = |\vec{\lambda}|$. In contrast to the shifting trick, the group changes, but the action and the notion of marginal remain the same!
The shifting tricks can be viewed as reductions to uniform case with the same group, but with a different group action and a different notion of marginals: namely the action $g \cdot (X^{\otimes \ell} \otimes v_{\vec{\lambda}^*})= (g\cdot X)^{\otimes \ell} \otimes (g \cdot v_{\vec{\lambda}^*})$ where $\ell = |\vec{\lambda}|$, and tensor marginal replaced by the moment map.\\

The purpose of this section is to verify that the reduction from \cite{franks2018operator} can be fit into the framework of this paper. The conclusion is that the reduction gives the same results in a more ad hoc way - though it may still have conceptual benefits.\\

Here we show how to generalize the reduction from \cite{franks2018operator} to $d \geq 3$. We will use some shorthand: Let $\Lambda^{(i)} = \diag(\vec\lambda_\uparrow^{(i)})$ and
 $$\vec{\Lambda} := \left(\Lambda^{(1)}, \dots, \Lambda^{(d)}\right);$$ similarly for $\vec{P}$, $P^{(i)}$, and $\vec{p}$. The reduction will map $Y \in \Ten(n_0;n_1,\dots,n_d)$ to a tensor $L(Y)$ in the larger space
 $$\Ten(n_0 \lambda^{(1)}_1 \cdots \lambda^{(d)}_1; \ell, \dots, \ell)$$ with the property that there is an element of $\overline{\GL(\ell)^d \cdot L(Y)}$ with uniform marginals if and only if there is an element of $\overline{B \cdot Y}$ with $i^{th}$ marginal equal to $\Lambda^{(i)}$, i.e. $\vec{\lambda}/\ell \in \Delta^B(Y)$.\\

In \cref{subsec:reduction_weights} we will use this map to recreate the construction of highest weight
polynomials from \cref{subsec:hwv} and to give an alternate proof that one of these polynomials does not vanish
if $\vec{\lambda}/\ell \in \Delta^B(Y)$.

\subsection{Properties of the reduction}
Let us see how to create $L(Y)$. Recall from \cref{lem:singular_spectra} in \cref{thm:scaling:proof} that an instance of membership in $\Delta^B(X)$ can be efficiently reduced to another instance of the same problem with non-singular target marginals, so we assume that $p^{(i)}_{n_i} > 0$ for all $i \in [d]$.

For the following discussion we will use the density matrix formalism, so the $\rho$'s that follow play the role of the reduced density matrix $\rho_X$ for a tensor $X$. Let $\rho$ be the density matrix on which we would like to perform the reduction. First, we can forget about the scaling and try to imagine a density matrix $\tilde{\rho}$ on a larger space that has uniform marginals if and only if $\rho^{(i)} = \diag (\vec \lambda_\uparrow^{(i)})$.\footnote{in \cite{franks2018operator} the choice $\rho^{(i)} = \diag (\vec \lambda^{(i)})$ was made instead; this is because in that paper the action was $b \cdot \rho := b^\dagger \rho b$ rather than $b \rho b^\dagger$.}
Our map $\rho \to \tilde{\rho}$ should preserve positivity, so a natural candidate is an \emph{completely positive map}. Recall that a \emph{completely positive map} $T: \operatorname{Mat}_{n\times n} \to \operatorname{Mat}_{m\times m}$ is any map of the form $T:X \mapsto \sum_i A_i X A_i^\dagger$ where $A_i \in \operatorname{Mat}_{m \times n}$. If $T:X \mapsto \sum_i A_i X A_i^\dagger$, then $T^*$ is defined by $T^*: X \mapsto  \sum_i A_i^\dagger X A_i$.

One may try to build the map $T$ in question as a tensor product of completely positive maps, each of which acts on one of the $d$ tensor factors. In order to do this, we would need an injective completely positive map $T_{\vec \lambda}$, depending on a single partition $\vec \lambda$ of $\ell$ with $n$ parts, that satisfies
\begin{align}T_\lambda(\Lambda^{(i)}) = I_\ell \textrm{ and } T_\lambda^*(I_\ell) = I_n\label{eq:reduction}.\end{align}
Let us show why such a family of maps would suffice:

\begin{prp}\label{prp:reduction_marginals}
If the partitions ${\vec \lambda}$ of $\ell$ with $n$ nonzero parts parameterizes a family $T_{\vec \lambda}$ of injective, completely positive maps satisfying \cref{eq:reduction}, then the completely positive map
\begin{align}T = T_{{\vec \lambda}^{(1)}} \otimes \dots \otimes T_{{\vec \lambda}^{(d)}} \label{eq:reduction_tensor}\end{align}
satisfies $T(\rho)^{(i)} = I_\ell$ if and only if $\rho^{(i)}   = \Lambda^{(i)}$ for all $i \in [d]$.
\end{prp}
\begin{proof}
By symmetry, it suffices to prove the proposition only for $i = 1$. Indeed,
\begin{align*}
T(\rho)^{(1)}
&=\left(T_{{\vec {\vec \mu}}^{(1)}} \otimes \dots \otimes T_{{\vec \lambda}^{(d)}} (\rho)\right)^{(1)}\\
&= T_{{\vec \lambda}^{(1)}}\left( \left(I_{n_1} \otimes T_{{\vec \lambda}^{(2)}} \otimes \dots \otimes T_{{\vec \lambda}^{(d)}} (\rho) \right)^{(1)}\right)\\
&= T_{{\vec \lambda}^{(1)}}\left( \left(I_{n_1} \otimes T_{{\vec \lambda}^{(2)}}^*(I_\ell) \otimes \dots \otimes T_{{\vec \lambda}^{(d)}}^*(I_\ell) (\rho) \right)^{(1)}\right)\\
&= T_{{\vec \lambda}^{(1)}} (\rho^{(1)}).
\end{align*}
The second and third equalities follows from the fact that $T$ is a completely positive map and properties of the partial trace. If $T_{\vec \lambda}^{(1)}$ is injective and satisfies \cref{eq:reduction}, then the last line is equal to $I_\ell$ if and only if $\rho^{(1)} =\Lambda^{(1)}$. \end{proof}

Next, we need to show that scalings of $\tilde{\rho}$ correspond to scalings of $\rho$. For this, it is enough to find a group homomorphism $h_{\vec \lambda}:B(n) \to B(\ell)$ satisfying
\begin{align}T_{\vec \lambda}(b \cdot X)  = h_{\vec \lambda}(b) \cdot T_{\vec \lambda} (X)\label{eq:reduction_scaling}  \end{align}
for all $n\times n$ matrices $X$ and all $b \in B(n)$, and as a consequence
\begin{align}T\left(\left(b^{(1)}, \dots, b^{(d)}\right) \cdot \rho\right)  = \left(h_{{\vec \lambda}^{(1)}}\left(b^{(1)}\right), \dots, h_{{\vec \lambda}^{(d)}}\left(b^{(d)}\right)\right) \cdot T(\rho)\label{eq:reduction_scaling_tensor} \end{align}
for all positive semidefinite matrices $\rho$.

\begin{prp}\label{prp:reduction_scalable}
Suppose the partitions ${\vec \lambda}$ of $\ell$ with exactly $n$ nonzero parts parametrize a family of injective, completely positive maps $T_{\vec \lambda}$ satisfying \cref{eq:reduction} and group homomorphisms $h_{\vec \lambda}:B(n) \to B(\ell)$ satisfying \cref{eq:reduction_scaling}. Let $T$ be as in \cref{eq:reduction_tensor}. Let $\vec \lambda = (\vec \lambda^{(1)}, \dots, \vec \lambda^{(d)})$. The following are equivalent:
\begin{enumerate}
\item\label{item:unred_marg} There is an element of $\overline{B \cdot \rho}$ with $i^{th}$ marginal
$\Lambda^{(i)}$ for all $i \in [d]$,
 i.e. $\vec{\lambda}/\ell \in \Delta^B(Y)$ for $\rho = \rho_Y$.
\item\label{item:red_unif} There is an element of $\overline{B_0 \cdot T(\rho)}$ with uniform marginals, where $B_0 = h_{{\vec \lambda}^{(1)}}B(n_{i_1}) \times \dots \times h_{{\vec \lambda}^{(d)}}B(n_{i_d})$.
\item\label{item:red_unif_general} There is an element of $\overline{G_0 \cdot T(\rho)}$ with uniform marginals, where $G_0 = \GL(\ell)\times \dots \times \GL(\ell)$.
\end{enumerate}
\end{prp}
\begin{proof}
We first show
\cref{item:unred_marg} $\iff$ \cref{item:red_unif}. Indeed, by \cref{eq:reduction_scaling_tensor} and the fact that $T$ is an injective, linear map, we have
\begin{align}\overline{B_0 \cdot T(\rho)} = \overline{T(B \cdot \rho)} = T(\overline{B \cdot \rho}).\label{eq:closures}
\end{align}
By \cref{prp:reduction_marginals}, $T(\overline{B \cdot \rho})$ has an element with uniform marginals if and only if $\overline{B \cdot \rho}$ has an element with $i^{th}$ marginal equal to $\Lambda^{(i)}$ for all $i \in [d]$.

Next we show \cref{item:red_unif_general} $\iff$ \cref{item:red_unif}. Clearly \cref{item:red_unif} $\implies$ \cref{item:red_unif_general}. If \cref{item:red_unif} holds, then $T(\rho)$ is not in the null-cone of the action of the action of $G_0$, so by \cite{burgisser2017alternating}, Sinkhorn style scaling of $T(\rho)$ converges to a mixed state with uniform marginals. More precisely, there is a sequence $({i_t}:t \geq 0)$ such that the sequence of density matrices defined by $$\rho(t + 1) = (I_{n_1} \otimes \dots \otimes I_{n_{{i_t}} - 1} \otimes g(t) \otimes I_{n_{{i_t}} + 1} \otimes \dots \otimes I_{n_d})\cdot \rho(t)$$ for $t \geq 0$ and $\rho(0) = T(\rho)$ converges to a mixed state with uniform marginals provided
\begin{align}g(t) \rho(t)^{(i_t)} g(t)^\dagger = I_{\ell}\label{eq:sinkhorn}\end{align}
for all $t \geq 0$. However, we are lucky, and we may choose each of our scalings $g(t)$ to be $h_{{\vec \lambda}^{({i_t})}}({b_t})$ for some $b \in B(n_{{i_t}})$! Indeed, suppose inductively that $s \geq 1$ and for all $t \leq s -1$ there exists ${b_t}$ such that $g(t) = h_{{\vec \lambda}^{(i)}}({b_t})$ satisfies \cref{eq:sinkhorn}. Then by \cref{eq:reduction_scaling_tensor} and group homomorphism property of $h_{\vec \lambda}$, we have $\rho(t) = T(\underline{\rho})$ for some $\underline{\rho}$. In particular, $T(\underline{\rho})^{(i_t)} = T_{{\vec \lambda}^{({i_t})}}(\underline{\rho}^{(i)}).$ Take ${b_t}$ upper triangular such that $${b_t}\underline{\rho}^{(i_t)} {b_t}^\dagger = \Lambda^{(i_t)}.\footnote{this is precisely the update step in \cref{alg:scaling}!}$$
which is always possible due to the existence of the Cholesky decomposition. By \cref{eq:reduction_scaling} we have
$$ h_{{\vec \lambda}^{({i_t})}}({b_t}) \cdot  \rho(t)^{(i_t)} = T_{{\vec \lambda}^{({i_t})}}\left({b_t} \cdot \underline{\rho}^{(i_t)}\right) = I_\ell.$$
By induction, $g(t) = h_{{\vec \lambda}^{({i_t})}}({b_t})$ satisfies \cref{eq:sinkhorn} for all $t \geq 0$, so $\rho(t)$ is a sequence of elements of $H_0 \cdot T(\rho)$ converging to a density matrix with uniform marginals, and hence \cref{item:red_unif} holds.
\end{proof}

\begin{rem}
The above proof also shows that \cref{alg:scaling} converges to a tensor with the appropriate marginals if $\vec{ \lambda}/\ell \in \Delta^B(Y)$: each scaling step of \cref{alg:scaling} is exactly a step of the scaling algorithm from \cite{burgisser2017alternating} applied to $L(Y)$, which we now know converges to a tensor with uniform marginals if $\vec{\lambda}/\ell \in \Delta^B(Y)$.
\end{rem}

\subsection{The construction}
So far, \cref{prp:reduction_scalable} may be vacuously true. That is, we have not yet proven that the families $T_{\vec \lambda}$ and $h_{\vec \lambda}$ exist. We do this here by computing a reduction $L_{\vec \lambda}$ between pure tensors and taking the partial trace. In what follows, the only nontrivial part is the guess for $L_{\vec \lambda}$. All else is elementary linear algebra. $L_{\vec \lambda}$ is this is very similar to \cite{franks2018operator}, but firstly it is a map between tensors rather than density matrices, and secondly it commutes with the action of the upper triangular rather than lower triangular matrices (hence the choice to use projections to the \emph{last} few coordinates).

Let $\nu_j:\CC^n \to \CC^j$ denote the projection to the last $j$ coordinates in some fixed orthonormal basis for $\CC^n$. In that basis,
$$ \begin{array}{ccc} & n & \\
 \nu_{j} = & \left[ \begin{array}{ccc ccc}0 &  \hdots & 0 & 1 & \hdots & 0  \\
\vdots & \ddots & \vdots & \vdots & \ddots & \vdots \\
0 &  \hdots & 0 & 0 & \hdots & 1  \\
 \end{array}\right] & j.\end{array} $$
The dependence of $\nu_j$ on $n$ and the basis will be suppressed.
\begin{dfn}[Reduction, $d = 1$] Suppose ${\vec \lambda}$ is a partition of $\ell$ with at most $n$ nonzero parts. Let ${\vec \mu}$ be the conjugate partition to ${\vec \lambda}$.
For $i \in [\lambda_1]$, define
$$ \tau_i^{{\vec \lambda}}:\CC^n \to \bigoplus_{j \in \lambda_1} \CC^{\mu_j}  = \CC^\ell$$ by
\begin{align*}  \tau_j^{{\vec \lambda}}x  &= (\underbrace{0, \dots, 0}_{\mu_1}, \cdots, \underbrace{0, \dots, 0}_{\mu_{i - 1}}, \nu_{\mu_i} x, \underbrace{0, \dots, 0}_{\mu_{i +1 }},   \dots , \underbrace{0 \dots 0}_{\mu_k})\\
&= (\underbrace{0, \dots, 0}_{\mu_1}, \cdots, \underbrace{0, \dots, 0}_{\mu_{i - 1}}, \underbrace{x_{n-\mu_i + 1}, \dots, x_n}_{\mu_i}, \underbrace{0, \dots, 0}_{\mu_{i +1 }},   \dots , \underbrace{0 \dots 0}_{\mu_k}).
\end{align*}
Now define $$L_{{\vec \lambda}}: \CC^n \to \CC^{\lambda_1} \otimes \CC^{\ell}$$ by
$$ L_{{\vec \lambda}} v = \sum_{j \in [\lambda_1]} e_j \otimes (\tau_j^{\vec \lambda} v), $$
where $(e_j: j \in [\lambda_1])$ is some fixed orthonormal basis of $\CC^{\lambda_1}$.
As a matrix in the same basis used to define $\tau^{\vec \lambda}$,
$$L _{\lambda} v = \left[\begin{array}{cccc} v_{n - \mu_1 +1} & 0 & \hdots & 0 \\
\vdots & 0 & \hdots & 0 \\
v_{n} & 0 & \hdots & 0 \\
 0 & v_{n - \mu_2 + 1} & \hdots & 0 \\
0 & \vdots & \hdots & 0 \\
0 & v_{n} & \hdots & 0 \\
\vdots & \vdots & \ddots & \vdots \\
0 & 0 & \hdots & v_{n - \mu_k + 1} \\
0 & 0 & \hdots & \vdots \\
0 & 0 & \hdots & v_{n} \\
\end{array}\right].$$
Note that the $i^{th}$ column of $L_{{\vec \lambda}} v$ is $\tau^\lambda_i v$.
\end{dfn}

\begin{prp}
Suppose ${\vec \lambda}$ is a partition of $\ell$ with exactly $n$ nonzero parts and $\Lambda = \diag(\vec \lambda_\uparrow)$. Define
\begin{align*}&\underline{T}_{\vec \lambda}:X \mapsto \tr_{\CC^{\lambda_1}} \left[ L_{\vec \lambda}  X L_{\vec \lambda}^\dagger \right],  \textrm{ and} \\
 &T_{\vec \lambda}:X \mapsto \underline{T}_{\vec \lambda}({\Lambda}^{-1/2} X {\Lambda}^{-1/2}).
\end{align*}
Then $T_{\vec \lambda}$ is injective, completely positive, and satisfies \cref{eq:reduction}, and the map $h_{\vec \lambda}:B(n) \to B(\ell)$ given by
$$ h_{\vec \lambda}: b \mapsto \underline{T}_{\vec \lambda} ({\Lambda}^{-1/2} b {\Lambda}^{1/2})$$
is a group homomorphism and satisfies \cref{eq:reduction_scaling}.\footnote{$\underline{T}_{\vec \lambda}$ is almost the reduction from \cite{franks2018operator}, but a change of variables makes the presentation simpler.}

\end{prp}

\begin{proof}
We expand the expression for $\underline{T}_{\vec \lambda}$.
\begin{align} \tr_{\CC^{\lambda_1}} \left[ L_{\vec \lambda} X  L_{\vec \lambda}^\dagger \right]
&= \tr_{\CC^{\lambda_1}} \left[ \sum_{i = 1}^{\lambda_1} \sum_{j = 1}^{\lambda_1} e_i e_j^\dagger \otimes  \tau^{{\vec \lambda}}_i  X  \tau^{{\vec \lambda}, \dagger}_j \right] \nonumber\\
&=  \sum_{j = 1}^{\lambda_1}   \tau^{{\vec \lambda}}_j  X  \tau^{{\vec \lambda}, \dagger}_j.\label{eq:kraus}
\end{align}
The last line shows explicitly that $\underline{T}_{\vec \lambda}$ (and hence $T_{\vec \lambda}$) is a completely positive map. To see that $\underline{T}_{\vec \lambda}$ (and hence $T_{\vec \lambda}$) is injective, observe that $ \tau^{{\vec \lambda}}_1 \underline{T}_{\vec \lambda}(X) \tau^{\lambda, \dagger}_1  = X$. This holds because ${\vec \lambda}$ has $n$ nonzero parts, and so $\nu_{\mu_1} = \nu_n =  I_n$.\\

Here we show $T_{\vec \lambda}$ satisfies \cref{eq:reduction}. This follows from the below pair of equations: \begin{align*}\underline{T}_{\vec \lambda}(I) &= \sum_{i =1}^{\lambda_1} \tau_j^{\vec \lambda} \tau_j^{{\vec \lambda}, \dagger} = I_\ell, \textrm{ and  }\\
\underline{T}^*_{\vec \lambda}(I_n) &= \sum_{i =1}^{\lambda_1} \tau_j^{{\vec \lambda},\dagger} \tau_j^{{\vec \lambda}} = \sum_{j = 1}^{\lambda_1} \nu_{\mu_j}^\dagger \nu_{\mu_j} = {\Lambda}.
\end{align*}
The first equality in both lines is from \cref{eq:kraus}.The third equality in the second line is our only use of the fact that ${\vec \mu}$ is the conjugate partition of ${\vec \lambda}$.\\

Next we must show $h_{\vec \lambda}: b \mapsto \underline{T}_{\vec \lambda} ({\Lambda}^{-1/2} b {\Lambda}^{1/2})$ is a group homomorphism. Because conjugation by $\sqrt{{\Lambda}}$ is a group homomomorphism, it is enough to show $ b \mapsto \underline{T}_{\vec \lambda}(b)$ is a group homomorphism.
We'll show something even more useful, namely
\begin{align} L_{\vec \lambda} b = (I_{\lambda_1} \otimes \underline{T}_{\vec \lambda}(b)) L_{\vec \lambda}.\label{eq:b-linear}
\end{align}
Indeed,
\begin{align*}
(I_{\lambda_1} \otimes \underline{T}_{\vec \lambda}(b)) L_{\vec \lambda} v &= \left(I_{\lambda_1} \otimes \sum_j \tau^{\vec \lambda}_j b \tau^{{\vec \lambda}, \dagger}_j\right) L_{\vec \lambda} v\\
&= \sum_{i = 1}^{\lambda_1} e_i \otimes \sum_j \tau^{\vec \lambda}_j b \tau^{{\vec \lambda}, \dagger}_j \tau^{\vec \lambda}_i v\\
&= \sum_{i = 1}^{\lambda_1} e_i \otimes \tau^{\vec \lambda}_i \nu^{ \dagger}_{\mu_i} \nu_{\mu_i} b \nu^{ \dagger}_{\mu_i} \nu_{\mu_i} v\\
&= \sum_{i = 1}^{{\vec \lambda}_1} e_i \otimes \tau^{\vec \lambda}_i  b  v.
\end{align*}
The second-to-last equality uses the easy facts that $\tau^{{\vec \lambda}, \dagger}_j \tau^{\vec \lambda}_i = \delta_{ij} \nu^{ \dagger}_{\mu_i} \nu_{\mu_i}$ and $\tau^{{\vec \lambda}}_i = \tau^{{\vec \lambda}}_i \nu^{ \dagger}_{\mu_i} \nu_{\mu_i}$.  The last equality uses the simple identity $\nu_{i}b \nu^{ \dagger}_{i} \nu_{i} = \nu_i b$ for all $b \in B(n)$, $i \in [n]$.  From \cref{eq:b-linear}, we have
$$ \underline{T}_{\vec \lambda}(bX) = \tr_{\CC^{\lambda_1}} L_{\vec \lambda} bX  L_{\vec \lambda}^\dagger = \tr_{\CC^{\lambda_1}}(I_{\lambda_1} \otimes \underline{T}_{\vec \lambda}(b)) L_{\vec \lambda} X  L_{\vec \lambda}^\dagger =\underline{T}_{\vec \lambda}(b) \underline{T}_{\vec \lambda}(X)$$
for $b \in B(n)$ and \emph{any} $n\times n$ matrix $X$. \cref{eq:b-linear} means $L_{\vec \lambda}$ is a $B(n)$-linear map, or \emph{intertwiner}, between the representations $b \to b$ and $b \to I_{\lambda_1} \otimes \underline{T}_{\vec \lambda}(b)$ of $B(n)$!\\

 It remains to show \cref{eq:reduction_scaling}. Something a bit stronger follows from $T_{\vec \lambda}(bX)  = T_{\vec \lambda}(b) T_{\vec \lambda}(X)$:
$$\underline{T}_{\vec \lambda}({\Lambda}^{-1/2} b {\Lambda}^{1/2}) \underline{T}_{\vec \lambda}({\Lambda}^{-1/2} X {\Lambda}^{-1/2}) = \underline{T}_{\vec \lambda}({\Lambda}^{-1/2} bX {\Lambda}^{-1/2}).$$
Equivalently, $h_{\vec \lambda}(b)T_{\vec \lambda}(X) = T_{\vec \lambda}(b X)$. \end{proof}

We can use the map $L_{\vec \lambda}$ to phrase \cref{prp:reduction_scalable} in terms of the null-cone.
\begin{dfn}[Reduction between pure tensors]
Define
$$L(Y) = (I_{n_0} \otimes L_{{\vec \lambda}^{(1)}} \otimes \dots \otimes L_{{\vec \lambda}^{(d)}})  Y.$$
\end{dfn}
 First note that $L( Y) \in \operatorname{Ten}_{n_0;  \lambda_1^{(1)}, \ell, \dots, \lambda_1^{(d)}, \ell}(\CC)$. By reorganizing, we'll think of it as an element of $\Ten_{n'_0;  \ell, \dots, \ell}(\CC)$ where $n'_0 = n_0\lambda_1^{(1)} \cdots  \lambda_1^{(d)}$. We then allow $\SL(\ell)^d$ to act on all of the $d$ tensor factors of $L(X)$ of dimension $\ell$. From \cite{burgisser2017alternating}, we have that
$$ L\left(\vec{\Lambda}^{-1/2}\cdot Y\right)$$
is outside the null cone if and only if its reduced density matrix (tracing out $\CC^{n'_0}$) is scalable to uniform marginals. However, by our definition of $T_{\vec \lambda}$, the reduced density matrix is precisely $T (\tr_{\CC^{n_0}} YY^\dagger)$ of \cref{prp:reduction_scalable}! By \cref{prp:reduction_scalable}, $L\left(\vec{\Lambda}^{-1/2}\cdot Y\right)$ is outside the null-cone if and only if $\vec{\lambda}/\ell \in \Delta^B(Y)$. Further, by \cref{eq:b-linear}, $L\left(\vec{\Lambda}^{-1/2}\cdot Y\right)$ is in the null-cone of the action of $\SL(\ell)^d$ if and only if $L(Y)$ is. The preceding reasoning gives us yet another reduction to the null-cone problem:
\begin{cor}[Reduction to tensor scaling null cone]\label{cor:pure_reduction}
We have $\vec{\lambda}/\ell \in \Delta^B(Y)$ if and only if
$L(Y)$ is outside the null-cone of the action of $\SL(\ell)^d$.
\end{cor}

\begin{rem}[Capacity]\label{rem:capacity}
Let us look at the familiar, and easy to prove, formula for the determinant after applying the reduction:
\begin{align}
\det \underline{T}_{\vec \lambda}(b^{-1}) = \chi_{{\vec \lambda}^*}(b).\label{eq:reduction_determinant}
\end{align}
Using this, we can show that the capacity from \cref{dfn:capacity} given by
$$\capacity_{\vec{\lambda}}(X) =  \inf_{R \in B} \|R \cdot X\| |\chi_{\vec{\lambda}^*}(R)|$$
is a natural choice of capacity. \cref{prp:reduction_scalable} implies $\SL(\ell)^d$ scaling of $L(Y)$ to uniform marginals, if it is possible, can be performed by scalings of the form $(\underline{T}_{{\vec \lambda}^{(1)}}(b_1), \dots , \underline{T}_{{\vec \lambda}^{(d)}}(b_d))$ for $b_i$ in $B(n_i)$. For short, denote this element $\underline{T}_{\vec{\lambda}}(b)$. Thus, $L(Y)$ is scalable to uniform marginals if and only if
\begin{align*}
0 &< \inf_{b: \det( \underline{T}_{\lambda^{(i)}}(b_i)) = 1} \|\underline{T}_{\vec{\lambda}}(b)\cdot L(X)\|^2\\
&= \inf_{b: \chi_{\vec{\lambda}^*}(b^{-1}) = 1} \|L(b \cdot X)\|^2\\
& = \inf_{b \in B} \|\sqrt{\vec{\Lambda}} \cdot b \cdot X\|^2 |\chi_{\vec{\lambda}^*}(b)|^{2/\ell}\\
& = |\chi_{\vec{\lambda}^*}(\vec{\Lambda}^{-1/2})|^{2/\ell} \inf_{b \in B} \| b \cdot X\|^2 |\chi_{\vec{\lambda}^*}(b)|^{2/\ell}.
\end{align*}
Up to a constant, and a power of $\ell$, this matches $\capacity(X)$.
\end{rem}

\subsection{Highest weights from the reduction}\label{subsec:reduction_weights}
The map $L: \Ten_{n_0;  n_1, \dots, n_d}(\CC) \to \Ten_{n'_0;  \ell, \dots, \ell}(\CC)$ can be viewed as an intermediate step in the classical construction of highest weights. We would like to show that if $Y$ is Borel-scalable to an element with the appropriate marginals, then some highest weight with bounded integer coefficients is nonvanishing on $Y$. Here we show that composing the homogeneous $\SL(n)^d$-invariant polynomials used in \cite{burgisser2017alternating} with $L$ yields a subset of the highest weight vectors defined in \cref{eq:hwv concrete}! This amounts to an alternate proof that one of the polynomials in \cref{eq:hwv concrete} is nonzero at $Y$ if ${\vec \lambda}/\ell \in \Delta^B(Y)$. \\

We may start by computing a homogeneous, $\SL(\ell)^d$-invariant polynomial on $L(Y)$. By \cref{cor:pure_reduction} and \cite{burgisser2017alternating}, if ${\vec \lambda}/\ell \in \Delta^B(Y)$, then some $\SL(\ell)^d$-invariant, homogeneous polynomial does not vanish on $L( Y)$. Further, for any $\SL(\ell)^d$-invariant, homogeneous polynomial $Q$ on $\Ten_{n'_0;  \ell, \dots, \ell}(\CC)$ of degree $m\ell$, the polynomial
$$Y \mapsto Q(L (Y))$$
 is a highest weight vector of weight $m \vec{\lambda}^*$.
Indeed, using \cref{eq:b-linear} and \cref{eq:reduction_determinant}.
  \begin{align*}
 b \cdot Q(L (Y)) &= Q((I_{n_0} \otimes L_{{\vec \lambda}^{(1)}}b_1^{-1} \otimes \dots \otimes L_{{\vec \lambda}^{(d)}}b_d^{-1}) Y )\\
 & = Q( (\underline{T}_{{\vec \lambda}^{(1)}}(b_1^{-1}), \dots, \underline{T}_{{\vec \lambda}^{(d)}}(b_d^{-1}))\cdot L (Y))\\
 & = \prod_{i = 1}^d \det(\underline{T}_{{\vec \lambda}^{(i)}}(b_i^{-1}))^{\ell m/\ell} Q(  L(Y))\\
 & = \chi_{m\vec{\lambda}^*}(b)P(X).
 \end{align*}

 By \cite{burgisser2017alternating}, the $\SL(\ell)^d$-invariant, homogeneous polynomials of degree $\ell m$ on $\Ten_{n'_0;  \ell, \dots, \ell}(\CC)$ (all of them are $0$ unless $\ell$ divides the degree) are spanned by polynomials of the form
$$P(X) =   (\varepsilon_{\vec{i}^{(0)}} \otimes \varepsilon_{\vec{i}^{(1)}} \otimes  \dots \otimes\varepsilon_{\vec{i}^{(d)}}\otimes p_{\ell m,\pi_1}\otimes \dots \otimes p_{\ell m,\pi_d} ) (X^{\otimes \ell m})$$
where $\vec{i}^{(j)}:[\ell m] \to [ \lambda_1^{(j)}],$ $\vec{i}^{(0)}:[\ell m] \to [n_0]$, and $\varepsilon_{\vec{i}}$ denotes the linear form $\varepsilon_{\vec{i}}(e_{j(1)} \otimes \dots \otimes e_{j(t)}) = \delta_{\vec{i}\vec{j}}$. Here $\pi$ denotes a permutation of $[\ell m]$ and
$$p_{\ell m,\pi} (v_1 \otimes \dots \otimes v_{\ell m}) = \det (v_{\pi(1)},  \dots , v_{\pi(\ell)}) \det(v_{\pi(\ell + 1)} , \dots , v_{\pi(2\ell)}) \dots  \det (v_{\pi(\ell m - \ell + 1)},  \dots , v_{\pi(\ell m)}).$$

\begin{prp}$P(L(Y))$ vanishes unless for all $i \in [d], k \in [m]$, the sequence $i^{(j)}(\pi_j(\ell k +1)), \dots, i^{(j)}(\pi_j(\ell k))$ contains precisely $\mu_1^{(j)}$ many $1's$, $\mu_2^{(j)}$ many $2's$, $\dots$, $\mu_{\lambda_1^{(j)}}^{(j)}$ many $\lambda_1^{(j)}$'s. If this occurs, then
$$P(L(Y)) = \pm \left( \varepsilon_{\vec{i}^{0}} \otimes \Det_{m{\vec \lambda}^{(1), *}, \pi'_1} \otimes \dots \otimes \Det_{m{\vec \lambda}^{(d), *}, \pi'_d}\right)(Y^{\otimes \ell m}) $$
for some permutations $\pi'_1, \dots, \pi'_d$ of $[\ell m]$.
\end{prp}
\begin{proof}
It is enough to compute the linear form
$$p(\underline{L}(X)^{\otimes m}) = (\varepsilon_{\vec{i}^{(1)}} \otimes  \dots \otimes\varepsilon_{\vec{i}^{(d)}}\otimes p_{\ell m,\pi_1}\otimes \dots \otimes p_{\ell m,\pi_d} ) (\underline{L}(X)^{\otimes m}) $$
at where $\underline{L} = L_{{\vec \lambda}^{(1)}} \otimes \dots \otimes  L_{{\vec \lambda}^{(d)}}$ and $X \in \Ten(n_1, \dots, n_d)$.
We compute $p(\underline{L}(X)^{\otimes \ell m}) = p(\underline{L}^{\otimes \ell m} (X^{\otimes \ell m}))$ on a spanning set of $\Ten( n_1, \dots, n_d)^{\otimes \ell m}$ given by
$$ Z =\left( v^{(1)}_1 \otimes \dots \otimes v^{(1)}_{\ell m} \right) \otimes \dots \otimes \left( v^{(d)}_1 \otimes \dots \otimes v^{(d)}_{\ell m} \right)$$
where $v^{(i)}_j $ ranges over $\CC^{n_i}$. Applying $\underline{L}^{\otimes \ell m}$, we have

$$ \underline{L}^{\otimes \ell m} Z =  L_{{\vec \lambda}^{(1)}}^{\otimes \ell m} \left(  v^{(1)}_1 \otimes \dots \otimes v^{(1)}_{\ell m} \right) \otimes \dots \otimes L_{{\vec \lambda}^{(d)}}^{\otimes \ell m}\left( v^{(d)}_1 \otimes \dots \otimes v^{(d)}_{\ell m} \right).$$
Hence,
$$p(L^{\otimes \ell m} Z) = \prod_{k = 1}^d ( \varepsilon_{\vec{i}^{(k)}} \otimes p_{\ell m, \pi_k}) \left( L_{{\vec \lambda}^{(k)}}^{\otimes \ell m} \left(  v^{(k)}_1 \otimes \dots \otimes v^{(k)}_{\ell m} \right)\right).$$
It remains to compute the value of $ (\varepsilon_{\vec{i}} \otimes p_{\ell m, \pi}) \left( L_{{\vec \lambda}}^{\otimes \ell m} \left(  v_1 \otimes \dots \otimes v_{\ell m} \right)\right)$ for $\vec{i}:[\ell m] \to [\lambda_1]$. In fact,
\begin{align*}
(\varepsilon_{\vec{i}} \otimes p_{\ell m, \pi}) \left( L_{{\vec \lambda}}^{\otimes \ell m} \left(  v_1 \otimes \dots \otimes v_{\ell m} \right)\right)\\
= (\varepsilon_{\vec{i}} \otimes p_{\ell m, \pi}) \left(\sum_{j \in [\lambda_1]} e_j \otimes \tau_j^{\vec \lambda} v_1 \right) \otimes \dots \otimes \left(\sum_{j \in  [\lambda_1]} e_j \otimes \tau_j^{\vec \lambda} v_{\ell m} \right)\\
= p_{\ell m, \pi} \left( \tau_{i(1)}^{\vec \lambda} v_1  \otimes \dots \otimes  \tau_{i(\ell m)}^{\vec \lambda} v_{\ell m} \right).
\end{align*}
Suppose, without loss of generality, that $\pi$ is the identity permutation. Notice that if the sequence $i(1), \dots, i(\ell)$ does not contain precisely $\mu_1$ many $1's$, $\mu_2$ many $2's$, and so on, the above expression will vanish. Otherwise, it will be equal to $\pm\Det_{{\vec \lambda}^*, \pi'} (v_1 \otimes \dots \otimes v_{\ell})$ where $\pi'$ is a permutation such that $i(\pi(j))$ is decreasing for $j \in [\ell]$. Applying similar reasoning for $i(\ell+1), \dots, i(2\ell - 1)$ and so on whilst combining the polynomials using $\Det_{{\vec \lambda}, \pi'} \otimes \Det_{{\vec \lambda}', \pi''} = \Det_{{\vec \lambda} + {\vec \lambda}', \pi'''}$ completes the proof. \end{proof}

\section{Distance lower bound}\label{sec:distance}

In this section, we will show that if $\vec p=\vec\lambda/\ell$ is not contained in the moment polytope~$\Delta(\mathcal X)$, then its distance to the moment polytope can be lower bounded only in terms of~$\ell$ and the dimensions $n_0,n_1,\dots,n_d$ -- independently of~$\mathcal X$. The high level strategy is as follows. We first lower bound (\cref{prp:distance-to-gap}) the distance in terms of something called the gap constant (\cref{dfn:gap-constant}). Then we lower bound the gap constant (\cref{lem:lb_distance_0} and \cref{lem:gap_constant_lb}) using duality (Farkas' lemma) and well known bounds on solutions to linear programs.

We will establish the lower bound by studying the geometry of weights underlying the representation that underlies the shifting trick (\cref{lem:geometric shift}):

\begin{dfn}[Gap constant]\label{dfn:gap-constant}
  Let $\vec\lambda$ be a highest weight and $\ell=\lvert\vec\lambda\rvert$.
  We define the \emph{gap constant} by
  \begin{align*}
    \gamma(\vec\lambda) := \min \left\{ \frac{\lVert\vec x\rVert_2}\ell : S \subseteq \Omega(\Sym^\ell(V)\ot V_{\vec\lambda^*}), 0\not\in\conv(S), \vec x\in\conv(S) \right\},
  \end{align*}
  where we recall that $\Omega(W)$ denotes the set of weighs that occur in a representation~$W$ (see \cref{subsec:hwtheory}).
\end{dfn}

The argument in the following proof is essentially from~\cite{kirwan1984cohomology}.

\begin{prp}\label{prp:distance-to-gap}
  Let $X\in\Ten(n_0;n_1,\dots,n_d)$ be a nonzero tensor, $\vec\lambda$ be a highest weight, and $\ell=\lvert\vec\lambda\rvert$.
  If $\vec p:=\vec\lambda/\ell\not\in\Delta(\mathcal X)$ then
  \begin{align*}
    \min \{ \lVert \vec q - \vec p \rVert_2 : \vec q \in \Delta(\mathcal X) \} \geq \gamma(\vec\lambda).
  \end{align*}
\end{prp}
\begin{proof}
  It suffices to show that, for all $[X]\in\mathcal X$ and $U\in K$,
  \begin{align*}
    \sum_{i=1}^d \lVert \mu^{(i)}([X]) - U^{(i)} \diag(p^{(i)}_{n_i},\dots,p^{(i)}_1) (U^{(i)})^\dagger \rVert_F^2 \geq \gamma^2(\vec\lambda).
  \end{align*}
  By the shifting trick, \cref{eq:geometric shift}, this is equivalent to
  \begin{align}\label{eq:gap goal}
    \frac1{\ell^2} \sum_{i=1}^d \lVert \mu^{(i)}_W([Z])\rVert^2_F \geq \gamma^2(\vec\lambda),
  \end{align}
  where $Z := X^{\ot \ell} \ot U \cdot v_{\vec\lambda^*} \in W := \Sym^\ell(V)\ot V_{\vec\lambda^*}$.
  By assumption, $\vec\lambda\not\in\Delta(X)\subseteq\Delta(\mathcal X)$.
  According to \cref{lem:invariant-theoretic shift,thm:mumford,}, this means that any $G$-invariant polynomial vanishes on $Z$, so $0 \in \overline{G \cdot Z}$.
  By the Hilbert-Mumford criterion, there exists a 1-parameter subgroup of the form $\exp(At)$, where $A=(A^{(1)},\dots,A^{(d)}))$ is a tuple of Hermitian matrices, such that $\exp(At) \cdot Z \to 0$ for $t\to\infty$.

  Without loss of generality we may assume that the $A^{(i)}$ are diagonal matrices for $i=1,\dots,d$ (otherwise conjugate each by an appropriate unitary $Y^{(i)}$, and replace $X$ by $(U^{(1)}\ot\dots\ot U^{(d)})$, which leaves the left-hand side of \cref{eq:gap goal} invariant).
  Thus, $\exp(At)\in T$ for all~$t\in\CC$.
  We now expand~$Z$ in terms of weight vectors, $Z = \sum_{\vec\omega\in\Omega(W)} Z_{\vec\omega}$, so that
  \begin{align*}
    \exp(At) \cdot Z = \sum_{\vec\omega\in\Omega(W)} e^{\sum_{i=1}^d \vec\omega^{(i)} \cdot \vec a^{(i)}} Z_{\vec\omega},
  \end{align*}
  where we write $\vec a^{(i)}$ for the diagonal entries of $A^{(i)}$, $i=1,\dots,d$;
  the dot $\cdot$ in the exponent denotes the standard inner product on $\RR^{n_i}$.
  Since we know that $\exp(At) \cdot Z\to0$, it follows that $\sum_{i=1}^d \vec\omega^{(i)} \cdot \vec a^{(i)} < 0$ whenever $Z_{\vec\omega}\neq0$.
  This implies that $0 \not\in \conv(S)$, where $S := \{\omega : Z_{\vec\omega}\neq 0 \}$.

  On the other hand, note that by the definition of the moment map (\cref{eq:general moment map}) and the orthogonality of the weight space decomposition, we have
  \begin{align*}
    \tr[\mu^{(i)}_W([Z]) \diag(\vec b^{(i)})]
  &= \partial_{t=0} \frac{\braket{Z, \exp(\diag(\vec b^{(i)})t) \cdot Z}}{\braket{Z, Z}}
  = \partial_{t=0} \sum_{\vec\omega\in S} \frac{\lvert Z_{\vec\omega}\rvert^2}{\lVert Z\rVert^2} e^{t \vec\omega^{(i)} \cdot \vec b^{(i)}}\\
  &= \sum_{\vec\omega\in S} \frac{\lvert Z_{\vec\omega}\rvert^2}{\lVert Z\rVert^2} (\vec\omega^{(i)} \cdot \vec b^{(i)})
  \end{align*}
  for all $\vec b^{(i)}\in\RR^{n_i}$.
  This implies that if we orthogonally project each component of~$\mu_W([Z])$ onto the diagonal, we obtain $\sum_{\vec\omega\in S} \frac{\lvert Z_{\vec\omega}\rvert^2}{\lVert Z\rVert^2} \vec\omega\in\conv(S)$.
  As the Frobenius norm of a matrix is never smaller than the $\ell^2$-norm of is diagonal, we obtain
  \begin{align*}
    \frac1{\ell^2}\sum_{i=1}^d \lVert \mu^{(i)}_W([Z])\rVert^2_F
    \geq \frac1{\ell^2} \min \{ \lVert x\rVert_2^2 : x\in\conv(S) \}
    \geq \gamma^2(\vec\lambda),
  \end{align*}
  which establishes \cref{eq:gap goal}.
\end{proof}

Next, we lower-bound the gap constant. First, we will need the following elementary lemma.

\begin{lem}\label{lem:lb_distance_0} Suppose $S = \{v_1,\ldots, v_m\} \in \ZZ^N$ is a set of integer vectors s.t. $0\not\in\conv(S)$. Let the bit complexity of entries of $v_i$'s be at most $b$. Denote by $\gamma(S)$ denote the Euclidean distance of $0$ to $\conv(S)$. Then
$$
\gamma(S) \ge \expon\left( -O(N (\log(N) + b))\right)
$$
\end{lem}

\begin{proof} Since $0\not\in\conv(S)$, by Farkas' lemma (e.g. see \cite{schrijver1998theory}), there exists a vector $\vec w \in \RR^N$ s.t. $\braket{\vec w, \vec v_i} \ge 1$ for all $i$. By (\cite{schrijver1998theory}, Corollary 3.2b, Theorem 10.1), there exists such a rational $\vec w$ with bit complexity bounded of each entry bounded by $O(N (\log(N) + b))$. Hence there exists a vector $\vec w' \in \RR^N$ (normalization of $\vec w$) s.t. $\norm{\vec w'} = 1$ and $\braket{\vec w', \vec v_i} \ge \expon\left( -O(N (\log(N) + b))\right)$ for all $i$. Now consider any element $\vec x \in \conv(S)$. Then $\braket{\vec w',\vec x} \ge \expon\left( -O(N (\log(N) + b))\right)$. Hence by Cauchy-Schwarz inequality and using the fact that $\norm{\vec w'} = 1$,
$$\norm{\vec x} \ge \expon\left( -O(N (\log(N) + b))\right)
$$
This completes the proof.
\end{proof}

We are now ready to lower bound the gap-constant.

\begin{lem}\label{lem:gap_constant_lb}
  Let $\vec\lambda$ be a highest weight and $\ell=\lvert\vec\lambda\rvert$.
  Then,
  \begin{align*}
    \gamma(\vec\lambda) \geq \gamma(n_1,\dots,n_d,\ell) := \expon\left( -O\left((n_1 + \cdots + n_d) \log(\ell \max_j n_j)\right) \right)
  \end{align*}
\end{lem}
\begin{proof}
  Consider a minimizer of $\gamma(\vec\lambda)$ given by some $S \subseteq \Omega(\Sym^\ell(V)\ot V_{\vec\lambda^*})$ s.t. $0\not\in\conv(S)$. We will apply \cref{lem:lb_distance_0}. Here the dimension $N = n_1 + \cdots + n_d$. We know that for a degree $e_1$ polynomial representation $\rho$ of $\GL(n)$, the weights $\vec \mu$ that appear in the representation must satisfy $\mu_1 + \cdots + \mu_n = e_1$ and $\mu_i \ge 0$ for all $i$. If we have a rational representation $\sigma$ which is of the form $\rho/\det^{e_2}$, then the weights $\vec \tilde{\mu}$ that appear are of the form $\vec \tilde{\mu} = (\mu_1 - e_2,\ldots, \mu_n - e_2)$, where $\mu_1 + \cdots + \mu_n = e_1$ and $\mu_i \ge 0$ for all $i$. Hence $\tilde{\mu}_i \in [-e_2, \max\{e_1, e_2\}]$ for all $i$. Now $\Sym^\ell(V)\ot V_{\vec\lambda^*}$ is a representation of $GL(n_1) \times \cdots \times \GL(n_d)$, where the $e_1^{(j)} = \ell + n_j \lambda^{(j)}_1 - \sum_{i=1}^{n_j-1} \lambda^{(j)}_i$ and $e_2 = \lambda^{(j)}_1$ for the action of component $j$. Thus the weights in $S$ have entries in $[-\ell, \ell \max_j n_j]$. Hence $b \le \log(\ell \max_j n_j)$. Therefore by Lemma \ref{lem:lb_distance_0}, we get
$$
\gamma(\vec\lambda) \ge \expon\left( -O\left((n_1 + \cdots + n_d) \log(\ell \max_j n_j)\right) \right)
$$

\end{proof}

Together, we obtain the following important result:

\minimalgap*

\section{Extensions}\label{sec:extensions}

In this section, we discuss extensions of our algorithm for general varieties and also scaling using elements of the parabolic subgroup instead of the Borel subgroup (in our case, block-upper-triangular matrices instead of upper-triangular matrices in our case). Our main theorems of this section, \cref{thm:general scaling} and \cref{thm:general main} in \cref{subsec:parametrizations}, asserts that the extended algorithm has the same guarantees in this more general setting.

\subsection{The parabolic subgroup}
Consider again the highest weight theory for $\GL(n)$; let $\vec \lambda = (\lambda_1, \dots, \lambda_n)$ be a highest weight, i.e. decreasing sequence of numbers.
If the highest weight is degenerate then the highest weight vector is an eigenvector of larger group, the so-called \emph{parabolic subgroup}~$P_{\vec\lambda}$, which is given by the upper-triangular block matrices, where the blocks correspond precisely to the degeneracies of the highest weight.
That is, suppose that the distinct values in $\vec\lambda$ are denoted by $\lambda_{[j]}$ and their multiplicities by $b_{j}$ , where, of course, $\sum_j b_j=n$.
That is, $\lambda_1 = \dots = \lambda_{b_1} = \lambda_{[1]}$, $\lambda_{b_1+1} = \dots = \lambda_{b_1+b_2} = \lambda_{[2]}$, etc.
Then $P_{\vec\lambda}$ consists of the upper-triangular block matrices~$R$ whose diagonal blocks, denoted $R_{[j]}$, have size $b_j \times b_j$.
In this case, the eigenvalues are $\chi_{\vec\lambda}(R) = \prod_j \det(R_{[j]})^{\lambda_{[j]}}$, which extends the formula given previously for $B(n)$ to $P_{\vec\lambda}$.\\

Again, this generalizes to the product group setting $G = \GL(n_1) \times \dots \times \GL(n_d)$. Now the parabolic subgroup associated to a highest weight~$\vec\lambda = (\vec\lambda^{(1)},\dots, \vec\lambda^{(d)})$ is by definition~$P_{\vec\lambda}=P_{\vec\lambda^{(1)}}\times\dots\times P_{\vec\lambda^{(d)}}$, and \cref{eq:hwv} generalizes to
\begin{align}\label{eq:hwv parabolic}
  R \cdot w = \chi_{\vec\lambda}(R) w,
\quad\text{where}\quad
\chi_{\vec\lambda}(R) = \prod_{i=1}^d \prod_{j=1}^{n_i} (R^{(i)}_{[j]})^{\lambda^{(i)}_{[j]}}
\end{align}
for all tuples $R=(R^{(1)},\dots,R^{(d)})\in P_{\vec\lambda}$ using the notation introduced above.

\subsection{Good parametrizations}\label{subsec:parametrizations}

To extend \cref{thm:scaling} from orbit closures to $\PP(V)$ or more general varieties~$\mathcal X$ of tensors, we need an effective way of sampling generic points (again, since \cref{prb:general,prb:orbit closure} are equivalent for generic points).
Suppose, e.g., that we have a homogeneous polynomial map
\begin{align*}
  \Phi \colon \PP(\CC^p) \dashrightarrow \mathcal X\subseteq \PP(\Ten(n_0;n_1,\dots,n_d)),
\end{align*}
defined on a Zariski-dense subset, such that the image of~$\Phi$ is Zariski-dense in~$\mathcal X$.
E.g., for projective space we can just choose $\Phi$ as the identity map!
(In fact, we only need to demand that $P_{\vec\lambda^*} \cdot \im(\Phi)$ is dense, where $P_{\vec\lambda^*}$ is the parabolic subgroup corresponding to a target spectrum~$\vec p$ -- see \cref{subsec:effective mumford,subsec:general} for details.)
We call such $\Phi$ a \emph{good parametrization} (obvious variations and generalization are possible). The two most basic examples of good parametrizations are orbit-closures (\cref{exa:orbit closure}) and the full space $\PP(V)$ (\cref{exa:projective space}).

\begin{dfn}[Good parametrization]\label{dfn:parametrization}
Let
\begin{align*}
  \Phi \colon \PP(\CC^p) \dashrightarrow \mathcal X\subseteq \PP(\Ten(n_0;n_1,\dots,n_d))
\end{align*}
be defined on a Zariski-dense subset by homogeneous polynomials of the same degree, denoted by~$\deg(\Phi)$.
We say that $\Phi$ is a \emph{good parametrization (of~$\mathcal X$)} if $P_{\vec\lambda^*} \cdot \im(\Phi)$ is Zariski-dense in~$\mathcal X$.
Here, we recall that $P_{\vec\lambda^*}$ denotes the parabolic subgroup corresponding to the highest weight~$\vec\lambda^*$.
We call the set of $Z\in\CC^p$ for which $\Phi(Z)\neq0$ the domain of~$\Phi$.
\end{dfn}
The statement in \cref{cor:random is generic} is slightly technical due to the presence of the parabolic subgroup~$P_{\vec\lambda^*}$. However, including the parabolic subgroup~$P_{\vec\lambda^*}$ rather than $B$ can be useful as it allows us to relax the assumptions on the parametrization depending on the degeneracy of the target spectrum. Here is a dramatic example of this phenomenon.
\begin{exa}[Uniform marginals]\label{exa:uniform}Suppose $X \in \Ten(n_0;n_1,\dots,n_d)$ and $\mathcal{X} = \overline{G \cdot [X]}$, but we have the very special condition $\lambda_j^{(i)} = 1/{n_i}$ for all $i \in [d]$. This is the uniform tensor scaling setting of~\cite{burgisser2017alternating}. Here $P_{\vec\lambda^*}$ is in fact the full group $G$; thus we may simply take the image of $\Phi$ to be $X$ and $\Phi$ will be a good parametrization. This shows that no randomness is required at all, and our parabolic scaling algorithm (\cref{alg:general scaling}) will fully recover the algorithmic guarantees of \cite{burgisser2017alternating}.
\end{exa}

On the other hand, a particularly simple case of a good parametrization is when the image of $\Phi$ is already dense, as in the following three important examples.

\begin{exa}[Orbit closure]\label{exa:orbit closure}
When $\mathcal X=\overline{G \cdot [X]}$ is a single orbit closure, as in \cref{prb:orbit closure}, then a good parametrization is given by
\begin{align*}
  \Phi_X\colon \PP(M(n_1)\times\dots\times M(n_d)) \dashrightarrow \mathcal X, \quad [A=(A^{(1)},\dots,A^{(d)})] \mapsto [(A^{(1)}\ot\dots\ot A^{(d)}) X],
\end{align*}
where $M(n)$ denotes the space of $n_i\times n_i$-matrices.
Since $\overline{\GL(n)} = M(n)$, the image (when nonzero) is contained in $\overline{G \cdot [X]}$.
Note that $\Phi$ is homogeneous of degree~$\deg(\Phi)=d$.
Thus, the constant~$M$ in \cref{prp:random is generic parabolic} is given by~$M=2dK$.
\end{exa}

\begin{exa}[All tensors]\label{exa:projective space}
If $\mathcal X=\PP(\Ten(n_0;n_1,\dots,n_d))$ is the space of all tensors of a given format, as in \cref{prb:qmp}, then we can simply choose~$\Phi$ as the identity map.
Thus, $\deg(\Phi)=1$, so~$M=2K$.
\end{exa}

\begin{exa}[Matrix product states]\label{exa:matrix product states}
For simplicity, we only discuss translation-invariant matrix product states (see, e.g., \cite{verstraete2008matrix} for the general definition).
Given a family of $N\times N$-matrices~$\{M_j\}_{j=1,\dots,n}$, define a corresponding tensor in $\Ten(1;n,\dots,n)$ by
\begin{align*}
  X[\{M_j\}]_{1;j^{{1)}},\dots,j^{(d)}} := \tr[M_{j^{(1)}} \cdots M_{j^{(d)}}].
\end{align*}
The closure $\mathcal X$ of the set of all tensors of this form is called the variety of \emph{matrix product states} with \emph{bond dimension}~$N$.
It is clear that $\mathcal X$ is a $G$-stable subvariety of $\PP(\Ten(1;n,\dots,n))$. Moreover,
\begin{align*}
  \Phi\colon \PP(\CC^{n\times N\times N}) \dashrightarrow \Ten(1;n,\dots,n), \quad [\{M_j\}] \mapsto X[\{M_j\}]
\end{align*}
is dominant.
It follows that~$\Phi$ is a good parametrization of $\mathcal X$ and that $\mathcal X$ is irreducible.
Since $\deg(\Phi)=d$, the constant~$M$ in \cref{prp:random is generic parabolic} is given by~$M=2dK$.

Note that we parametrize a tensor with~$n^d$ entries by only $N^2n$ parameters -- this is the power of matrix product states.
Note also that $\Phi$ is equivariant with respect to the natural $\GL(n)$-action on $\PP(\CC^{n\times N\times N})$ -- so we can implement \cref{alg:scaling} by working solely in the small parameter space.
\end{exa}
The following is the main result of this section.
\begin{restatable}[Tensor scaling for good parametrizations]{thm}{genthm}\label{thm:general scaling}
Let $\mathcal X\subseteq\PP(\Ten(n_0;n_1,\dots,n_d))$ be a $G$-stable irreducible projective subvariety.
Let $\Phi\colon\PP(\CC^p)\dashrightarrow\mathcal X$ be a good parametrization (in the sense of \cref{dfn:parametrization}) with Gaussian integer coefficients of bitsize no more than~$b$.
Also, let $\vec p \in P_+(n_1,\dots,n_d)$ with rational entries of bitsize no more than~$b$ such that $p^{(i)}_{n_i}>0$~for all $i=1,\dots,d$.
Finally, let $\eps>0$.
Then, with probability at least 1/2, \cref{alg:general scaling} either correctly identifies that $\vec p\not\in\Delta(\mathcal X)$, or it outputs $X\in\mathcal X$ and $g\in G$ such that the marginals of $Y=g\cdot X$ are $\eps$-close to~$\vec p$ (in fact, satisfy \cref{eq:final marginals}).
\end{restatable}
We'll also see that, just like \cref{alg:scaling}, \cref{alg:general scaling} can be modified to give the same guarantees even if the rank of some $\diag(\vec{p}^{(i)})$ is not full.
\begin{thm}\label{thm:general main}Let $\Phi$ be as in the statement of \cref{thm:general scaling}, and further suppose that dimension of the domain of $\Phi$ is $\poly(N)$ and on inputs of bit-complexity $c$, $\Phi$ can be computed in $\poly(N,c)$ time. Then there is a randomized algorithm running in time $\poly(N, 1/\eps)$, that takes as input $X\in\Ten(n_0;n_1,\dots,n_d)$ with Gaussian integer entries (specified as a list of real and complex parts, each encoded in binary, with bit size $\le b$) and
$\vec p\in P_+(n_1,\dots,n_d)$ with rational entries (specified as a list of numerators and denominators, each encoded in binary, with bit size $\le b$). The algorithm either correctly identifies that $\vec p \notin \Delta(\mathcal X)$, or it outputs a scaling $g \in G$ and $X \in \mathcal X$ such that the marginals of $g \cdot X$ are $\eps$-close to the target spectra~$\vec p$. Here $N$ is the total bit-size of the input, $N = 2 n_0 n_1 \cdots n_d b + 2 (n_1 + \cdots n_d) b$.

\end{thm}
\begin{rem}
The maps $\Phi$ encountered in the three examples before \cref{thm:general scaling} are indeed computable in time polynomial in the input size; as a corollary, there are efficient (in the sense of  \cref{thm:main}) algorithms to correctly declare $\vec p \notin \Delta(\mathcal X)$ or outputs a scaling $g \in G$ and $X \in \mathcal X$ such that the marginals of $g \cdot X$ are $\eps$-close to the target spectra~$\vec p$ with probability at least $1/2$ for $\mathcal X, G$ as in \cref{exa:uniform}, \cref{exa:orbit closure}, \cref{exa:projective space}, and \cref{exa:matrix product states}. Note that that no randomness is required at all for \cref{exa:uniform}, so the main algorithmic result of \cite{burgisser2017alternating} is a special case of \cref{thm:general main}. \cref{exa:orbit closure} is already covered by \cref{thm:main}, which is a special case of \cref{thm:general main}.
\end{rem}

\begin{Algorithm}[th!]
\textbf{Input}:
$\vec p\in P_+(n_1,\dots,n_d)$ with rational entries (specified as a list of numerators and denominators, each encoded in binary, with bit size $\le b$) such that $p^{(i)}_{n_i}>0$~for all $i=1,\dots,d$.
\\[.1ex]

\textbf{Output:}
Either the algorithm correctly identifies that $\vec p \notin \Delta(\mathcal X)$, or it outputs $X\in\mathcal X$ and $g\in G$ such that the marginals of $Y := g \cdot X$ satisfy \cref{eq:final marginals}; in particular the marginals are $\eps$-close to the target spectra~$\vec p$.
\\[.1ex]

\textbf{Algorithm:}\vspace{-.2cm}
\begin{enumerate}
\item \label{it:general randomize}
Let $\ell>0$ such that $\ell \vec p^{(i)}$ has integer entries~for all $i=1,\dots,d$.
Let $Z = (Z^{(1)},\dots,Z^{(p)})$ be a vector with entries chosen independently uniformly at random from $\{1,\dots,M\}$, where
$M := 2 \deg(\Phi) K$ and $K := (\ell d \textstyle\max_{i=1}^d n_i )^{d \max_{i=1}^d n_i^2}$.
Set $X := \Phi(Z)$.
\item\label{it:general full rank}
For $i=1,\dots,d$, if the marginal~$\rho_X^{(i)}$ is singular then output $\vec p\not\in\Delta(X)$ and return. \\
Otherwise, set $g := (I_{n_1}/\Norm{X},I_{n_2},\dots,I_{n_d})$.
\item\label{it:general borel scale}
For $t=1,\dots,T :=  \frac{16 k \ln 2 }{\eps^2 } \left( \sum_{i=0}^d \log_2(n_i)  +  b  + \deg \Phi (\log_2 p + \log_2 M) \right)$, repeat the following:
\begin{itemize}
\item Compute $Y := g \cdot X$ and, for $i=1,\dots,d$, the one-body marginals~$\rho_Y^{(i)}$ and the distances~$\eps^{(i)} := \lVert \rho_Y^{(i)} - \diag(\vec p_\uparrow^{(i)}) \rVert_{\tr}$.
\item Select an index $i\in\{1,\dots,d\}$ for which~$\eps^{(i)}$ is largest. If~$\eps^{(i)}\leq\eps$, output $g$ and return.
\item Compute the Cholesky decomposition $\rho_Y^{(i)} = R^{(i)} (R^{(i)})^\dagger$, where $R^{(i)}$ is an upper-triangular matrix.
Update $g^{(i)} \leftarrow \diag(p^{(i)}_{n_i},\dots,p^{(i)}_1)^{1/2} (R^{(i)})^{-1} g^{(i)}$.
\end{itemize}
\item\label{it:general give up} Output $\vec p\not\in\Delta(X)$.
\end{enumerate}
\caption{Scaling algorithm for \cref{thm:general scaling}}\label{alg:general scaling}
\end{Algorithm}

\subsection{Randomization step}\label{subsec:randomization step}

We have the following extension of \cref{cor:random is generic} for the parabolic subgroup. This allows us to find and element on which some highest weight vector does not vanish. The proof is almost identical to that of \cref{cor:random is generic}, but instead of scaling by a random element we evaluate $\Phi$ on a random element.
\begin{prp}[Generic orbits]\label{prp:random is generic parabolic}
Let $\vec p\in\Delta(\mathcal X)$ and $\ell>0$ such that $\vec\lambda:=\ell\vec p$ is integral.
Moreover, let $\Phi \colon \PP(\CC^p) \dashrightarrow \mathcal X$ be a good parametrization in the sense of \cref{dfn:parametrization}.
Finally, choose~$Z_1,\dots,Z_p$ independently and uniformly at random from~$\{1,\dots,M\}$, with
\begin{align*}
   M = 2 \deg(\Phi) K,
\quad
  K := \left(\ell d \textstyle\max_{i=1}^d n_i \right)^{d \max_{i=1}^d n_i^2}.
\end{align*}
Then, with probability at least $1/2$, $Z$ is in the domain of~$\Phi$ and there exists a highest weight vector $P\in\HWV_{m\vec\lambda^*}(\CC[V]_{(\ell m)})$ of degree $0<\ell m\leq K$ such that $P(\Phi(Z))\neq0$.
In particular, $\vec p\in\Delta(\Phi(Z))$.
\end{prp}
\begin{proof}
Set $\vec\lambda:=\ell\vec p$.
According to \cref{prp:effective mumford}, there exists a highest weight vector $P\in\HWV_{m\vec\lambda^*}(\CC[V]_{(\ell m)})$ of degree~$0<\ell m\leq K$ such that $P(X)\neq0$ for some $X\in\mathcal X$.
But then $\{P\neq0\}$ is a nonempty Zariski-open subset of~$\mathcal X$.
Since by assumption $P_{\vec\lambda^*} \cdot \im(\Phi)$ is Zariski-dense and $\mathcal X$ is irreducible, it follows that $\{P\neq 0\} \cap P_{\vec\lambda^*} \cdot \im(\Phi)$ is nonempty.
Since $P$ is a highest weight vector, $\{P\neq0\}$ is $P_{\vec\lambda^*}$-stable, so in fact $\{P\neq0\} \cap \im(\Phi)$ is nonempty.
This means that the polynomial
$
  Q(Z) := P(\Phi(Z))
$
is not equal to the zero polynomial.
Its degree is no larger than $\deg(\Phi) K$, so the Schwartz-Zippel lemma implies that for our random choice of~$Z$, $Q(Z)\neq0$ with probability at least~$1/2$.
But then not only $P(\Phi(Z))\neq0$, but also $\Phi(Z)\neq0$, i.e., $Z$ is in the domain of $\Phi$, since $\Phi$ is homogeneous.
\end{proof}

\subsection{Parabolic scaling step}\label{subsec:parabolic scaling}
In step~\ref{it:borel scale} of our original algorithm (\cref{alg:scaling}) we replace the Cholesky decomposition~$\rho_Y^{(i)} = R^{(i)} (R^{(i)})^\dagger$, where $R^{(i)}$ is an upper-triangular matrix, by an element from $R^{(i)}$ from the parabolic subgroup corresponding to the target spectrum~$\vec p^{(i)}$.
In particular we can use the Hermitian square root $(\rho^{(i)}_Y)^{1/2}$ for scaling to the uniform spectrum, as in~\cite{burgisser2017alternating}.
This follows directly by substituting \cref{prp:progress} by \cref{prp:parabolic progress} below in the proof of \cref{thm:scaling}.

Our scaling step is \begin{align*}
  Y' := \diag(\vec p_\uparrow^{(i)})^{1/2} (R^{(i)})^{-1} \cdot Y,
\end{align*}
with $\rho_Y^{(i)} = R^{(i)} (R^{(i)})^\dagger$, where we now allow that $R^{(i)} \in P_{(\vec\lambda^*)^{(i)}}$ is an element of the parablic subgroup corresponding to the target spectrum ($\vec\lambda=k\vec p$ for some $k>0$).

\begin{prp}[Progress under parabolic scaling]\label{prp:parabolic progress}
  Let $P\in\HWV_{\vec\lambda^*}(\CC[V]_{(k)})$. Then,
  \begin{align*}
    \lvert P(Y') \rvert \geq 2^{\frac k{32\ln2} \lVert \diag(\vec p_\uparrow^{(i)}) - \rho^{(i)} \rVert_{\tr}^2} \lvert P(Y) \rvert.
  \end{align*}
\end{prp}
\begin{proof}
  Let $(\lambda^*)^{(i)}_{[j]}$ denote the distinct values in $(\vec\lambda^{(i)})^*$, $b^{(i)}_{[j]}$ their multiplicities, and $R^{(i)}_{[j,j]}$ the corresponding diagonal blocks of~$R^{(i)}$.
  Moreover, let $p^{(i)}_{\uparrow,[j]}:=-(\lambda^*)^{(i)}_{[j]}$ denote the distinct values of $\vec p^{(i)}_{\uparrow}:=(p^{(i)}_{n_i},\dots,p^{(i)}_1)$.
  Using \cref{eq:hwv parabolic} instead of \cref{eq:hwv poly}, we obtain
\begin{align*}
  \lvert P(Y') \rvert^2
&= \left( \prod_j (p^{(i)}_{\uparrow,[j]})^{k b^{(i)}_{[j]} p^{(i)}_{\uparrow,[j]}} \lvert \det(R^{(i)}_{[j,j]}) \rvert^{-2k p^{(i)}_{\uparrow,[j]}} \right) \lvert P(Y)\rvert^2\\
&= \left( \prod_j (p^{(i)}_{\uparrow,[j]})^{b^{(i)}_{[j]} p^{(i)}_{\uparrow,[j]}} \det(R^{(i)}_{[j,j]}(R^{(i)}_{[j,j]})^\dagger)^{-p^{(i)}_{\uparrow,[j]}} \right)^k \lvert P(Y)\rvert^2\\
&= 2^{k \sum_j \left( b^{(i)}_{[j]} p^{(i)}_{\uparrow,[j]} \log_2(p^{(i)}_{\uparrow,[j]}) - p^{(i)}_{\uparrow,[j]}\tr[\log_2(R^{(i)}_{[j,j]}(R^{(i)}_{[j,j]})^\dagger)] \right)} \lvert P(Y)\rvert^2\\
&= 2^{k D(\diag(\vec p^{(i)}_{\uparrow}) \Vert Q^{(i)})} \lvert P(Y)\rvert^2
\;\geq\; 2^{\frac k{16\ln2} \lVert \diag(\vec p^{(i)}_\uparrow) - \rho^{(i)} \rVert_{\tr}^2} \lvert P(Y)\rvert^2.
\end{align*}
where $Q^{(i)}=\diag(R^{(i)}_{[1,1]} (R^{(i)}_{[1,1]})^\dagger,R^{(i)}_{[2,2]} (R^{(i)}_{[2,2]})^\dagger,\dots)$.
The inequality is \cref{lem:generalized lsw blocks}, stated and proved below.
\end{proof}

The following generalizes \cref{lem:generalized lsw}.

\begin{lem}\label{lem:generalized lsw blocks}
Let $\rho$ be a PSD $n\times n$-matrix with unit trace such that $\rho=RR^\dagger$, where $R$ is an arbitrary $n\times n$-matrix.
Partition~$R$ into blocks $R_{[j,k]}$ of size $b_j\times b_k$.
Then the block-diagonal matrix $Q=\diag(R_{[1,1]}R_{[1,1]}^\dagger,R_{[2,2]}R_{[2,2]}^\dagger,\dots)$ is a PSD matrix with $\tr[Q]\leq1$, and, for every probability distribution $\vec p$,
\begin{align*}
  D(\diag(\vec p)\Vert Q) \geq \frac1{16\ln2} \lVert \diag(\vec p) - \rho \rVert_{\tr}^2.
\end{align*}
where $D(P\Vert Q) := \tr[P(\log_2P-\log_2Q)]$ is the quantum relative entropy.
\end{lem}
\begin{proof}
  To see that $Q$ is subnormalized, observe that
  \begin{align*}
    \tr[Q] = \sum_j \tr[R_{[j,j]} R_{[j,j]}^\dagger]
  \leq \tr[R R^\dagger] = \tr[\rho] = 1.
  \end{align*}
  On the one hand, the quantum Pinsker's inequality in the form~\cite[Thm.~11.9.1]{wilde2013quantum} yields
  \begin{align*}
    D(\diag(\vec p)\Vert Q)
  \geq \frac1{2\ln2} \lVert \diag(\vec p) - Q\rVert_{\tr}^2.
  \end{align*}
  On the other hand,
  \begin{align*}
    D(\diag(\vec p)\Vert Q)
  &= D(\diag(\vec p)\Vert Q/\tr[Q]) - \log_2 \tr[Q]
  \geq - \log_2 \tr[Q]
  \geq \frac1{\ln2} \left( 1 - \tr[Q] \right) \\
  &= \frac1{\ln2} \sum_{j\neq k}^n \tr[R_{[j,k]} R_{[j,k]}^\dagger]
  = \frac1{\ln2} \lVert R - D \rVert_F^2
  \geq \frac1{4\ln2} \lVert \rho - Q \rVert_{\tr}^2,
  \end{align*}
  where $D=\diag(R_{[1,1]},R_{[2,2]},\dots)$, so that $DD^\dagger=Q$.
  In the last step, we used that for any two matrices~$A$ and $B$, $\lVert AA^\dagger - BB^\dagger \rVert_{\tr} \leq \lVert A + B \rVert_F \lVert A - B\rVert_F \leq (\lVert A\rVert_F + \lVert B\rVert_F) \lVert A - B \rVert_F$ (see \cite[Proof of X.2.4]{bhatia2013matrix}).
  Averaging both inequalities, we find that
  \begin{align*}
    D_{KL}(\vec p\Vert\vec q)
  &\geq \frac1{4\ln2} \lVert \diag(\vec p) - Q\rVert_{\tr}^2 + \frac1{8\ln2} \lVert \rho - Q \rVert_{\tr}^2\\
  &\geq \frac1{8\ln2} \left( \lVert \diag(\vec p) - Q\rVert_{\tr}^2 + \lVert \rho - Q \rVert_{\tr}^2 \right)\\
  &\geq \frac1{16\ln2} \left( \lVert \diag(\vec p) - Q\rVert_{\tr} + \lVert \rho - Q \rVert_{\tr} \right)^2\\
  &\geq \frac1{16\ln2} \lVert \diag(\vec p) - \rho \rVert_{\tr}^2. \qedhere
  \end{align*}
\end{proof}

\subsection{Proof of Theorem~\ref{thm:general scaling}}\label{subsec:general}
We will now outline the proof of \cref{thm:general scaling} (restated below).
\genthm*

The proof is nearly identical to that of \cref{thm:scaling} in \cref{thm:scaling:proof}.

\begin{proof}[Proof of \cref{thm:general scaling}]
Assume $X, \vec p, \epsilon$ are an input for \cref{alg:general scaling}.
Assume first $\vec p\in\Delta(X)$.
We need to show that, with probability at least 1/2, \cref{alg:general scaling} terminates in step~\ref{it:general borel scale} by outputting an appropriate scaling.

In step~\ref{it:general randomize}, rather than selecting random matrices, we computed $\Phi(Z)$ on a random tuple of integers $Z$ according to the parameters explained in \cref{exa:orbit closure}. with probability at least $1/2$, $Z$ is in the domain of~$\Phi$ and there exists a highest weight vector $P\in\HWV_{m\vec\lambda^*}(\CC[V]_{(\ell m)})$ of degree $0<\ell m\leq K$ such that $P(\Phi(Z))\neq0$.

Again, we condition on this event.
By \cref{prp:hwv eval}, we may further assume that $P$ has integer coefficients and that it satisfies the bound
\begin{align}\label{eq:general hwv bound}
  \lvert P(Y)\rvert \leq (n_1\dots n_d)^k \lVert Y\rVert
\end{align}
for all tensors $Y\in V$.

We now move to the scaling step~\ref{it:general borel scale}.
Let us denote by $g[t]\in G$ the value of the group element~$g$ at the beginning of the $t$-th iteration, and by $Y[t] := g[t] \cdot X$ the corresponding tensor.
Suppose for sake of finding a contradiction that the algorithm has not terminated after~$T$ steps but instead proceeds to \ref{it:general give up}.
We will prove the following three statements:
\begin{itemize}
\item \textbf{Lower bound:} $\lvert P(Y[1]) \rvert \geq  2^{\frac{1}{2} \left( \sum_{i=0}^d \log_2(n_i) - b  - \deg \Phi( \log_2 p  + \log_2 M \right)}$,
\item \textbf{Progress per step:} $\lvert P(Y[t+1])\rvert > 2^{\frac k{32\ln 2}\eps^2} \lvert P(Y[t])\rvert$ for $t=1,\dots,T$,
\item \textbf{Upper bound:} $\lvert P(Y[t]) \rvert \leq 2^{k\sum_{i=1}^d \log_2(n_i)}$.
\end{itemize}
The proof of the upper bound is identical to that in the proof of \cref{thm:scaling} in \cref{thm:scaling:proof}. The proof of the progress per step is also identical, except we use \cref{prp:parabolic progress} for parabolic scalings instead of \cref{prp:progress} which only applies to Borel scalings. For the lower bound, this time we obtain
\begin{align*}  \lVert X\rVert &= \lVert \Phi(Z)\rVert\\
& \leq \sqrt{n_0 \dots n_d 2^b  M^{\deg \Phi} p^{\deg \Phi}}\\
& \leq 2^{\frac{1}{2} \left( \sum_{i=0}^d \log_2(n_i) + b  + \deg \Phi( \log_2 p  + \log_2 M \right)}.
\end{align*}
Again, $|P(X)| \geq 1$ by integrality and so combining \cref{eq:general hwv bound} with the previous equation gives us the lower bound.
Suppose for sake of finding a contradiction that the algorithm has not terminated after~$T$ steps but instead proceeds to \ref{it:general give up}; the three inequalities imply
$$
\frac{T k}{32\ln 2}\eps^2
< \frac{k}{2} \left( \sum_{i=0}^d \log_2(n_i)  +  b  + \deg \Phi( \log_2 p  + \log_2 M ) \right),
$$
which gives the desired contradiction.
\end{proof}
The proofs of \cref{lem:singular_spectra} and \cref{cor:singular_algorithm} work mutatis mutandis for parabolic scalings, with $B$ replaced by $P_{\vec\lambda^*}$ and $B_+$ replaced by the parabolic subgroup of $\GL(r_1)\times \dots \GL(r_d)$ corresponding to ${\vec p_+}$. This implies the following:
\begin{cor}\label{cor:general main true}
\cref{thm:general main} is true.
\end{cor}

\bibliographystyle{alphaurl}
\addcontentsline{toc}{section}{References}
\bibliography{nonuniform}

\appendix
\section{Appendix}

\subsection{Borel polytope}\label{sec:polytope}

Here we include an elementary description of the Borel polytope $\Delta^B(X)$ for $X \in \Ten(n_0; n_1, \dots, n_d)$.
and prove that it is indeed a polytope with rational vertices.

By the remarks after the proof of \cref{prp:borel} we have
$\vec{p} \in \Delta^B(Y)$ if and only if
$$0 <(\capacity_{\vec{p}}(\rho_Y))^2 =  \inf_{b \in B} \tr (b \cdot \rho_Y) |\chi_{\vec{p}^*}(b)|^2.$$

It's easy to see that $\log\capacity_{\vec{\lambda}}(\rho)$ is \emph{concave} in $\lambda$!
 This immediately implies $\Delta_B(\rho)$ is convex, but we can say more:
\begin{thm}\label{thm:polytope}
$\Delta^B(X)$ is a polytope with rational vertices.
\end{thm}
This follows from a more detailed description of $\Delta^B(X)$, which requires some technical definitions. The ideas are very similar to the elementary derivation of the tensor case of the Hilbert-Mumford criterion \cite{burgisser2017alternating}. In the uniform case, a density matrix $\rho$ can be scaled to uniform marginals if in \emph{every} orthonormal basis, the diagonal of $\rho$ (regarded as a classical tensor) can be scaled by diagonal matrices to uniform marginals.\\
\indent We find a similar criterion, but the reduced density matrix will be blown up (much as in \cref{sec:reductions}, and the diagonal must be scaled to certain \emph{nonuniform} marginals. Let $\vec n:=(n_1, \dots, n_d)$.
\begin{dfn}[unitary family]\label{dfn:basis_family}
A $(d,\vec{n})$-\emph{unitary family} is a tuple
$$U = (U_j^{(i)}, i \in [d], j\in [n_i])$$
where $U_j^{(i)}$ is a $j\times j$ unitary matrix.
\end{dfn}
Let $S$ be the set of pairs $(\vec{j}, \vec{l})$ such that $\vec{j} = (j(1), \dots, j(d)) \in [n_1]\times \dots \times [n_d]$ and $\vec{l} = (l(1), \dots, l(d)) \in [j(1)]\times \dots \times [j(d)]$. We use $S$ as an index set for $\operatorname{Ten}_{\binom{n_1 + 1}{2}, \dots, \binom{n_1 + 1}{d}}(\RR_{\geq 0})$.
\begin{dfn}[expanded classical tensor]\label{dfn:exp_classical}
Given a $(d, \vec{n})$-unitary family $U$, denote by $C(\rho, U)$ the element of $\Ten_{\binom{n_1 + 1}{2}, \dots, \binom{n_1 + 1}{d}}(\RR_{\geq 0})$ given by
\begin{align*} C(\rho, U)_{\vec{j}, \vec{l}}:= (U_{\vec{j}} \cdot \nu_{\vec{j}}\cdot \rho)_{\vec{l}, \vec{l}},
\end{align*}
where $u_{\vec{j}, \vec{l}}$ denotes the tuple $(U_{j(1)}^{(1)}, \dots, U_{j(d)}^{(d)})$ and $\nu_{\vec{j}}$ denotes $(\nu_{j(1)}, \dots, \nu_{j(d)})$. For any $C \in \Ten_{\binom{n_1 + 1}{2}, \dots, \binom{n_1 + 1}{d}}(\RR_{\geq 0})$, we define
$$\operatorname{Supp}(C) = \{({\vec{j}, \vec{l}}) \in S: C_{\vec{j}, \vec{l}} > 0\}.$$
\end{dfn}
We may now state a description of $\Delta_B(Y)$. We use the shorthand $\Delta p_j^{(i)} = p_j^{(i)} - p_{j+1}^{(i)}$, where $p_{n_i + 1} \equiv 0$.
\begin{prp}\label{prp:polytope}
$\vec{p} \in \Delta_B(Y)$ if and only if for every $(d,\vec{n})$-basis-family $U$, there is a tensor
$$ D \in \Ten_{\binom{n_1 + 1}{2}, \dots, \binom{n_1 + 1}{d}}(\RR_{\geq 0})$$
satisfying
\begin{align}\operatorname{Supp} D \subset \operatorname{Supp} C(\rho_Y, U)\label{eq:support}
\end{align} and
\begin{align} \sum_{(\vec{j}, \vec{l}) \in S: j(i)= j,  l(i) = l} D_{\vec{j}, \vec{l}}& = \Delta p_j^{(i)}\nonumber\\
\textrm{ for all }i \in [d], j &\in [n_i], l \in [j].\label{eq:classical_margins}
\end{align}
That is, $D$ has as $i^{th}$ classical margin the vector $(\Delta p_j^{(i)} : j \in [n_i], l \in [j]) \in \RR_{\geq 0}^{\binom{n_i + 1}{2}}$.
\end{prp}

Before we prove \cref{prp:polytope}, we use it to prove \cref{thm:polytope}.

\begin{proof}[Proof of \cref{thm:polytope}:] Fix $\operatorname{Supp} C(\rho_Y, U)$. By Farkas' lemma, the existence of $D$ satisfying \cref{eq:support} and \cref{eq:classical_margins} is equivalent to the following statement:\\
\indent Every sequence $\vec{a} = (a_{j,l}^{(i)}, i \in [d], j \in [n_i], l \in [j])$ of numbers satisfying
\begin{align} \sum_{i = 1}^d a_{j(i), l(i)}^{(i)} \geq 0 \textrm{ for all } (\vec{j}, \vec{l}) \in \operatorname{Supp} C(\rho_Y, U) \label{eq:dual_support}
\end{align}
also satisfies
\begin{align} \sum_{i = 1}^d \sum_{j = 1}^{n_i}\Delta p_j^{(i)}\sum_{i = 1}^l a_{j(i), l(i)} (i) \geq 0. \label{eq:dual_large}\end{align}
Since the set of $\vec{a}$ satisfying \ref{eq:dual_support} is a convex cone with finitely many constraints, it is generated by a finite set of rational vectors; it is enough to check that \ref{eq:dual_support} implies \ref{eq:dual_large} on that finite set of rational vectors; this implies the Borel polytope is indeed a polytope with rational vertices (there are only a finite number of possibilities for $\operatorname{Supp}C(\rho_Y, U)$!)
\end{proof}

Before we prove \cref{prp:polytope}, we must prove a lemma. We wish to characterize when $\capacity_{\vec{p}} \rho >0$; the lemma allows us to pass to a larger set for the infimum.

\begin{lem} [Adapted from \cite{franks2018operator}]\label{lem:relaxed_capacity} Let $p \in P_+(n)$ with $p_n > 0$ and $b \in B(n)$. The character
$$|\chi_{p^*}(b)|^{-2}.$$
is equal to
\begin{align} \textrm{ sup \hspace{1cm}} &\prod_{i = 1}^n \det (Y_i)^{\Delta p_i} \label{eq:dets_prod} \\
\textrm{ subject to \hspace{1cm}} &0 \prec Y_i: \CC^i \to \CC^i\label{eq:psd}\\
\textrm{ and    \hspace{1cm}  } &\sum_{i = 1}^n \Delta p_i \nu_i^\dagger Y_i \nu_i = b^\dagger \diag(p_{op})b \label{eq:correct_sum}
\end{align}
\end{lem}
\begin{proof}
Let $\sigma$ be the permutation reversing the order of the coordinates (the dimension is suppressed in an abuse of notation). The lemma can be obtained from Claim 4.6 in \cite{franks2018operator} by making the change of variables $h = \sigma b \sigma^\dagger$ and noting that the projection $\eta_i$ to the \emph{first} $i$ coordinates is given by $\sigma \nu_i \sigma^\dagger$. Finally one uses $|\chi_{p^*}( b )|^{-2}= \det(\diag(p), h^\dagger h)$.
\end{proof}
We have one more easy lemma. We use the shorthand $P^{(i)} := \diag(\vec p_\uparrow^{(1)})$ and
 $$\vec{P} := \left(P^{(1)}, \dots, P^{(d)}\right).$$
\begin{lem}[modified capacity]\label{lem:modified_capacity}
$$\inf_{R \in B} |\chi_{\vec{p}^*}(R)|^2\tr \sqrt{\vec{P}}  \cdot (R\cdot \rho) = 0 \iff \capacity_{\vec{p}}(\rho) = 0.$$
\end{lem}
\begin{proof}
If $\vec{P}$ is nonsingular, we apply the same change of variables argument from \cref{rem:capacity}.
Otherwise, from \cref{lem:singular_spectra}, we have $\vec{p} \in \Delta^B(\rho)$ if and only if $\vec{p}_+ \in \Delta^{B_+}(\rho_{+})$. This is because $\rho_+ = \rho_{X_+}$ if $\rho = \rho_X$. Thus,
$$ \capacity_{\vec{p}_+}\rho_+ = 0 \iff \capacity_{\vec{p}} \rho = 0.$$
By a change of variables, $\capacity_{\vec{p}_+} \rho_+ = 0$ if and only if
\begin{align*}
\inf_{b_+ \in B_+} |\chi_{\vec{p}_+^*}(b_+)|^2 \tr \nu_{\vec{r}}  \sqrt{\vec{P}} \nu_{\vec{r}}^\dagger \cdot  (b_+ \cdot \rho_+) = 0,
\end{align*}
but
\begin{align*}
\inf_{b_+ \in B_+} |\chi_{\vec{p}_+^*}(b_+)|^2\tr \nu_{\vec{r}}  \sqrt{\vec{P}} \nu_{\vec{r}}^\dagger \cdot  (b_+ \cdot \rho_+) &= \inf_{b \in B}|\chi_{\vec{p}^*}(b)|^2 \tr \nu_{\vec{r}}  \sqrt{\vec{P}} \nu_{\vec{r}}^\dagger \cdot  (\nu_{\vec{r}}  b \nu_{\vec{r}}^\dagger\cdot \rho_+) \\
&= \inf_{b \in B} |\chi_{\vec{p}^*}(b)|^2\tr (\nu_{\vec{r}}  \sqrt{\vec{P}} \nu_{\vec{r}}^\dagger)  (\nu_{\vec{r}}\cdot (b\cdot \rho))\\
&= \inf_{b \in B}|\chi_{\vec{p}^*}(b)|^2 \tr \sqrt{\vec{P}}  \cdot (b\cdot \rho).
\end{align*}
The last equality follows form cyclicity of trace and $\nu_{\vec{r}}^\dagger \nu_{\vec{r}} \sqrt{\vec{P}} = \sqrt{\vec{P}}$.
\end{proof}
Finally we prove \cref{prp:polytope}.
\begin{proof}[Proof of \cref{prp:polytope}]

We want to find necessary and sufficient conditions under which $\capacity_{\vec{p}}(\rho) = 0$. By \cref{lem:modified_capacity}, this happens if and only if
\begin{align}\inf_{b \in B}|\chi_{\vec{p}^*}(b)|^2 \tr \sqrt{\vec{P}}  \cdot (b\cdot \rho)\nonumber\\
 = \inf_{b \in B}|\chi_{\vec{p}^*}(b)|^2 \tr (b_1^\dagger P^{(1)} b_1 \otimes \dots \otimes b_d^\dagger P^{(d)} b_d)\rho  = 0.\label{eq:capacity_with_ps}
 \end{align}

Due to \cref{lem:relaxed_capacity}, we may replace $b_i^\dagger P^{(i)} b_i$ by $\sum_{j =1}^{n_i}\nu_j^\dagger Y_j^{(i)}\nu_j$ and $|\chi_{\vec{p}^*}(b)|^2$ by the product over $i\in [d]$ of $\prod_{j \in [n_i]} (\det Y_j^{(i)})^{\Delta p_j^{(i)}}$, and the infimum will remain the same!  Thus, $\capacity_{\vec{p}}(\rho) = 0$ if and only if
\begin{align}
 \textrm{ inf \hspace{1cm}}&\sum_{\vec{j} \in [n_1] \times \dots \times [n_d]} \left(\prod_{i = 1}^d \Delta p_{j(i)}^{(i)}\right) \tr \left(Y_{j(1)}^{(1)}  \otimes \dots \otimes
 Y_{j(d)}^{(d)} \right) (\nu_{\vec{j}}\cdot \rho) & = 0 \label{eq:sup_whys}\\
 \textrm{ subject to \hspace{1cm}}& 0 \prec Y_j^{(i)}:\CC^{j} \to \CC^{j} &\textrm{ for all } j \in [n_i] \label{eq:psd_whys}\\
\textrm{ and \hspace{1cm}}&\prod_{j \in [n_i]} (\det Y_j^{(i)})^{\Delta p_j^{(i)}} = 1 &\textrm{ for all }  i \in [d]. \label{eq:dets_whys}
\end{align}

We now prove the ``if" direction of \cref{prp:polytope}, namely that if the value of the above program is zero then there is some $(d, \vec{n})$-basis family $U$ such that $\operatorname{Supp} C(\rho, U)$ does not admit a solution $D$ to \cref{eq:support} and \cref{eq:classical_margins}.
\subsubsection*{The ``if'' direction}
Suppose there is a sequence $Y^{(i)}_j (t)$ satisfying \cref{eq:psd_whys} and \cref{eq:dets_whys} such that the expression in \cref{eq:sup_whys} tends to zero. We can diagonalize $Y^{(i)}_j(t) = U^{(i), \dagger}_j(t) \diag(z^{(i)}_{j,l}(t)) U^{(i)}_j(t)$ such that
\begin{enumerate}
\item $U^{(i)}_j(t)$ unitary.
\item $z^{(i)}_{j,l}(t) > 0$,
\item $\prod_{j = 1}^n \left(\prod_{l \in [j]} z^{(i)}_{j,l}(t)\right)^{\Delta p^{(i)}_j} = 1$,
\item  and $z^{(i)}_{j,l}(t)$ tends to zero if $\Delta p^{(i)}_j = 0$ (in that case $Y^{(i)}_j$ appears neither in \cref{eq:dets_whys} nor \cref{eq:sup_whys}).
\end{enumerate}
Let $U(t)$ be the $(d,\vec{n})$-basis family $(U^{(i)}_j(t), i \in [d], j \in [n_i])$.  By compactness, we pass to a convergent subsequence such that $\lim_{t \to \infty} U^{(i)}_j(t) = U^{(i)}_j$. Let $U$ be the $(d,\vec{n})$-basis family $(U^{(i)}_j, i \in [d], j \in [n_i])$. We claim that for all $(\vec{j}, \vec{l})$ in $\operatorname{Supp} C(\rho, U)$,
\begin{align} \lim_{t \to \infty} \prod_{i = 1}^d z^{(i)}_{j(i), l(i)}(t) = 0 \textrm{ for all}
(\vec{j}, \vec{l}) \in \operatorname{Supp} C(\rho, U).\label{eq:vanish_on_support} \end{align}
This follows from the calculation
\begin{align*}
\sum_{\vec{j} \in [n_1] \times \dots \times [n_d]} \left(\prod_{i = 1}^d \Delta p_{j(i)}^{(i)}\right) \tr \left(Y_{j(1)}^{(1)}(t)  \otimes \dots \otimes
 Y_{j(d)}^{(d)}(t) \right) (\nu_{\vec{j}}\cdot \rho)\\
 = \sum_{\vec{j} \in [n_1] \times \dots \times [n_d]} \left(\prod_{i = 1}^d \Delta p_{j(i)}^{(i)} \right)\left(\prod_{i = 1}^d z^{(i)}_{j(i), l(i)}(t)\right) C(\rho, U(t))_{\vec{j}, \vec{l}}.
\end{align*}
We have $C(\rho, U(t))_{\vec{j}, \vec{l}} \to C(\rho, U)_{\vec{j}, \vec{l}}$. If $C(\rho, U)_{\vec{j}, \vec{l}} > 0$, then $\lim_{t \to \infty} \left(\prod_{i = 1}^d z^{(i)}_{j(i), l(i)}(t)\right) < \delta$. Note that we could ignore the case when some $\Delta p_{j(i)}^{(i)} = 0$ because of our assumption that $z^{(i)}_{j,l}(t)$ tends to zero in that case.\\

This implies that subject to $\operatorname{Supp} D \subset \operatorname{Supp}C(\rho, U)$ there is no solution to \cref{eq:classical_margins}. Suppose there were. Note that $\cref{eq:classical_margins}$ and $\vec{p} \in P_+$ implies $\sum_{(\vec{j}, \vec{l}) \in S} D_{\vec{j}, \vec{l}} = 1$. Now
\begin{align*}
\log \left( \sum_{(\vec{j}, \vec{l}) \in S} \left( \prod_{i = 1}^d z_{j(i), l(i)}^{(i)}(t)\right) D_{\vec{j}, \vec{l}}  \right) &\geq \sum_{(\vec{j}, \vec{l}) \in S}  \left(\sum_{i = 1}^d \log z_{j(i), l(i)}^{(i)}(t)\right) D_{\vec{j}, \vec{l}} \\
&= \sum_{i = 1}^d \sum_{j \in [n_i]}  \Delta p_j^{(i)} \sum_{l \in [j]} \log z_{j(i), l(i)}^{(i)}(t) = 0.
\end{align*}
However, this contradicts our assumption that
 \begin{align*} \sum_{(\vec{j}, \vec{l}) \in S} \left( \prod_{i = 1}^d z_{j(i), l(i)}^{(i)}(t)\right) D_{\vec{j}, \vec{l}} = 0,
\end{align*}
which follows from \cref{eq:vanish_on_support}.

\subsubsection*{The ``only if'' direction}
We now prove the easier direction. Suppose that $(d, \underline{n})$-unitary family $U$ such that there is no $D$ satisfying $\operatorname{Supp}(D) \subset \operatorname{Supp} C(\rho, U)$ and \cref{eq:classical_margins}. We will show \cref{eq:sup_whys} holds.\\
By Farkas' lemma, there exists a sequence $\vec{a} = (a_{j,l}^{(i)}, i \in [d], j \in [n_i], l \in [j])$ of numbers such that
\begin{align} &\sum_{i = 1}^d a_{j(i), l(i)}^{(i)} \geq 0 \textrm{ for all } (\vec{j}, \vec{l}) \in \operatorname{Supp} C(\rho, U)\\
\textrm{ and }& \sum_{i = 1}^d \sum_{j = 1}^{n_i}\Delta p_j^{(i)}\sum_{i = 1}^l a_{j(i), l(i)} (i) < 0.
\end{align}
Set $\tilde{a}_{j,l}^{(i)} = a_{j,l}^{(i)} - \overline{a}$ where $\overline{a} = \sum_{i = 1}^d \sum_{j = 1}^{n_i} \Delta p_j^{(i)} \sum_{i = 1}^l a_{j(i), l(i)}^{ (i)}$. Now
\begin{align*}
&\sum_{i = 1}^d \tilde{a}_{j(i), l(i)} (i) > 0 \textrm{ for }\vec{j}, \vec{l} \in \operatorname{Supp}T(\rho, U)\\
\textrm{ and } & \sum_{i = 1}^d \sum_{j = 1}^{n_i} \Delta p_j(i) \sum_{i = 1}^l \tilde{a}_{j(i), l(i)} (i) = 0.
\end{align*}
Set
$$Y_j^{(i)}(t) = U_j^{(i), \dagger} \exp\left(- t \diag(\tilde{a}_{j, l}^{(i)})\right) U_j^{(i)}$$ and let $t$ tend to $\infty$. This shows \cref{eq:sup_whys} holds.\end{proof}

\end{document}